\documentclass[10pt,aps,prx,showpacs,showkeys,twocolumn]{revtex4-2}
\usepackage{amsmath, amssymb}
\usepackage{amsthm}
\usepackage{bm}
\usepackage{graphicx}
\usepackage{subcaption}
\usepackage{xcolor}

\newtheorem{thrm}{Theorem}
\newtheorem{lmn}{Lemma}

\DeclareMathOperator{\tr}{Tr}
\DeclareMathOperator{\re}{Re}
\DeclareMathOperator{\im}{Im}

\DeclareMathOperator{\str}{str}
\DeclareMathOperator{\rank}{rank}
\DeclareMathOperator{\Sp}{Sp}

\renewcommand{\vec}[1]{\mathbf{#1}}

\begin{document}
	
\title{Revisiting Gaussian genuine entanglement witnesses with modern software}
\author{E. Shchukin}
\email{shchukin@uni-mainz.de}
\author{P. van Loock}
\email{loock@uni-mainz.de}
\affiliation{Johannes-Gutenberg University of Mainz, Institute of Physics, Staudingerweg 7, 55128 Mainz, Germany}

\begin{abstract}
Continuous-variable Gaussian entanglement is an attractive notion, both as a
fundamental concept in quantum information theory, based on the well-established
Gaussian formalism for phase-space variables, and as a practical resource in
quantum technology, exploiting in particular, unconditional room-temperature
squeezed-light quantum optics. The readily available high level of scalability,
however, is accompanied by an increased theoretical complexity when the
multipartite entanglement of a growing number of optical modes is considered.
For such systems, we present several approaches to reconstruct the most probable
physical covariance matrix from a measured non-physical one and then test the
reconstructed matrix for different kinds of separability (factorizability,
concrete partite separability or biseparability) even in the presence of
measurement errors. All these approaches are based on formulating the desired
properties (physicality or separability) as convex optimization problems, which
can be efficiently solved with modern optimization solvers, even when the system
grows. To every optimization problem we construct the corresponding dual problem
used to verify the optimality of the solution. Besides this numerical part of
work, we derive an explicit analytical expression for the symplectic trace of a
positive definite matrix, which can serve as a simple witness of an entanglement
witness, and extend it for positive semidefinite matrices. In addition, we show
that in some cases our optimization problems can be solved analytically. As an
application of our analytical approach, we consider small instances of bound
entangled or genuine multipartite entangled Gaussian states, including some
examples from the literature that were treated only numerically, and a family of
non-Gaussian states.
\end{abstract}



\maketitle

\section{Introduction} 

The largest entangled system ever experimentally demonstrated is that of a fully
inseparable light-mode state for more than one million modes
\cite{10.1063/1.4962732}. This extraordinary level of scalability of
continuous-variable systems \cite{RevModPhys.77.513, Adesso_2007,
RevModPhys.84.621, serafini} relies upon the deterministic nature of
field-quadrature squeezing for photonic ``qumodes'', not available for photonic
qubits encoded in single photons and using standard quantum optics sources, and
on advanced methods such as time-domain multiplexing. While the photonic qumode
and qubit approaches share the benefit of room-temperature operations, unlike
typical solid-state platforms for quantum computing, with regards to a universal
set of quantum gates the main experimental hurdles are distinct: deterministic
non-Gaussian gates are hard to obtain for optical qumodes and deterministic
entangling gates are not readily available for photonic qubits. Since building
large continuous-variable cluster states, or multipartite multi-mode entangled
states \cite{Furusawa-Look, PhysRevA.109.040101}, can be done efficiently with
optical qumodes, it is of high interest to theoretically characterize and
experimentally determine the entanglement properties of such states ---
efficiently and reliably, even on a large scale. The Gaussian entanglement can
then be employed for the ultimate applications such as universal
measurement-based quantum computing
\cite{RevModPhys.77.513, Furusawa-Look, PhysRevA.109.040101, RevModPhys.84.621}
or for more near-term, sub-universal quantum advantage demonstrations
\cite{Nature.606.75}. In this work, we shall focus on the increased complexity
of Gaussian entangled states of up to almost 20 optical modes, exhibiting a rich
variety of (in)separability properties including the ultimate form of genuine
multipartite multi-mode entanglement. Supplementing such large-scale optical
systems with additional non-Gaussian elements may eventually lead to the
ultimate quantum applications \cite{PhysRevA.64.012310, PhysRevLett.123.200502}.
Non-Gaussian (entangled) states reach a high level of complexity even on a
smaller scale \cite{PRXQuantum.2.030204, 10.1088/1367-2630/ad6475}.

We present convex optimization formulations of two problems for multi-mode
Gaussian states --- how to extract the best possible physical covariance matrix
(CM) from the measured set of data and how to verify genuine multipartite
entanglement of a physical covariance matrix. Our work is not the first one to
apply methods of convex optimization for studying entanglement properties of
multipartite multi-mode states, but we go beyond just verifying entanglement. We
also compute the degree to which such a verification can be trusted in the
presence of imperfect measurement data. Taking into account the recent progress
in hardware and, more importantly, in software, it is the right time to
reconsider application of convex optimization to such problems. In contrast with
most other works where just a few states are tested for multipartite
entanglement, we can perform such tests ``on an industrial scale'' for a
sufficiently large number of parts. Our approach is not restricted to CMs only,
it can be extended to a more general problem of finding the most likely set of
data with special properties from a set of measurements which do not have these
properties. As long as the properties can be formulated in terms of convexity,
the ideas presented in this work can also be applied to such a problem as well.

For every problem developed in this work we present a detailed comparison of
three solvers --- Clarabel \cite{Clarabel}, Mosek and SCS \cite{SCS}. The
hardware is a high-end desktop PC with 24 cores and 128Gb of RAM. We do not
program to each solver's API individually, we use the JuMP \cite{JuMP} framework
instead. All solvers are used with the default settings, except Clarabel, where
we changed the default single-threaded linear solver to the multithreaded mkl
solver. Using non-default settings for parameters that influence the convergence
to the solution is mostly ``guess and try'' approach. We thus avoid guessing the
parameters and use the defaults. Note that thanks to the dual problem and the
Karush-Kuhn-Tucker (KKT) conditions we never have to guess whether the produced
solution is correct or not --- the dual solution is a witness of optimality and
accuracy. Irrespective of the parameters chosen for the execution of a solver,
we can always control the quality of the produced output. 

Our results will be illustrated by six figures for each problem. We show the
average time needed to solve a problem of a given size (averaged over 10 runs)
together with its standard deviation, the maximal amount of RAM required
\footnote{We obtain the maximal amount of RAM by reading the \texttt{VmHWM}
value from \texttt{/proc/PID/status}, where \texttt{PID} is the process ID of
the Julia script that solves the problem. Clearly, this approach works in Linux
only. We do not know, nor do we care about, how to obtain this value on other
platforms.} and three accuracy metrics (which are not averages, but the worst
cases of the ten runs). The first two metrics are the magnitudes ($\log_{10}$)
of the minimal eigenvalue of the primal and dual optimal solution, the other one
being the magnitude of the relative duality gap, which we define as the absolute
value of the difference between the optimal values of the primal and dual
objective functions divided by the optimal value of the primal function. It
turns out that the primal and dual solutions are boundary, so their minimal
eigenvalues are zero, the other metric characterizes the degree of optimality of
the solution. This choice of the metrics is motivated by the fact that boundary
conditions of the primal and dual solutions are not explicit in the problem
formulation. It is a consequence of one of the KKT conditions, but solvers do
not have to preserve it explicitly. These metrics are thus a good test for the
quality of the produced solution. For a perfect solution the minimal eigenvalues
and the duality gap must be zero, so all these metrics must be $-\infty$. Due to
limited precision of the machine arithmetic the numerical solutions are
imperfect and the metrics are finite. The more negative they are, the more
accurate the solution is. SCS solver is the least accurate for the problems
considered, but its truly remarkable feature is low RAM requirement by
comparison with the other two solvers, so we included it into our comparison. In
addition, its precision is likely to be improved by increasing the number of
iteration (which would also increase the time to find a solution), but that is
the topic of a more detailed comparison and is outside the scope of this work.
Whether the default accuracy is enough or not can be decided on a case-by-case
basis. 

The figures also demonstrate that the reduced case can be solved faster and with
less RAM than the full case, so our specialization of the problems for the full
and reduced cases separately is \textit{very valuable for practical applications
of states with a large number of modes.} We emphasize that the exact results
might depend on the test sets (on the same hardware), but the curves give a
rough estimate of the complexity to solve the problems of various sizes. For
Clarabel and Mosek we show the maximal problem that is possible to solve with
these solvers with 128 Gb of RAM. For SCS, which typically requires far less
RAM, the main limiting factor is the time to perform the tests.

A significant feature of our ``modern treatment'' is that we can deal with
Gaussian multipartite multi-mode states of a larger scale compared with existing
works and even treat them analytically in some cases. More specifically, we
shall treat systems of up to almost 20 optical modes, thus significantly going
beyond the existing schemes. With SCS we can test 17-partite full covariance
matrices and 18-partite reduced covariance matrices (when $xp$-part is zero).
This range of state sizes can be divided into three categories --- trivial,
simple and difficult. States with up to $\sim$5 modes are totally trivial and
can be solved almost instantly, for the number of modes from $\sim$6 to $\sim$10
the solver gives small but noticable delay and if the number of modes is larger
than $\sim$10 then the time required to perform the test rapidly increases.
These numbers give a very rough boundary between trivial, simple and difficult
sizes of states for a typical, non-high-end PC. For the next generations of
hardware and new versions of the optimization software these numbers will likely
increase, so boundary for difficult state sizes will increase and more states
become simple or trivial to solve. 

We explicitly specialize every optimization problem to this reduced case. This
specialization is insignificant for small values of the number of parts $n$ (let
us say, for $n \leqslant 8$), has observable effect for $n = 9, 10$, and
dramatically changes the situation for $n > 10$. The optimization problem in the
full case of $n=17$ has nearly $19$M variables (1M = $10^6$), and the problem in
the reduced case of $n=18$ has more than $22$M variables. Though the precision
of SCS is somewhat reduced by comparison with the other two solvers, by tweaking
the parameters of the solver it is likely possible to improve the precision (and
increase the time to solve the problem), but this should be done for each
problem individually.

Another observation is the RAM required to perform the genuine entanglement
tests. It turned out that in all cases (different problems, full and reduced
forms, various solvers) the RAM grows super-exponentially. This effect becomes
important only for the difficult category of state sizes, which currently means
that the number of modes is larger than $\sim$10.

Our work is a combination of both numerical and analytical results. In the first
part, we formulate all the optimization problems and show how computational
resources needed to solve them are scaled as functions of the numbers of parts.
In the second part, we demonstrate that in some cases the problems can be solved
exactly and provide analytical solutions. In more details, the manuscript is
structured as follows. In Sec.~\ref{sec:II} we give a non-exhaustive list of
existing works on the topic. In Sec.~\ref{sec:realIII} we give a short review of
the convex optimization theory. Sec.~\ref{sec:realIV} is devoted to physical
covariance matrix reconstruction. After a short introduction to physicality of
covariance matrices and establishing the notation, we formulate the problem of
reconstructing a physical CM from measured data as optimization problems. In
Sec.~\ref{sec:IIA} we use Chebyshev approximation to determine the closest
physical CM to a measured non-physical one. In Sec.~\ref{sec:IIB} we use an
alternative approach to find the most probable physical CM under the assumption
of a Gaussian distribution of the results of the measurements. In
Sec.~\ref{sec:III} we derive an analytical expression for the symplectic trace
of a symmetric, positive-defined matrix as an optimal objective value of some
trace minimization problem and construct its dual problem. In Sec.~\ref{sec:VI}
we review the symmetric states used later as a benchmark for our entanglement
detection procedures. In Sec.~\ref{sec:IV} we review the CM-based test for
entanglement and establish terminology we use later. In Sec.~\ref{sec:IVA},
using the dual problem for symplectic trace obtained in Sec.~\ref{sec:III}, we
derive the optimization problem and its dual for the best witness of a given
physical CM. Here we also construct a parametric family of $n$-partite states
and analytically show that these states are genuine multipartite entangled.
These states are then used for benchmarking of our programs. In addition, from
the optimal values of the optimization problems we construct two entanglement
measures. These measures are zero on separable states and only on them, and they
are convex, i.e. mixtures can only decrease these measures, but never increase.
These measures can be computed numerically and in some cases even analytically.
In Sec.~\ref{sec:IVB} we give an alternative approach, based on scaling, and
construct the corresponding optimization problem and its dual. These tests do
not take into account the inaccuracy of the measurements. In Sec.~\ref{sec:IVC}
we include the inaccuracy in the form of covariances of the measurements. The
formulations of the theorems are rather long and very detailed, with lots of KKT
conditions which might look excessive at first. But it is these conditions that
allow one to construct a system of equations for the optimal solution and obtain
an analytical solution for it. Without the KKT conditions it would be impossible
to construct such a system. In Sec.~\ref{sec:VIII} we analyze a rather new
condition obtained in Ref.~\cite{NewJPhys.20.023030} and show its performance
for testing the genuine multipartite entanglement relative to the standard,
``classical'' approach. For testing single partitions the performance advantage
can be huge, but for multiple partitions this advantage is more modest. In
Sec.~\ref{sec:VIIA} we present applications of the established results by
comparing them with some existing works. In Sec.~\ref{sec:V} B-D we construct
analytical solutions of two optimization problems for three different states ---
fixed bipartitions for a bound 4-partite entangled state and all bipartitions
for a 3- and 4-partite state. Our examples are either new or they are known from
the literature, but were only numerically treated before. Our solutions turn out
to be algebraic numbers and we explicitly present their minimal polynomials.
These solutions demonstrate highly non-trivial relations between some algebraic
numbers of various orders. Explicit expression of the solution as algebraic
numbers in terms of their minimal polynomials might not be the most practical
way to work with this solution, but it allows one to get the result exactly,
without relying on any approximation at all. In addition, it gives a very
detailed illustration of how the primal and dual solutions work in harmony to
deliver the common optimal value. The KKT conditions used to obtain analytical
results can thus be tested for a numerical result to verify the accuracy of the
numerical solution produced by the solver, which is especially important for
large-partite systems. In Sec.~\ref{sec:secPar} we construct an analytical
solution for a parametric family of states, not just for one concrete state. In
Sec.~\ref{sec:X} we give a detailed study of a family of tripartite symmetric
states and obtain analytical conditions of the various kinds of entanglement. We
present a method to find an analytical equation for the boundary of an
entanglement region even if the problem cannot be solved analytically. In
addition we show that the boundary of the genuine entanglement region consists
of two different curves. As an extra, more special example, we consider a family
of non-Gaussian states and test our CM-based formalism for such states. In
Sec.~\ref{sec:IX} we give some conjectures and open problems related to the
results presented in this work. In the conclusion we summarize our results.

\section{Existing works}\label{sec:II}

As our goal is to revisit the subject of Gaussian entanglement from a modern
perspective, let us first discuss the similarity and difference to some other
works on the topic. The semidefinite optimization approach to determine Gaussian
entanglement properties was introduced in Ref.~\cite{NewJPhys.8.51}. The
condition on the witness matrix that it is really a witness was given in an
appendix and referred to a manuscript in preparation. To our knowledge, that
manuscript has never been published. The KKT conditions that relate primal and
dual solutions of the optimization problem were not explicitly presented in
\cite{NewJPhys.8.51}. In some cases KKT conditions allow one to solve the
problem analytically. In fact, this is how we derived an explicit expression for
the symplectic trace and obtained explicit analytical solutions for a few 3- and
4-partite states. Where an analytical solution is either impossible or not
practical, the KKT conditions can be used to estimate the correctness and
quality of the numerical solution. The page with the software that the authors
of \cite{NewJPhys.8.51} used in their work does not seem to exist anymore, so it
is not immediately possible to use that software and test how well it performs
on modern hardware for larger problems. In the arXiv version of their work the
authors of Ref.~\cite{NewJPhys.8.51} give concrete numbers --- 2 seconds to test
for 3-partite entanglement and 40 minutes for 4-partite. It means that each new
part increases the computation time by more than one thousand. These numbers
have not been included in the published journal version which, together with the
absence of the software, makes it impossible to reproduce that result. However,
if the scalability of the computation time with the size of the states (number
of modes) was similar to those cases mentioned, that software appeared to be no
longer usable in practice for more than five parts. It gives a strong incentive
to reconsider the problem and explore the new possibilities of modern hardware
and software. Our analytical solution also shows the progress made in hardware
and software during the last 20 years --- what required a nontrivial amount of
time to solve numerically, now can be solved analytically in a second. 

To our knowledge, the term ``symplectic trace'', which is defined as the sum of
the symplectic eigenvalues, was introduced in Ref.~\cite{anders2006estimating}.
The symplectic eigenvalues were introduced in Ref.~\cite{GiedkeECP03}, see also
Refs.~\cite{Adesso_2007, S1230161214400010}. The Gaussian state represented by
the CM in Eq.~\eqref{eq:gammaM} below, which plays an important role in deducing
an expression for the symplectic trace through some minimization problem, was
mentioned in \cite{NewJPhys.20.023030} without relation to its minimization
properties. The expression of the symplectic trace as the minimal value of some
quantity, Eq.~\eqref{eq:minstr} below, was mentioned in
\cite{NewJPhys.25.113023} and referred back to \cite{anders2006estimating}. To
our knowledge, the explicit expression for the symplectic trace and its dual
characterization as a maximum of some quantity has not been known before. The
dual form allows one to construct an appropriate optimization problem, and the
explicit expression can be used to easily test the obtained solution. 

A treatment to incorporate error analysis into entanglement detection was
initiated in \cite{PhysRevLett.114.050501}. That approach was based on heuristic
genetic programming, which generally scales rather poorly and also delivers
suboptimal results. In that work the authors dealt with only fixed partitions
and our Fig.~\ref{fig:10} below clearly demonstrates the superiority of our new
approach. Other than that figure, we do not consider fixed partitions
numerically due to their computational simplicity and exclusively study genuine
multipartite entanglement. We deal with fixed bipartitions analytically when we
construct an exact solution for a bound entangled 4-partite state.

In Ref.~\cite{arXiv-2401.04376} several conditions for genuine multipartite
entanglement in terms of the covariance matrix (CM) are derived. Those
conditions work with subsets of the elements of CM, so they require less
measurements, but in general they are only sufficient for entanglement and thus
somewhat incomplete. The conditions are also based on optimization, but not
convex, which makes it harder to perform in practice. The authors of that work
give examples for states with up to $n=6$ parts. Because they say nothing about
how they performed the optimization, how much computational resources was needed
and how much time it took, we cannot compare the efficiency of their approach to
ours. With our method it takes just a fraction of a second to test 6-partite
states. The scalability of the entanglement verification (how much time and RAM
it takes) has not been discussed in the papers cited, and the accuracy of the
optimal solution has not been addressed as well.

\section{Convex optimization}\label{sec:realIII}

All problems in this work are instances of the following general convex
optimization problem:
\begin{alignat}{3}\label{eq:CO}
	\text{minimize} \quad &f_0(\vec{x}) & \nonumber \\
	\text{subject to} \quad &f_i(\vec{x}) \preccurlyeq_{\mathcal{K}_i} 0 \quad & i = 1, \ldots, m \\
	                        &h_j(\vec{x}) = 0 & j = 1, \ldots, p, \nonumber
\end{alignat}
where $f_i(\vec{x})$, $i = 0, \ldots, m$ are convex functions, $h_j(\vec{x})$,
$j = 1, \ldots, p$ are linear functions and $\mathcal{K}_i$ are convex cones of
dimensions $k_i$. The problems considered in this work use three types of cones
\begin{equation}\label{eq:cones}
\begin{split}
	\mathcal{K} &= \mathbf{R}^n_+ = \{(x_1, \ldots, x_n) |\ x_1 \geqslant 0, \ldots, x_n \geqslant 0 \} \\
	\mathcal{K} &= \mathbf{S}^n_+ = \{X \in \mathbf{R}^{n \times n} |\ X = X^{\mathrm{T}}, X \succcurlyeq 0\} \\
	\mathcal{K} &= \mathbf{K}^{n+1} = \{(x_1, \ldots, x_n, t) | \ \|\vec{x}\| \leqslant t\},
\end{split}
\end{equation} 
where the norm in the last cone is the standard Euclidean norm. The dimensions
of these cones are 
\begin{displaymath}
	\dim\mathbf{R}^n_+ = n, \ \ \dim\mathbf{S}^n_+ = \frac{n(n+1)}{2}, \ \ 
	\dim\mathbf{K}^{n+1} = n+1.
\end{displaymath}
The objective functions $f_0(\vec{x})$ are linear except the case of the most
probable CM, where $f_0(\vec{x})$ is quadratic.

To construct the dual problem to the problem \eqref{eq:CO}, we define the
Lagrangian function 
\begin{equation}
	\mathcal{L}(\vec{x}, \lambda, \bm{\nu}) = 
	f_0(\vec{x}) + \sum^m_{i=1} \bm{\lambda}^{\mathrm{T}}_i f_i(\vec{x})
	+ \sum^p_{j=1} \nu_j h_j(\vec{x}),
\end{equation}
where $\bm{\nu} = (\nu_1, \ldots, \nu_p)$ is a vector of numbers and $\lambda =
(\bm{\lambda}_1, \ldots, \bm{\lambda}_m)$ is a vector of vectors, where each
component is of the dimension of the corresponding cone $\mathcal{K}_i$. The
dual objective function is 
\begin{equation}
	g(\lambda, \bm{\nu}) = \inf_{\vec{x}} \mathcal{L}(\vec{x}, \lambda, \bm{\nu}).
\end{equation}
The dual problem then reads as
\begin{alignat}{3}\label{eq:DCO}
	\text{maximize} \quad &g(\lambda, \bm{\nu}) & \nonumber \\
	\text{subject to} \quad &\bm{\lambda}_i \succcurlyeq_{\mathcal{K}^*_i} 0 & \quad i = 1, \ldots, m \\
	                        &\text{finite constraints} & \nonumber
\end{alignat}
where $\mathcal{K}^*_i$ is the cone dual to $\mathcal{K}_i$ and finite
constraints are the conditions under which the dual objective $g(\lambda,
\bm{\nu})$ is finite. All the cones in \eqref{eq:cones} are self-dual, so no new
types of cones appear. The KKT conditions, which connect the primal
$\vec{x}^\star$ and dual $(\lambda^\star, \bm{\nu}^\star)$ optimal solutions,
are
\begin{align}
	h_j(\vec{x}^\star) &= 0 \\
	\bm{\lambda}^{\star\mathrm{T}}_i f_i(\vec{x}^\star) &= 0 \\
	\nabla f_0(\vec{x}^\star) + \sum^m_{i=1} D f_i(\vec{x}^\star)^{\mathrm{T}} \bm{\lambda}^\star_i +
	\sum^p_{j=1} \nu^\star_j \nabla h_j(\vec{x}^\star) &= 0, 
\end{align}
where the $k_i \times n$ matrix $D f_i(\vec{x}^\star)$ is the derivative of
$f_i(\vec{x})$ evaluated at $\vec{x} = \vec{x}^\star$. If there is a feasible
point $\vec{x}$, interior to the domain of all $f_i(\vec{x})$, $i = 1, \ldots,
m$, and such that the strong inequalities 
\begin{equation}
	f_i(\vec{x}) \prec_{\mathcal{K}_i} 0
\end{equation}
hold, then the strong duality between the problems \eqref{eq:CO} and
\eqref{eq:DCO} also holds,
\begin{equation}
	f_0(\vec{x}^\star) = g(\lambda^\star, \bm{\nu}^\star).
\end{equation}
The problems considered in this work, when feasible, have strong duality
property.

\section{Physical covariance matrix reconstruction}\label{sec:realIV}

The basic property that defines a physical $n$-partite CM $\gamma$ (of size $2n
\times 2n$) is the physicality condition
\begin{equation}\label{eq:PhysCond}
	\mathcal{P}(\gamma) = \gamma + \frac{i}{2}\Omega_n \succcurlyeq 0,
\end{equation}
where $\Omega_n$ is given by
\begin{equation}
	\Omega_n = 
	\begin{pmatrix}
		0 & E_n \\ 
		-E_n & 0 
	\end{pmatrix},
\end{equation}
and $E_n$ is the identity matrix of size $n$. We use the notation $A
\succcurlyeq 0$ to denote a positive semidefinite matrix $A$ and $A \geqslant 0$
to express that all elements of $A$ are nonnegative. Similar notation is used
for strict inequalities. From the physicality condition \eqref{eq:PhysCond} it
follows that any physical CM $\gamma$ is strictly positive-definite, $\gamma
\succ 0$, \cite{*[{}] [{, Theorem 98, p. 67.}] SMHA}.

We now establish a real form of physicality of a CM $\gamma$. We group all $x$
parts before all $p$ parts in covariance matrices (another notation in use is to
alternate $x$s and $p$s), so a general $n$-partite CM can be written in this
notation as
\begin{equation}\label{eq:matdec}
    \gamma = 
    \begin{pmatrix}
        \gamma_{xx} & \gamma_{xp} \\
        \gamma^{\mathrm{T}}_{xp} & \gamma_{pp}
    \end{pmatrix},
\end{equation}
where $\gamma_{xx}$, $\gamma_{pp}$ and $\gamma_{xp}$ are $n \times n$ blocks.
Let $H = S + i A$ be a Hermitian, where $S$ is a real symmetric matrix,
$S^{\mathrm{T}} = S$, and $A$ is a real antisymmetric matrix, $A^{\mathrm{T}} =
-A$. If $H$ is the matrix of a linear map in some basis $\{\vec{e}_1, \ldots,
\vec{e}_n\}$ over the field of complex numbers $\mathbb{C}$, then the same map
has the matrix
\begin{equation}\label{eq:tildeH}
	\tilde{H} = 
	\begin{pmatrix}
		S & \mp A \\
		\pm A & S
	\end{pmatrix}
\end{equation}
in the basis $\{\vec{e}_1, \ldots, \vec{e}_n, \pm i\vec{e}_1, \ldots, \pm
i\vec{e}_n\}$ over the field of real numbers $\mathbb{R}$ (where we always chose
either the upper sign or always the bottom sign). It follows that 
\begin{equation}
	H = S + i A \succcurlyeq 0 \quad \Leftrightarrow \quad \tilde{H} = 
	\begin{pmatrix}
		S & \mp A \\
		\pm A & S
	\end{pmatrix}\succcurlyeq 0.
\end{equation}
In particular, we have an equivalence 
\begin{equation}\label{eq:SA}
	S+iA \succcurlyeq 0 \quad \Leftrightarrow \quad S-iA \succcurlyeq 0.
\end{equation}
It also shows that it does not matter which sign to choose in
Eq.~\eqref{eq:tildeH}, so we choose the bottom one. Thus, the physicality
condition $\mathcal{P}(\gamma) \succcurlyeq 0$ of the covariance matrix
$\gamma$, which involves a complex Hermitian matrix, can be written in the real
form as
\begin{equation}\label{eq:physcond}
	\begin{pmatrix}
		\gamma & \frac{1}{2}\Omega \\
		-\frac{1}{2}\Omega & \gamma
	\end{pmatrix} = 
	\begin{pmatrix}
		\gamma_{xx} & \gamma_{xp} & 0 & \frac{1}{2}E \\
		\gamma^{\mathrm{T}}_{xp} & \gamma_{pp} & -\frac{1}{2}E & 0 \\
		0 & -\frac{1}{2}E & \gamma_{xx} & \gamma_{xp} \\
		\frac{1}{2}E & 0 & \gamma^{\mathrm{T}}_{xp} & \gamma_{pp}
	\end{pmatrix}
	\succcurlyeq 0.
\end{equation}
By rearranging the rows and columns of the latter matrix, which is achieved by
multiplying it with appropriate permutation matrices, we obtain the following
equivalent condition:
\begin{equation}
	\begin{pmatrix}
		\gamma_{xx} & \frac{1}{2}E & 0 & \gamma_{xp} \\
		\frac{1}{2}E & \gamma_{pp} & \gamma^{\mathrm{T}}_{xp} & 0 \\
		0 & \gamma_{xp} & \gamma_{xx} & -\frac{1}{2}E \\
		\gamma^{\mathrm{T}}_{xp} & 0 & -\frac{1}{2}E & \gamma_{pp}
	\end{pmatrix}
	\succcurlyeq 0.
\end{equation}
It follows that for the reduced states, i.e. the states with zero off-diagonal
part, $\gamma_{xp} = 0$, the condition \eqref{eq:physcond} is equivalent to two
simpler conditions,
\begin{equation}\label{eq:pcred}
	\mathcal{P}(\gamma_{xx}, \gamma_{pp}) = 
	\begin{pmatrix}
		\gamma_{xx} & \pm\frac{1}{2}E \\
		\pm\frac{1}{2}E & \gamma_{pp}
	\end{pmatrix}
	\succcurlyeq 0.
\end{equation}
As we have already seen above, the choice of the sign does not matter, so it is
enough to test just one of them. Below we use this condition with the sign $+$.
Additionally, if \eqref{eq:matdec} is a general, full CM, then the reduced CM
\begin{equation}\label{eq:matdecred}
	\gamma' = 
	\begin{pmatrix}
		\gamma_{xx} & 0 \\
		0 & \gamma_{pp}
	\end{pmatrix},
\end{equation}
obtained by setting $\gamma_{xp} = 0$ in $\gamma$, is also physical. It is this
observation that allows us to simplify the optimization problems developed in
later sections of this work in the cases where $\gamma_{xp}$ does not appear in
the objective functions. In such cases it is enough to optimize for
$\gamma_{xx}$ and $\gamma_{pp}$ only, which requires $n^2$ variables less for an
$n$-partite state. If the reduced $\gamma'$, obtained from a full physical CM
$\gamma$, were not always physical, we would have to optimize over the whole
matrices $\gamma$ in all cases. In the following, by a reduced state we mean a
pair of matrices $\gamma_{xx}$, $\gamma_{pp}$ that satisfies the condition
\eqref{eq:pcred}. \textit{In all the optimization problems presented below the
physicality conditions for all the CM variables $\gamma$, with or without sub-
or super-scripts, are assumed. We do not explicitly state them not to repeat them
all over again.}

An $n$-mode Gaussian state with CM $\gamma$ is pure iff
\begin{equation}\label{eq:purity}
	(\gamma \Omega_n)^2 = -\frac{1}{4}E_{2n}.
\end{equation}
For reduced states, $\gamma_{xp}=0$, this condition reads as follows:
\begin{equation}\label{eq:pureReduced}
	\gamma_{xx}\gamma_{pp} = \frac{1}{4}E_n.
\end{equation}
This condition will be used below to test states for purity.

A positive-semidefinite matrix $A \succcurlyeq 0$ is called boundary if it is on
the boundary of the positive-semidefinite cone, i.e. if the minimal eigenvalue
$\lambda_0(A)$ of $A$ is zero. A CM $\gamma$ is called boundary if
$\mathcal{P}(\gamma)$ is boundary. In the reduced case a pair $\gamma_{xx}$,
$\gamma_{pp}$ is boundary if $\mathcal{P}(\gamma_{xx}, \gamma_{pp})$ is
boundary. Since each physical CM is also positive-semidefinite, an ambiguity can
arise. Because we always distinguish physical CMs and arbitrary matrices, it is
always clear from the context which notion of boundary is used. Below we need
the following simple statement about arbitrary boundary matrices.
\begin{lmn}\label{lmn:AB} 
If $A$ and $B$ are Hermitian positive-semidefeinite
matrices, $A, B \succcurlyeq 0$, with $\tr(AB) = 0$, then $AB = 0$. In
addition, if $B \not= 0$, then $A$ is boundary.
\end{lmn}
\begin{proof}
The proof is based on a simple fact that a positive semidefinite matrix $X$ is
zero, $X = 0$, iff its trace is zero, $\tr X = 0$. Suppose that $\tr(AB) = 0$.
We cannot directly apply the fact above because $AB$ is, in general, not
Hermitian, but we can rewrite this equality as
\begin{equation}
	\tr(\sqrt{A}B\sqrt{A}) = 0,
\end{equation}
where $\sqrt{A}$ is a Hermitian square root of $A$. Square roots of matrices are
not unique, so there might be non-Hermitian square roots, but a Hermitian square
root of a positive-semidefinite Hermitian matrix always exists. The under trace
is Hermitian, so from the fact above we derive that
\begin{equation}
	\sqrt{A}B\sqrt{A} = 0.
\end{equation}
This equality can be rewritten as $C^\dagger C = 0$, where $C =
\sqrt{B}\sqrt{A}$. It follows that $C = 0$ and by multiplying by $\sqrt{B}$ on
the left and by $\sqrt{A}$ on the right, we obtain the equality $AB = 0$.

Now suppose that $B \not=0 $, but $A$ is not boundary. Then $A$ is invertible
and from the equality $AB=0$ we derive that $B=0$, which contradicts our
assumption. Thus $A$ must be boundary.
\end{proof}
\noindent As a consequence, if $\tr(AB) = 0$ and both $A$ and $B$ are non-zero,
then they are both boundary.

The reduced CM given by Eq.~\eqref{eq:matdecred} is boundary if 
\begin{displaymath}
	\det \mathcal{P}(\gamma') = 
	\det
	\begin{pmatrix}
		\gamma_{xx} & \frac{1}{2}E \\
		\frac{1}{2}E & \gamma_{pp}
	\end{pmatrix}
	\det
	\begin{pmatrix}
		\gamma_{xx} & -\frac{1}{2}E \\
		-\frac{1}{2}E & \gamma_{pp}
	\end{pmatrix}
	= 0.
\end{displaymath}
According to \cite{*[] [{ p. 27}] horn-johnson3} we have the equality
\begin{equation}
	\det
	\begin{pmatrix}
		\gamma_{xx} & \pm\frac{1}{2}E \\
		\pm\frac{1}{2}E & \gamma_{pp}
	\end{pmatrix}
	= \det\left(\gamma_{xx}\gamma_{pp} - \frac{1}{4}E\right),
\end{equation}
so the two matrices in Eq.~\eqref{eq:pcred} are simultaneously degenerate or
non-degenerate. It follows that a reduced state $\gamma_{xx}$, $\gamma_{pp}$ is
boundary if $\det \mathcal{P}(\gamma_{xx}, \gamma_{pp}) = 0$.

The experimentally measured covariance matrices $\gamma^\circ$ sometimes happen
to be non-physical, i.e. they violate the physicality condition
\eqref{eq:PhysCond}. If performing a more precise experiment is not an option,
then one needs to determine somehow the most likely physical covariance matrix
that agrees with the experimental data. In addition to the matrix $\gamma^\circ$
itself, the matrix $\sigma$ of standard deviations of the elements of
$\gamma^\circ$ is measured as well, so the true physical covariance matrix
$\gamma$ should not deviate too much from the measured $\gamma^\circ$, where the
distance is measured in units of $\sigma$. Below we present two approaches to
this problem.

Given a non-physical experimentally measured CM $\gamma^\circ$ with the matrix
of corresponding standard deviations $\sigma$ of measurements, what is the most
likely physical CM consistent with these data? Below we propose two approaches
to answer this question and test each approach on some randomly generated data.
The tests are performed on 10 randomly generated matrices for the number of
modes from 20 to 200 with step 10. The problems with sizes below 20 are too
simple to include into the figures. The covariance test matrices are randomly
generated with the elements in the range $[-1, 1]$ and the test matrices of
standard deviations are randomly generated with elements in the range $[0.001,
0.1]$. In general, these covariance matrices are far away from the set of
physical covariance matrices, but we do not have a better test suite.
Experiments with hundreds-partite states, determining and exploiting the full
CM, are not commonplace, although even larger-scale experiments have been
performed relying on a much smaller set of measurement data used as sufficient
criteria to witness full inseparability \cite{10.1063/1.4962732,
PhysRevA.67.052315}. These tests should be considered as a comparison of the
efficiency of different solvers on some randomly chosen problems. Realistic
data, which are typically much closer to the physical set should be easier to
solve, so the reported time is like to give an upper bound for the time it takes
to solve the corresponding problems in practice.

\subsection{Closest covariance matrix}\label{sec:IIA}

In \cite{PhysRevLett.117.140504} we presented a way to determine the closest
physical CM to a given non-physical one. Here we reproduce that optimization
problem in a more detailed form and specialize it to the reduced CM (where
$\gamma_{xp} = 0$). 
\begin{thrm}\label{thrm:1}
For any symmetric $2n \times 2n$ matrices $\gamma^\circ$ and $\sigma$, the
latter with strictly positive elements, the following equality takes place:
\begin{equation}\label{eq:covmatopt}
	\min_{\gamma} 
	\max_{1 \leqslant i, j \leqslant 2n} \frac{|\gamma_{ij} - \gamma^\circ_{ij}|}{\sigma_{ij}} = 
	\max_{\Lambda \succcurlyeq 0, U^\pm \geqslant 0} -\tr[\Lambda \mathcal{P}(\gamma^\circ)].
\end{equation}
The minimization is performed over all physical covariance matrices $\gamma$. The
maximization is over a Hermitian $2n\times2n$-matrix $\Lambda$ and a pair of
real symmetric $2n\times2n$-matrices $U^+$ and $U^-$ that satisfy the conditions
$\Lambda \succcurlyeq 0$, $U^\pm \geqslant 0$ and
\begin{equation}\label{eq:UUs}
	\tr[(U^+ + U^-)\sigma] = 1,  \quad U^+ - U^- = \re(\Lambda).
\end{equation}
The primal and dual optimal solutions are related by the equalities (KKT conditions)
$\mathcal{P}(\gamma^\star) \Lambda^\star = 0$ and
\begin{equation}
	U^{\pm\star} \cdot (\gamma^\star - \gamma^\circ \mp s^\star \sigma) = 0,
\end{equation}
where $s^\star$ is the common value of Eq.~\eqref{eq:covmatopt} and where $A
\cdot B$ is the Hadamard (element-wise) product of two matrices.

In the reduced case the previous result is specialized as follows. For any $n
\times n$ matrices $\gamma^\circ_{xx}$, $\gamma^\circ_{pp}$ and $\sigma_{xx}$,
$\sigma_{pp}$, the latter two with strictly positive elements, the dual relation
reads as
\begin{equation}\label{eq:covmatoptred}
\begin{split}
	\min_{\gamma_{xx}, \gamma_{pp}} 
	&\max_{1 \leqslant i,j \leqslant n} 
	\left[\frac{|\gamma_{xx, ij} - \gamma^\circ_{xx, ij}|}{\sigma_{xx, ij}}, 
	\frac{|\gamma_{pp, ij} - \gamma^\circ_{pp, ij}|}{\sigma_{pp, ij}}\right] \\
	&= \max_{\Lambda \succcurlyeq 0, U^\pm, V^\pm \geqslant 0} 
	-\tr[\Lambda \mathcal{P}(\gamma^\circ_{xx}, \gamma^\circ_{pp})].
\end{split}
\end{equation}
The minimization is over all reduced CMs. The maximization is over a real
symmetric $2n \times 2n$ matrix $\Lambda$ and real symmetric $n \times n$
matrices $U^{\pm}$ and $V^{\pm}$ subject to the conditions $\Lambda \succcurlyeq
0$, $U^{\pm}, V^{\pm} \geqslant 0$ and
\begin{equation}\label{eq:UUsxp}
\begin{split}
	&\tr[(U^+ + U^-)\sigma_{xx} + (V^+ + V^-)\sigma_{pp}] = 1, \\
	&\Lambda_{xx} = U^+ - U^-, \quad \Lambda_{pp} = V^+ - V^-.
\end{split}
\end{equation}
The KKT conditions read as $\mathcal{P}(\gamma^\star_{xx}, \gamma^\star_{pp})
\Lambda^\star = 0$ and
\begin{equation}
\begin{split}
	U^{\pm\star} \cdot (\gamma^\star_{xx} - \gamma^\circ_{xx} \mp s^\star \sigma_{xx}) &= 0, \\
	V^{\pm\star} \cdot (\gamma^\star_{pp} - \gamma^\circ_{pp} \mp s^\star \sigma_{pp}) &= 0,
\end{split}
\end{equation}
where $s^\star$ is the common value of Eq.~\eqref{eq:covmatoptred}. In both
cases, the optimal state $\gamma^\star$ or $\gamma^\star_{xx}$,
$\gamma^\star_{pp}$ is a boundary state.
\end{thrm}

\begin{figure}
	\includegraphics[scale=1.0]{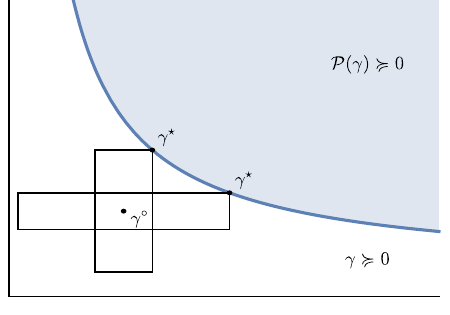}
\caption{The optimal solution of \eqref{eq:covmatopt} for different $\sigma$.}\label{fig:Capprox}
\end{figure}

\begin{figure*}
	\includegraphics[scale=0.73]{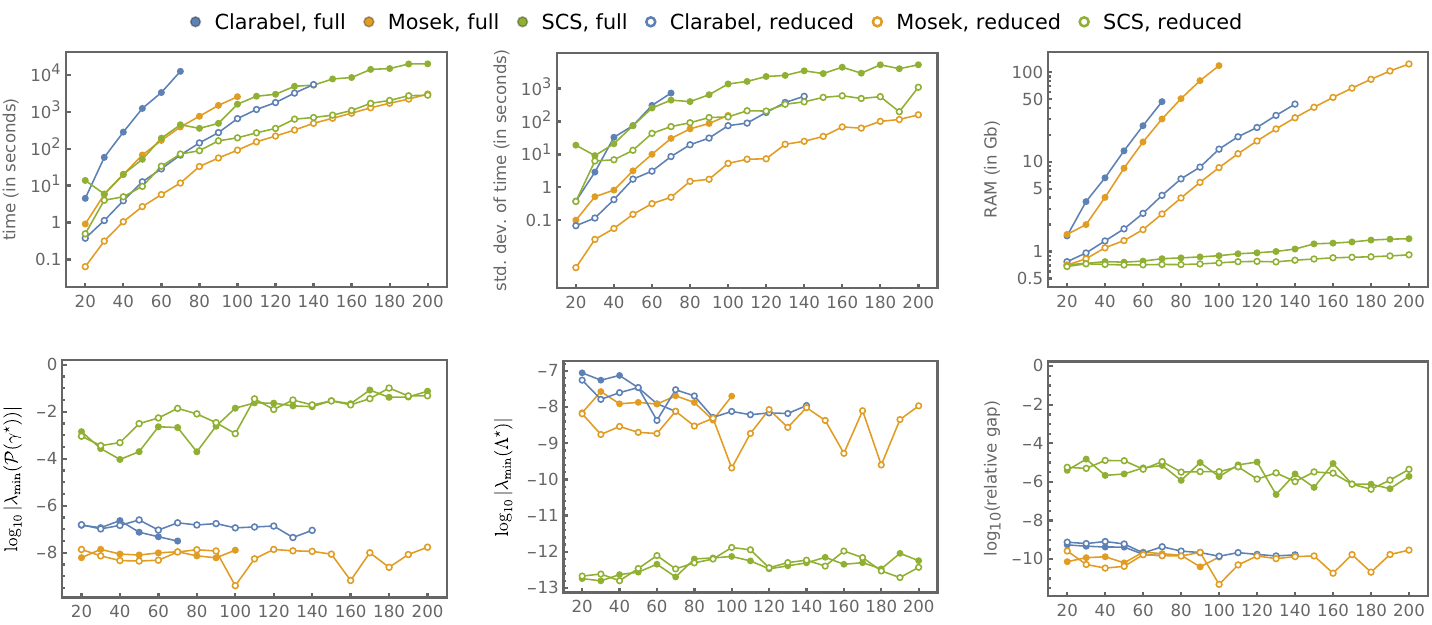}
	\caption{Comparing performance and accuracy characteristics of different solvers for the 
	closest physical covariance matrix problem as a function of the number of modes.}
	\label{fig:closestpcv}
\end{figure*}

The number of variables in the primal optimization problem is
\begin{equation}
	N_v = 
	\begin{cases}
		n(2n+1)+1 & \text{full case} \\
		n(n+1)+1 & \text{reduced case}
	\end{cases}
\end{equation}
In the reduced case the problem has $n^2$ variables less, so for large $n$ it is
almost half of the size of the full case. Such problems, minimization of the
maximal difference, are called Chebyshev approximation problems \cite{*[{}] [{,
Sect. \textbf{1.2.2}, p. 6.}] convex-opt}. If we represent a symmetric matrix as
a vector of its different components and denote by $s$ the objective function of
the primal problem of Eq.~\eqref{eq:covmatopt}, then
\begin{equation}\label{eq:snorminf}
	s = \left\|\frac{\gamma - \gamma^\circ}{\sigma}\right\|_{\infty},
\end{equation}
where the division is performed elementwise, so $\gamma^\star$ is the physical
covariance matrix closest to $\gamma^\circ$ in the $\infty$-norm. The unit
sphere of the $\infty$-norm is a unit square, and the scaling by $\sigma$ turns
square into a general rectangle. The optimization \eqref{eq:covmatopt} then
seeks for the minimal factor $s^\star$ by which we need to expand or shrink the
``unit rectangle'' so that it touches the boundary of the set of the physical
covariance matrices. The optimal value $s^\star$ gives the minimal $s^\star
\sigma$-neighborhood of $\gamma^\circ$ which contains a physical covariance
matrix, $\gamma^\star$. A couple of such minimal neighborhoods for different
$\sigma$ are shown on Fig.~\ref{fig:Capprox}, where the positive orthant denotes
the set of positive-semidefinite matrices, the blue region denotes the set of
physical CMs and the rectangles are the sets of $\gamma$ defined by
$\|(\gamma-\gamma^\circ)/\sigma\| \leqslant s$ for two $\sigma$'s. Put another
way, for any physical $\gamma$ at least one component $\gamma_{ij}$ is $s^\star$
or more $\sigma_{ij}$-intervals away from $\gamma^\circ_{ij}$. The smaller the
optimal $s^\star$, the higher the probability that the experiment was performed
correctly and the optimal solution $\gamma^\star$ is close to the true
covariance matrix of the state under study. If $\gamma^\circ$ is physical from
the very beginning, then $\gamma^\star = \gamma^\circ$ and $s^\star = 0$.
According to the standards of reliability \cite{enc-meas-stat, standards}, if
$s^\star > 2$ then the experiment should be considered very inaccurate. If
$s^\star > 3$ then the measured data are erroneous.

The objective of the dual problem is $-\tr[\Lambda \mathcal{P}(\gamma^\circ)]$.
If $\gamma^\circ$ is physical, then $\mathcal{P}(\gamma^\circ) \succcurlyeq 0$
and $-\tr[\Lambda \mathcal{P}(\gamma^\circ)] \leqslant 0$ for all $\Lambda
\succcurlyeq 0$, so the minimum of the primal objective is zero, and the
corresponding dual optimal solution is $\Lambda^\star = 0$ and $U^{+\star} =
U^{-\star}$, appropriately normalized so that the first condition of
Eq.~\eqref{eq:UUs} is satisfied. If $\gamma^\circ$ is non-physical, then
$\mathcal{P}(\gamma^\circ)$ has negative eigenvalues, which can be amplified by
a positive semidefinite $\Lambda^\star$. Exactly how much the negative
eigenvalues can be amplified is determined by the appropriate conditions on the
dual variables. In addition, $\gamma^\star$ and $\Lambda^\star$ are boundary. In
fact, $\gamma^\star$ is physical by construction and non-zero (its imaginary
part is $\Omega/2$), and $\Lambda^\star$ is non-zero since $\gamma^\circ$ is
non-physical and thus $\gamma^\star \not= \gamma^\circ$. As a consequence of
Lemma~\ref{lmn:AB} it follows that both $\gamma^\star$ and $\Lambda^\star$ are
boundary. The minimal eigenvalues of $\mathcal{P}(\gamma^\star)$ (or
$\mathcal{P}(\gamma^\star_{xx}, \gamma^\star_{pp})$ in the reduced case) and
$\Lambda^\star$ can be used as metrics of the accuracy of the solution produced
by a convex optimization solver. These metrics together with the relative gap
and performance characteristics are shown in Fig.~\ref{fig:closestpcv} for the
three solvers used in this work. From this figure one can conclude that Mosek is
likely to be the most appropriate choice for this problem and the data used as
tests. We see that true CM of a 100-partite state can be obtained in a
reasonable amount of time on a moderate desktop hardware.

\subsection{The most probable covariance matrix}\label{sec:IIB}

In the approach above we made no assumption about the exact form of the
probability distributions of the elements of $\gamma$ except that they have
averages $\gamma^\star$, standard deviations $\sigma$ and have bell-shape so
that the $\sigma$-intervals have meaning close to that of Gaussian
distributions. If we make a stronger assumption about the probability
distributions of the individual elements of $\gamma$, we can define the most
probable physical covariance matrix as the matrix that maximizes the probability
distribution. We assume now that each measured matrix element has perfect
Gaussian distributions with the average $\gamma^\circ_{ij}$ and standard
deviation $\sigma_{ij}$:
\begin{equation}\label{eq:gd}
	p(\gamma_{ij}) = \frac{1}{\sqrt{2\pi}\sigma_{ij}} 
	\exp\left[-\frac{(\gamma_{ij}-\gamma^\circ_{ij})^2}{2\sigma^2_{ij}}\right].
\end{equation}
The perfectly Gaussian distribution is non-physical for diagonal elements (they
cannot be negative), but once the measured covariance matrix is non-physical, we
already work with such objects. If the experiment is not totally inaccurate, the
ratios $\gamma^\circ_{ii}/\sigma_{ii}$ should be large enough, so the
probability of $\gamma_{ii}$ being negative under the distribution \eqref{eq:gd}
is exponentially small.

Assuming that the measurements of the matrix elements are done independently,
the probability density of $\gamma$ being the true physical matrix reads as the
product 
\begin{equation}\label{eq:pgamma}
	p(\gamma) = \prod_{1\leqslant i \leqslant j \leqslant 2n} p(\gamma_{ij}).
\end{equation}
In the reduced case the product is over the elements of $\gamma_{xx}$ and
$\gamma_{pp}$. Now we can state that the most likely physical matrix
$\gamma^\star$ is the one that maximizes the probability distribution
$p(\gamma)$. Expanding the product, we see that we need to minimize the
quadratic form
\begin{equation}\label{eq:Qop}
	\sum_{1\leqslant i \leqslant j \leqslant 2n} 
	\frac{(\gamma_{ij}-\gamma^\circ_{ij})^2}{2\sigma^2_{ij}}.
\end{equation}
This is a quadratic form in the variables $\gamma_{ij}$ and we need to minimize
this form subject to the semidefinite constraint $\mathcal{P}(\gamma)
\succcurlyeq 0$. Some components of the optimal solution might be far away from
their measured averages (but others might be very close), so this approach cannot
be considered completely satisfactory. A more reliable approach is to constrain
all components $\gamma_{ij}$ inside a reasonable interval around the average, 
\begin{equation}\label{eq:s0}
	|\gamma_{ij} - \gamma^\circ_{ij}| \leqslant s_0 \sigma_{ij},
\end{equation}
where, for example, $s_0 = 2$. Then the algorithm to find the most probable
physical covariance matrix reads as: We first solve the problem
\eqref{eq:covmatopt} and find $s^\star$. If $s^\star \geqslant s_0$ then the
experiment is considered inaccurate and further processing stops. If $s^\star <
s_0$, then the most likely physical covariance matrix is the solution that
minimizes the quadratic form given by Eq.~\eqref{eq:Qop}, subject to the
physicality condition and the conditions expressed by Eq.~\eqref{eq:s0}. At
least one component $\gamma^\star_{ij}$ of the solution of this problem will
have the property $s^\star \leqslant |\gamma^\star_{ij} - \gamma^\circ_{ij}|
\leqslant s_0$. It is this solution $\gamma^\star$ that we use in the
applications to entanglement verification below. The dual problem is given by
the following theorem.

\begin{thrm}\label{thrm:2}
For any symmetric $2n \times 2n$ matrices $\gamma^\circ$ and $\sigma$, the
latter with strictly positive elements, the following equality takes place:
\begin{equation}
\begin{split}
	\min_{\mathcal{P}(\gamma) \succcurlyeq 0} &\sum_{1 \leqslant i \leqslant j \leqslant 2n} 
	\frac{(\gamma_{ij}-\gamma^\circ_{ij})^2}{2\sigma^2_{ij}} \\
	&= \max_{\Lambda \succcurlyeq 0} -\tr\left[\Lambda \mathcal{P}(\gamma^\circ) + 
	\frac{1}{2}M^{\prime\cdot 2} \sigma^{\cdot 2}\right],
\end{split}
\end{equation}
where the maximization on the right-hand side is over all Hermitian positive
semidefinite matrices $\Lambda \succcurlyeq 0$ and $M'_{ij}$ are the elements
of $M = \re(\Lambda)$ with non-diagonal elements multiplied by $\sqrt{2}$ so that
\begin{equation}
	M'_{ij} = 
	\begin{cases}
		M_{ij} & i = j \\
		\sqrt{2}M_{ij} & i \not= j
	\end{cases}.
\end{equation}
Note that $M^{\prime\cdot 2}$ and $\sigma^{\cdot 2}$ are Hadamard square, but
their product is an ordinary matrix product. The KKT conditions are given by
$\tr[\Lambda^\star \mathcal{P}(\gamma^\star)] = 0$ and
\begin{equation}
	\gamma^\star = \gamma^\circ + M^{\star\prime\prime} \cdot \sigma^{\cdot 2},
\end{equation}
where $M'$ is matrix with elements $M'_{ij}$ (note that the prime is applied
twice and the product between $M^{\star\prime\prime}$ and $\sigma^{\cdot 2}$ is
Hadamard). 

\begin{figure*}
	\includegraphics[scale=0.73]{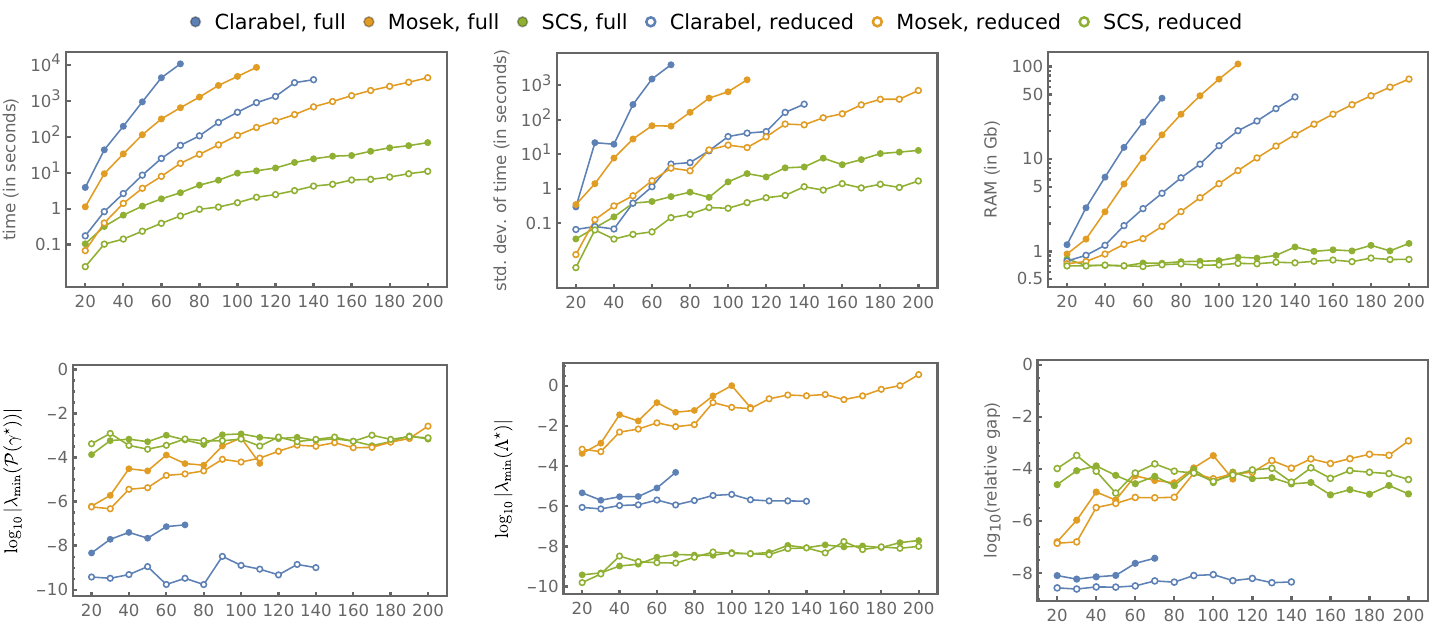}
	\caption{Comparing performance and accuracy characteristics of different solvers for the 
	most probable physical covariance matrix problem (with $s_0=+\infty$) as a function of the number of modes.}
	\label{fig:mostprobablepcv}
\end{figure*}

With the additional constraint \eqref{eq:s0}, where $s_0 > s^\star$ is larger
then the optimal solution of the problem \eqref{eq:covmatopt}, the following
primal and dual problems are feasible and their duality reads as:
\begin{equation}
\begin{split}
	&\min_{\substack{\mathcal{P}(\gamma) \succcurlyeq 0 \\ |\gamma - \gamma^\circ| \leqslant s_0 \sigma}} \sum_{i \leqslant j} 
	\frac{(\gamma_{ij}-\gamma^\circ_{ij})^2}{2\sigma^2_{ij}} = 
	\max_{\Lambda, U^\pm} -\tr\Biggl[\Lambda \mathcal{P}(\gamma^\circ) \Biggr.\\
	&\Biggl. +\frac{1}{2} (M - U^+ + U^-)^{\prime\cdot 2} \sigma^{\cdot 2} 
	+ s_0 (U^+ + U^-)\sigma\Biggr],
\end{split}
\end{equation}
where the maximization is over all Hermitian matrices $\Lambda \succcurlyeq 0$
and a pair of real symmetric matrices $U^\pm \geqslant 0$ with non-negative
elements. The KKT conditions are $\mathcal{P}(\gamma^\star) \Lambda^\star = 0$ and
\begin{equation}
\begin{split}
	&\gamma^\star = \gamma^\circ + (M^\star - U^{+\star} + U^{-\star})^{\prime\prime} \cdot \sigma^{\cdot 2} \\
	&U^{\pm \star} (\gamma^\star - \gamma^\circ \mp s_0 \sigma) = 0.
\end{split}
\end{equation}
In the reduced case the general problem above for $s_0 = +\infty$ is specialized
as follows:
\begin{equation}
\begin{split}
	&\min_{\mathcal{P}(\gamma_{xx}, \gamma_{pp}) \succcurlyeq 0} \sum_{i \leqslant j} 
	\left[\frac{(\gamma_{xx,ij}-\gamma^\circ_{xx,ij})^2}{2\sigma^2_{xx,ij}} + 
	\frac{(\gamma_{pp,ij}-\gamma^\circ_{pp,ij})^2}{2\sigma^2_{pp,ij}}\right] \\
	&= \max_{\Lambda \succcurlyeq 0} \left(-\tr[\Lambda \mathcal{P}(\gamma^\circ_{xx}, \gamma^\circ_{pp})]
	-\frac{1}{2}\tr[\Lambda^{\prime\cdot 2}_{xx} \sigma^{\cdot 2}_{xx} 
	+\Lambda^{\prime\cdot 2}_{pp} \sigma^{\cdot 2}_{pp}] \right),
\end{split}
\end{equation}
where the maximization is over all real $\Lambda \succcurlyeq 0$. The KKT
conditions are $\mathcal{P}(\gamma^\star_{xx}, \gamma^\star_{pp}) \Lambda^\star
= 0$ and  
\begin{equation}
	\gamma^\star_{xx} = \gamma^\circ_{xx} + \Lambda^{\star\prime\prime}_{xx} \cdot \sigma^{\cdot 2}_{xx} \quad
	\gamma^\star_{pp} = \gamma^\circ_{pp} + \Lambda^{\star\prime\prime}_{pp} \cdot \sigma^{\cdot 2}_{pp}.
\end{equation}
For $s_0 < +\infty$ several linear constrains are added in addition to the
semidefinite constraint. The problem reads as follows, where the maximization is
over all real $\Lambda \succcurlyeq 0$ and $U^\pm, V^\pm \geqslant 0$:
\begin{widetext}
\begin{equation}
\begin{split}
	&\min_{\gamma_{xx}, \gamma_{pp}} \sum_{i \leqslant j} 
	\left[\frac{(\gamma_{xx,ij}-\gamma^\circ_{xx,ij})^2}{2\sigma^2_{xx,ij}} + 
	\frac{(\gamma_{pp,ij}-\gamma^\circ_{pp,ij})^2}{2\sigma^2_{pp,ij}}\right] 
	= \max_{\Lambda \succcurlyeq 0, U^\pm, V^\pm \geqslant 0} \Biggl(-\tr[\Lambda \mathcal{P}(\gamma^\circ_{xx}, \gamma^\circ_{pp})]\Biggr. \\
	& -\tr\left[\frac{1}{2}(\Lambda_{xx} - U^+ + U^-)^{\prime\cdot 2} \sigma^{\cdot 2}_{xx} 
	+\frac{1}{2}(\Lambda_{pp} - V^+ + V^-)^{\prime\cdot 2} \sigma^{\cdot 2}_{pp} 
	+s_0 (U^+ + U^-)\sigma_{xx} + s_0 (V^+ + V^-)\sigma_{pp}\right]\Biggr),
\end{split}
\end{equation}
\end{widetext}
with the addition constraints of the primal problem (on the left-hand side) being
\begin{equation}
	|\gamma_{xx} - \gamma^\circ_{xx}| \leqslant s_0 \sigma_{xx}, \quad
	|\gamma_{pp} - \gamma^\circ_{pp}| \leqslant s_0 \sigma_{pp}.
\end{equation}
The KKT conditions are $\mathcal{P}(\gamma^\star_{xx}, \gamma^\star_{pp})
\Lambda^\star = 0$,
\begin{equation}
\begin{split}
	\gamma^\star_{xx} &= \gamma^\circ_{xx} + (\Lambda^\star_{xx} - U^{+\star} + U^{-\star})^{\prime\prime}
	\cdot \sigma^{\cdot 2}_{xx} \\
	\gamma^\star_{pp} &= \gamma^\circ_{pp} + (\Lambda^{\star}_{pp} - V^{+\star} + V^{-\star})^{\prime\prime}
	\cdot \sigma^{\cdot 2}_{pp},
\end{split}
\end{equation}
and
\begin{equation}
\begin{split}
	U^{\pm\star} \cdot (\gamma^\star_{xx} - \gamma^\circ_{xx} \mp s_0 \sigma_{xx}) &= 0 \\
	V^{\pm\star} \cdot (\gamma^\star_{pp} - \gamma^\circ_{pp} \mp s_0 \sigma_{pp}) &= 0.
\end{split}
\end{equation}
If the original CM $\gamma^\circ$ (or the pair $\gamma^\circ_{xx}$ and
$\gamma^\circ_{pp}$) is non-physical, then the optimal, most probable CM
$\gamma^\star$ is boundary.
\end{thrm}

The objectives of the primal problems are positive quadratic forms of primal
variables, thus bounded from below and the minimum of the forms is finite. The
objectives of dual problems are negative quadratic forms of dual variables, thus
bounded from above and the maximum is finite. The duality states that the
minimum of primal quadratic form equals the maximum of the dual quadratic form.
As in the previous case, if $\gamma^\circ$ is non-physical, then both
$\gamma^\star$ and $\Lambda^\star$ are boundary and their minimal eigenvalues
can be used as metrics of the accuracy of the numerical solution. In
Fig.~\ref{fig:mostprobablepcv} we show these metrics together with performance
characteristics of the three solvers. It seems that in this case Clarabel
produces more accurate solutions (in two out of three metrics) than the others.

We now have three approaches to determine the most likely physical CM that is
consistent with the experimental data $\gamma^\circ$ and $\sigma$ --- the
closest CM, the most probable CM without restriction on its elements and the
most probable CM with the restriction given by Eq.~\eqref{eq:s0}. The difference
between the three approaches can be described as follows. The closest CM gives
the best worst-case solution, the one where the component furthest away from the
average is as close to the average as possible for a physical CM. The most
probable CM without restriction on its elements is the best average-case, where
the average ``quality'' of all components is the best possible, but some
components might be of very low ``quality'' (far away from the average).
Finally, the most probable CM with the restriction given by Eq.~\eqref{eq:s0} is
the best average-case with bounded worst-case, where no component can be very
bad. We cannot distinguish one of these approaches as the ``right'' one based on
probability theory only, though the last approach might be more ``correct'' than
the other two (its solution maximizes the probability distribution of the
covariance matrix of which all components are inside a reasonable
$\sigma$-interval around the average). But the ultimate way to have a true
physical covariance matrix is to perform the measurements with higher accuracy.

\section{Symplectic trace}\label{sec:III}

Here we give an analytical expression for the symplectic trace of an arbitrary
positive definite matrix $M \succ 0$ by showing that it is the minimal value of
some optimization problem. In addition, we construct the dual problem, which
will be used below to implement a test for genuine multipartite entanglement.
The theorem that states these results is:
\begin{thrm}
For any $M \succcurlyeq 0$ the following relation is valid:
\begin{equation}\label{eq:Mstr}
\begin{split}
	\min_{\gamma} \tr(M\gamma) &= \max_{M + i K \succcurlyeq 0} \frac{1}{2} \tr (K\Omega) \\
	&= \frac{1}{2} \tr \sqrt{\sqrt{M}\Omega^{\mathrm{T}}M\Omega\sqrt{M}},
\end{split}
\end{equation}
where the optimization variable on the right-hand side is an anti-symmetric
matrix $K$, $K^{\mathrm{T}} = -K$. The KKT condition is given by
\begin{equation}\label{eq:strKKT}
	\mathcal{P}(\gamma^\star)(M + i K^\star) = 0.
\end{equation}
For $M \succ 0$ the optimal value \eqref{eq:Mstr} is equal to the symplectic
trace of $M$ and the optimal covariance matrix is unique and explicitly reads as
\begin{equation}\label{eq:gammaM}
	\gamma_M = \frac{1}{2} \sqrt{M^{-1}} 
	\sqrt{\sqrt{M}\Omega^{\mathrm{T}}M\Omega\sqrt{M}} \sqrt{M^{-1}}.
\end{equation}
In the reduced case we have
\begin{equation}
\begin{split}
	\min_{\gamma_{xx}, \gamma_{pp}} &\tr(X\gamma_{xx} + P\gamma_{pp}) = \max_{Z} \tr Z\\
	&= \tr \sqrt{\sqrt{X}P\sqrt{X}} = \tr \sqrt{\sqrt{P}X\sqrt{P}},
\end{split}
\end{equation}
where the optimization is over all matrices $Z$ such that
\begin{equation}
	\begin{pmatrix}
		X & Z \\
		Z^{\mathrm{T}} & P
	\end{pmatrix}
	\succcurlyeq 0.
\end{equation}
The KKT condition is given by
\begin{equation}\label{eq:KKTreduced}
	\mathcal{P}(\gamma^\star_{xx}, \gamma^\star_{pp})
	\begin{pmatrix}
		X & -Z^{\star} \\
		-Z^{\star\mathrm{T}} & P
	\end{pmatrix}
	= 0.
\end{equation}
For $X, P \succ 0$ the optimal solution is unique and reads as
\begin{equation}\label{eq:gammaXP}
\begin{split}
	\gamma^\star_{xx} &= \frac{1}{2}\sqrt{X^{-1}}\sqrt{\sqrt{X}P\sqrt{X}}\sqrt{X^{-1}}, \\
	\gamma^\star_{pp} &= \frac{1}{2}\sqrt{P^{-1}}\sqrt{\sqrt{P}X\sqrt{P}}\sqrt{P^{-1}},
\end{split}
\end{equation}
and the optimal value is equal to the symplectic trace of the following reduced
matrix:
\begin{equation}\label{eq:M-XP}
	M = 
	\begin{pmatrix}
		X & 0 \\
		0 & P
	\end{pmatrix}.
\end{equation}
The optimal states are pure and on the boundary of the physical states.
\end{thrm}
\begin{proof}
The Lagrangian of the primal problem \eqref{eq:Mstr} reads as
\begin{equation}
\begin{split}
	\mathcal{L}(\gamma, \Lambda) &= \tr(M\gamma) - \tr[\mathcal{P}(\gamma)\Lambda] \\
					&= \tr[(M-\re(\Lambda))\gamma] + \frac{1}{2}\tr(\im(\Lambda)\Omega),
\end{split}
\end{equation}
where $\Lambda \succcurlyeq 0$. The dual objective is finite if $\re(\Lambda) =
M$ and is given by
\begin{equation}
	g(\Lambda) = \inf_{\gamma} \mathcal{L}(\gamma, \Lambda) = \frac{1}{2}\tr(\im(\Lambda)\Omega).
\end{equation}
Denoting $K = \im(\Lambda)$, we obtain the desired dual optimization problem and
the KKT condition \eqref{eq:strKKT}.

Let us first assume that $M \succ 0$ and derive an analytical expression for the
solution $\gamma^\star$ of the primal optimization problem \eqref{eq:Mstr}. The
primal and dual optimal solutions $\gamma^\star$ and $K^\star$ satisfy the KKT
condition 
\begin{equation}
	\tr[\mathcal{P}(\gamma^\star)(M + i K^\star)] = 0,
\end{equation}
and from Lemma~\ref{lmn:AB} the equality \eqref{eq:strKKT} follows. Explicitly
it reads as
\begin{equation}
    \left(\gamma^\star + \frac{i}{2}\Omega\right)(M+iK^\star) = 0.
\end{equation}
Separating the real and imaginary parts, we derive two equalities
\begin{equation}\label{eq:MgK}
    \gamma^\star M = \frac{1}{2}\Omega K^\star, \quad \frac{1}{2}\Omega M + \gamma^\star K^\star = 0.
\end{equation}
By multiplying both sides of the first equality by $\gamma^\star \Omega$ on the
left, we obtain
\begin{equation}
	\gamma^\star\Omega\gamma^\star M = -\frac{1}{2}\gamma^\star K^\star = \frac{1}{4}\Omega M.
\end{equation}
Since $M$ is non-degenerate, we can cancel it from both sides and arrive to the equality
\begin{equation}\label{eq:gOg}
    \gamma^\star \Omega \gamma^\star = \frac{1}{4}\Omega,
\end{equation}
which is equivalent to the condition \eqref{eq:purity}. It means that any
optimal solution of the primal problem is pure.

Now we derive an explicit expression for $\gamma^\star$. The second equality of
Eq.~\eqref{eq:MgK} is equivalent to
\begin{equation}
    K^\star \gamma^\star = \frac{1}{2} M \Omega^{\mathrm{T}}.
\end{equation}
Multiplying the first equality of Eq.~\eqref{eq:MgK} by $\gamma^\star$ on the
right, we obtain
\begin{equation}\label{eq:XAXB}
    \gamma^\star M \gamma^\star = \frac{1}{2} \Omega K^\star \gamma^\star
    = \frac{1}{4} \Omega M \Omega^{\mathrm{T}}
	= \frac{1}{4} \Omega^{\mathrm{T}} M \Omega.
\end{equation}
According to \cite{*[] [{ p. 445}] horn-johnson}, the unique solution of this
equation is given by
\begin{equation}\label{eq:gammastar}
	\gamma^\star = \frac{1}{2} \sqrt{M^{-1}} \sqrt{\sqrt{M}\Omega^{\mathrm{T}}M\Omega\sqrt{M}} \sqrt{M^{-1}}.
\end{equation}
Such states were considered in \cite{NewJPhys.20.023030}, but without studying
their minimizing property. For the minimum of the trace $\tr(M\gamma)$ we thus
have
\begin{displaymath}
    \min_{\gamma} \tr(M\gamma) = \tr(M\gamma^\star)
    = \frac{1}{2} \tr{\sqrt{\sqrt{M}\Omega^{\mathrm{T}}M\Omega\sqrt{M}}}. 
\end{displaymath}

We show that this (obviously symmetric) matrix satisfies the physicality
condition. In fact, we have
\begin{equation}
	\mathcal{P}(\gamma^\star) = \frac{1}{2}\sqrt{M^{-1}}
	\left[\sqrt{A^{\mathrm{T}} A} + i A\right]\sqrt{M^{-1}},
\end{equation}
where 
\begin{equation}
	A = \sqrt{M}\Omega\sqrt{M}.
\end{equation}
Since $A A^{\mathrm{T}} = A^{\mathrm{T}} A = -A^2$, according to \cite{*[] [{ p. 504}]
horn-johnson-2}, the matrix
\begin{equation}
	\begin{pmatrix}
		\sqrt{A^{\mathrm{T}} A} & A \\
		A^{\mathrm{T}} & \sqrt{A^{\mathrm{T}} A}
	\end{pmatrix}
	\succcurlyeq 0
\end{equation}
is positive-semidefinite, which proves the physicality condition
$\mathcal{P}(\gamma^\star) \succcurlyeq 0$. In addition, the matrix
$\mathcal{P}(\gamma^\star)$ is singular, i.e. the matrix $\gamma^\star$ is on
the boundary of the physical covariance matrices set, as it must be due to the
KKT condition \eqref{eq:strKKT} and Lemma~\ref{lmn:AB}. It is instructive to see
that the covariance matrix given by Eq.~\eqref{eq:gammastar} satisfies the
relation \eqref{eq:gOg}. In fact, we have
\begin{equation}
\begin{split}
	4\gamma^\star \Omega \gamma^\star &= -\sqrt{M^{-1}}\sqrt{-A^2}\, A^{-1\,} \sqrt{-A^2} \sqrt{M^{-1}} \\
	&= \sqrt{M^{-1}}A\sqrt{M^{-1}} = \Omega.
\end{split}
\end{equation}
Here we used that fact that any two functions of the same matrix $A$ commute
($\sqrt{-A^2}$ and $A^{-1}$).

We now show that 
\begin{equation}\label{eq:minstr}
	\min_\gamma \tr(M\gamma) = \str(M),
\end{equation}
where $\str(M)$ is the symplectic trace of $M$, the sum of the symplectic
eigenvalues of $M$. The term ``symplectic eigenvalue'' was introduced in
Ref.~\cite{GiedkeECP03}. According to the Williamson theorem \cite{*[] [{
Theorem 8.11, p. 244}] SGQM}, for a positive definite $M$ there is a symplectic
matrix $S$ such that 
\begin{equation}\label{eq:SD}
	S^{\mathrm{T}} M S = 
	\begin{pmatrix}
		D & 0 \\
		0 & D
	\end{pmatrix},
\end{equation}
where $D$ is a diagonal matrix with diagonal elements $\lambda_1, \ldots,
\lambda_n$. These diagonal elements are uniquely defined by $M$ as follows: $\pm
i \lambda_j$ are the eigenvalues of $\Omega M$. The sum of these eigenvalues is
referred to as the symplectic trace of $M$
\begin{equation}
	\str(M) = \tr(D) = \lambda_1 + \ldots + \lambda_n.
\end{equation}
For the primal minimization problem we have
\begin{displaymath}
\begin{split}
	\min_\gamma &\tr(M\gamma) = \min_\gamma \tr\left(S^{\mathrm{T}-1}
	\begin{pmatrix}
		D & 0 \\
		0 & D
	\end{pmatrix}
	S^{-1} \gamma
	\right) \\
	&= \min_\gamma \tr\left(
	\begin{pmatrix}
		D & 0 \\
		0 & D
	\end{pmatrix}
	S^{-1} \gamma S^{\mathrm{T}-1}
	\right) \\
	&= \min_{\gamma'} \tr\left(
		\begin{pmatrix}
			D & 0 \\
			0 & D
		\end{pmatrix}
		\gamma'
	\right)
	= \sum^n_{j=1} \lambda_j (\gamma_{xx, jj} + \gamma_{pp, jj}),
\end{split}
\end{displaymath}
where $\gamma' = S^{-1} \gamma S^{\mathrm{T}-1}$ is another physical CM and
$\gamma'$ ranges over all CMs when $\gamma$ does. The classical uncertainty
relation states the inequality
\begin{equation}
\begin{split}
	&\gamma'_{xx, jj} + \gamma'_{pp, jj} \equiv \langle(\Delta x_j)^2\rangle + \langle(\Delta p_j)^2\rangle \\
    &\geqslant 2\sqrt{\langle(\Delta x_j)^2\rangle \langle(\Delta p_j)^2\rangle} \geqslant 1,
\end{split}
\end{equation}
and thus we conclude that
\begin{equation}
	\min_\gamma \tr(M\gamma) = \lambda_1 + \ldots + \lambda_n = \str(M).
\end{equation}
The equality is attained for 
\begin{equation}
	\gamma^{\prime\star} = S^{-1} \gamma^\star S^{\mathrm{T}-1} = \frac{1}{2}E_{2n}.
\end{equation}
For the optimal solution we have
\begin{equation}
	\gamma^\star = \frac{1}{2} S S^{\mathrm{T}},
\end{equation}
where $S$ is defined by Eq.~\eqref{eq:SD}. Note that such $S$ does not have to
be unique, but according to \cite{*[] [{ Proposition 8.12, p. 245}] SGQM-2} the
product $S S^{\mathrm{T}}$ is the same for all symplectic $S$ that diagonalize
$M$. An explicit expression for this product has been obtained above,
Eq.~\eqref{eq:gammastar}.

For any symplectic matrix $S \in \Sp(2n)$ the matrix $\gamma$ defined via
\begin{equation}
	\gamma = \frac{1}{2} S S^{\mathrm{T}}
\end{equation}
is a physical CM, and the minimum in Eq.~\eqref{eq:minstr} is attained on a CM
of this form, so we obtain the following relation:
\begin{equation}
	\min_{S \in \Sp(2n)} \tr(S^{\mathrm{T}}MS) = 2\str(M).
\end{equation}
This relation has been derived in Ref.~\cite{PhysRevA.73.012330} in a more
complicated way.

All the constructions up to now were made in the assumption $M \succ 0$. If $M$
is degenerate then the derivations above are not applicable, but the minimum of
$\tr(M\gamma)$ still exists. Any degenerate $M$ is a limit for $\varepsilon \to
0$ of some non-degenerate $M_\varepsilon$ with small $\varepsilon$-disturbance
of the elements, and the continuity argument shows that the equality
\eqref{eq:Mstr} is also valid in the degenerate case.

The results for the reduced case are straightforwardly derived from the general
case since for any full physical CM $\gamma$ its blocks $\gamma_{xx}$ and
$\gamma_{pp}$ (i.e. setting $\gamma_{xp} = 0$ in $\gamma$) also form a physical
(reduced) CM.
\end{proof}

The symplectic trace is not defined for degenerate matrices, but the equality
\eqref{eq:Mstr} allows one to extend this notion to all positive-semidefinite
matrices $M$:
\begin{equation}\label{eq:strM}
	\str(M) = \frac{1}{2} \tr \sqrt{\sqrt{M}\Omega^{\mathrm{T}}M\Omega\sqrt{M}},
\end{equation}
and to all positive-semidefinite pairs $X$ and $P$:
\begin{equation}
	\str(X, P) = \tr \sqrt{\sqrt{X}P\sqrt{X}} = \tr \sqrt{\sqrt{P}X\sqrt{P}}.
\end{equation}
The latter expression is the symplectic trace of the block-diagonal matrix given
by Eq.~\eqref{eq:M-XP}. For degenerate matrices the minimum given by these
expressions is not necessarily achievable. Consider an extreme (but legal) case
of $X = E$ and $P = 0$. Then $\str(X, P) = 0$ and
\begin{equation}
	\tr(X\gamma_{xx} + P\gamma_{pp}) = \tr(\gamma_{xx}) > 0
\end{equation}
for all physical CMs $\gamma_{xx}$. In this case the minimum is not achievable,
though can be approached arbitrarily close by the CM $\gamma_{xx} = (1/r) E$,
$\gamma_{pp} = (r/4) E$ for $r \to +\infty$.

It is not easy to give necessary and sufficient conditions for the existence of
the optimal solution. In the full case we need to solve the equation
\eqref{eq:XAXB}, which has the form
\begin{equation}
	XAX = B.
\end{equation}
Some general conditions for just the solvability of such equations are given in
Ref.~\cite{quadme}, but we also need that a solution be a physical CM. In the
reduced case the KKT condition \eqref{eq:KKTreduced} gives the equality
\begin{equation}
	X\gamma^{\star}_{xx} = \gamma^{\star}_{pp}P,
\end{equation}
so the matrices $X$ and $P$ should be of the same rank. Because we need only the
optimal value of the optimization problems and not the condition for the
attainability of this value, we do not explore this topic further.

\section{Symmetric states}\label{sec:VI}

We now construct a class of genuine multipartite entangled states with any
number of parts which we will use as a benchmark below. Let us start with the
matrices of a special form
\begin{equation}\label{eq:Gxy}
	G(x, y) = 
	\begin{pmatrix}
		x & y & \ldots & y & y \\
		y & x & \ldots & y & y \\
		\hdotsfor{5} \\
		y & y & \ldots & x & y \\
		y & y & \ldots & y & x
	\end{pmatrix}.
\end{equation}
The eigenvalues of this matrix are
\begin{equation}\label{eq:lmxy}
	\lambda = x + (n-1)y, \quad \mu = x - y,
\end{equation}
where the latter is with multiplicity $n-1$. The eigenvector corresponding
$\lambda$ is $(1, 1, 1, \ldots, 1)$, the eigenvectors of $\mu$ are $(-1, 1, 0,
\ldots, 0)$, $(-1, 0, 1, \ldots, 0)$, \ldots, $(-1, 0, 0, \ldots, 1)$. It
follows that the matrix \eqref{eq:Gxy} is positive semidefinite iff $\lambda,
\mu \geqslant 0$. The necessary condition $x \geqslant 0$ automatically follows
from these two.

To compute the symplectic trace of a matrix, we need to take square roots and
products (and for the minimizing state we also need the inverse, provided that
the matrix is non-degenerate). It turns out that the square root, product and
inverse of matrices of the form \eqref{eq:Gxy} are also matrices of the same
form. For the square root we have
\begin{equation}\label{eq:GGsqrt}
	\sqrt{G(x, y)} = G(u, v),
\end{equation}
where $u$ and $v$ read as
\begin{equation}
\begin{split}
	u &= \frac{1}{n}\left[\sqrt{x + (n-1)y} + (n-1)\sqrt{x-y}\right], \\
	v &= \frac{1}{n}\left[\sqrt{x + (n-1)y} - \sqrt{x-y}\right].
\end{split}
\end{equation}
The product is given by
\begin{displaymath}
	G(x, y) G(p, q) = G(x p + (n-1)y q, x q + y p + (n-2)y q),
\end{displaymath}
and the inverse is
\begin{equation}\label{eq:GGinv}
	G^{-1}(x, y) = G\left(\frac{x+(n-2)y}{\lambda\mu}, -\frac{y}{\lambda\mu}\right),
\end{equation}
where the denominator is non-zero when $\lambda, \mu > 0$. The relations between
the parameters of the resulting matrix and the original matrices are
complicated, so we construct a parametrization in which these relations become
trivial. Note that we can express $x$ and $y$ in terms of the eigenvalues
\eqref{eq:lmxy} as follows:
\begin{equation}\label{eq:lm}
	x = \frac{\lambda + (n-1)\mu}{n}, \quad y = \frac{\lambda - \mu}{n}.
\end{equation}
We now change the parametrization
\begin{equation}
	C(\lambda, \mu) = G\left(\frac{\lambda + (n-1)\mu}{n}, \frac{\lambda - \mu}{n}\right).
\end{equation}
In this parametrization the rules for the square root, product and inverse have
an especially simple form:
\begin{equation}
\begin{split}
	\sqrt{C(\lambda, \mu)} &= C(\sqrt{\lambda}, \sqrt{\mu}), \\
	C(\lambda, \mu) C(\lambda', \mu') &= C(\lambda \lambda', \mu \mu'), \\
	C(\lambda, \mu)^{-1} &= C(\lambda^{-1}, \mu^{-1}),
\end{split}
\end{equation}
where the last equality is valid if $\lambda, \mu > 0$. 

\begin{thrm}
The reduced symmetric states $\gamma^S(\lambda, \mu)$ of the form 
\begin{equation}\label{eq:gammalm}
	\gamma^{S}_{xx} = \frac{1}{2}C(\lambda, \mu), \quad \gamma^{S}_{pp} = \frac{1}{2}C(\lambda^{-1}, \mu^{-1})
\end{equation}
are pure boundary quantum states for all $\lambda, \mu > 0$.
\end{thrm}
\begin{proof}
The physicality condition reads as
\begin{equation}
	\mathcal{P}(\gamma_{xx}, \gamma_{pp}) = 
	\frac{1}{2}
	\begin{pmatrix}
		C(\lambda, \mu) & E \\
		E & C(\lambda^{-1}, \mu^{-1})
	\end{pmatrix}
	\succcurlyeq 0.
\end{equation}
The eigenvalues of this matrix are
\begin{equation}
	\left(0, \ldots, 0, \lambda + \lambda^{-1}, \mu + \mu^{-1}, \ldots, \mu + \mu^{-1}\right),
\end{equation}
all of which are non-negative and also contain zeros. The eigenvectors
corresponding to zero eigenvalues are
\begin{displaymath}
	\left(\frac{\lambda-\mu}{n\lambda\mu}, \ldots, \underline{-\frac{(n-1)\lambda + \mu}{n\lambda\mu}}, 
	\ldots, \frac{\lambda-\mu}{n\lambda\mu}, 0, \ldots, \underline{1}, \ldots, 0\right),
\end{displaymath}
where the underlined terms are on the $i$ and $n+i$ position, respectively, for
$i = 1, \ldots, n$. The eigenvector with the eigenvalue $\lambda + \lambda^{-1}$ is
\begin{equation}
	(\lambda, \ldots, \lambda, 1, \ldots, 1).
\end{equation}
The eigenvectors with the eigenvalue $\mu + \mu^{-1}$ are
\begin{equation}
	(-\mu, 0, \ldots, \underline{\mu}, \ldots, 0, -1, 0, \ldots, \underline{1}, \ldots 0),
\end{equation}
where the underlined terms are in the position $i$ and $n+i$, respectively, for
$i = 2, \ldots, n$.
\end{proof}
These CMs are diagonal if $\lambda = \mu$, and for $\lambda = \mu = 1$ it is the
CM of the vacuum state, $\gamma^S(1, 1) = (1/2)E$. These states have been
introduced in \cite{PhysRevLett.84.3482}, see also
Refs.~\cite{PhysRevA.67.052315, PhysRevLett.93.220504, PhysRevA.71.032349}.
Below we show that these non-diagonal states, i.e. if $\lambda \not= \mu$ or,
equivalently, if $y \not= 0$ in \eqref{eq:Gxy}, are genuine multipartite
entangled. These states will be used as a benchmark for the entanglement
detection optimization problems presented in the next section.

If we take a reduced matrix with $X = C(\lambda, \mu)$ and $P = C(\lambda',
\mu')$, where all the parameters are non-zero, then the corresponding optimal
covariance matrix, Eq.~\eqref{eq:gammaXP}, is of the same form 
\begin{displaymath}
	\gamma^\star_{xx} = \frac{1}{2}C\left(\sqrt{\frac{\lambda'}{\lambda}}, \sqrt{\frac{\mu'}{\mu}}\right), \quad
	\gamma^\star_{pp} = \frac{1}{2}C\left(\sqrt{\frac{\lambda}{\lambda'}}, \sqrt{\frac{\mu}{\mu'}}\right).
\end{displaymath}
In other words, the optimal solution is the symmetric state
\begin{equation}
	\gamma^\star = \gamma^S\left(\sqrt{\frac{\lambda'}{\lambda}}, \sqrt{\frac{\mu'}{\mu}}\right).
\end{equation}
It follows that for the special choice of $\lambda' = \lambda^{-1}$ and $\mu' =
\mu^{-1}$ we have
\begin{displaymath}
	\gamma^\star_{xx} = \frac{1}{2}C(\lambda^{-1}, \mu^{-1}) = \frac{1}{2}P, \quad
	\gamma^\star_{pp} = \frac{1}{2}C(\lambda, \mu) = \frac{1}{2}X,
\end{displaymath}
so the optimal solution is the same as the reduced matrix with the blocks $X$ and $P$
exchanged (and scaled by the factor $1/2$). 

\section{Entanglement detection}\label{sec:IV}

Let $\mathcal{I} = \{I_1, \ldots, I_l\}$ be some partition of the set $[n] =
\{1, \ldots, n\}$ of $n$ indices, i.e. $I_i \cap I_j = \varnothing$ and 
\begin{equation}
	I_1 \cup \ldots \cup I_l = [n].
\end{equation}
An $n$-partite state $\hat{\varrho}$ is called $\mathcal{I}$-factorizable if it
is a product of $l$ states with indices $I_i$, $i = 1, \ldots, l$:
\begin{equation}
	\hat{\varrho} = \hat{\varrho}^{(I_1)} \otimes \ldots \otimes \hat{\varrho}^{(I_l)}.
\end{equation}
A state $\hat{\varrho}_{\mathcal{I}}$ is called $\mathcal{I}$-separable if
it is a convex combination of $\mathcal{I}$-factorizable states: 
\begin{equation}
	\hat{\varrho}_{\mathcal{I}} = \sum_i p_i \hat{\varrho}^{(I_1)}_i \otimes \ldots \otimes \hat{\varrho}^{(I_l)}_i,
\end{equation}
where $p_i \geqslant 0$ sum up to one. If $\mathfrak{I} = \{\mathcal{I}_1,
\ldots, \mathcal{I}_k\}$ is a set of partitions then a state
$\hat{\varrho}_{\mathfrak{I}}$ is called $\mathfrak{I}$-separable if it is a
convex combination of $\mathcal{I}_i$-separable states, $i = 1, \ldots, k$:
\begin{equation}
	\hat{\varrho}_{\mathfrak{I}} = \sum_{\mathcal{I} \in \mathfrak{I}} 
	p_{\mathcal{I}} \hat{\varrho}_{\mathcal{I}},
\end{equation}
where $p_{\mathcal{I}} \geqslant 0$ sum up to one. For $k=1$ and $\mathfrak{I} =
\{\mathcal{I}\}$ this definition reduces to $\mathcal{I}$-separability. For $k =
2^{n-1}-1$ and $\mathfrak{I}_2$ being the set of all bipartitions of $[n]$, a
$\mathfrak{I}_2$-separable state is called biseparable. A non-biseparable state
is called genuine multipartite entangled. Biseparability puts the loosest
requirement on the state, so the set of the biseparable states is the largest
among other sets of separable states. Conversely, the set of genuine entangled
states is the smallest among other sets of entangled states.

Given a partition $\mathcal{I} = \{I_1, \ldots, I_l\}$, the CM
$\gamma_{\mathcal{I}}$ of any $\mathcal{I}$-separable state satisfies the
inequality 
\begin{equation}\label{eq:gAB}
	\gamma \succcurlyeq \gamma^{(I_1)} \oplus \ldots \oplus \gamma^{(I_l)}
\end{equation}
for some physical covariance matrices $\gamma^{(I_i)}$
\cite{PhysRevLett.86.3658}. Using the concavity of CM
\begin{equation}\label{eq:CMconcave}
	\gamma\left(\sum_i p_i \hat{\varrho}_i\right) \geqslant \sum_i p_i \gamma(\hat{\varrho}_i),
\end{equation}
we derive that the CM $\gamma_{\mathfrak{I}}$ of any $\mathfrak{I}$-separable
state $\hat{\varrho}_{\mathfrak{I}}$ satisfies the inequality
\begin{equation}\label{eq:gammatest}
	\gamma_{\mathfrak{I}} \geqslant \sum_{\mathcal{I} = \{I_1, \ldots, I_l\} \in \mathfrak{I}} 
	p_{\mathcal{I}} \gamma^{(I_1)} \oplus \ldots \oplus \gamma^{(I_l)},
\end{equation}
for some $p_{\mathcal{I}}$ and $\gamma^{(I)}$. It is this condition that we will
use for the two types of tests below. We refer to a CM $\gamma$ as
$\mathfrak{I}$-separable if it satisfies an inequality of the form
\eqref{eq:gammatest}. The CM of any $\mathfrak{I}$-separable state is also
$\mathfrak{I}$-separable, but the reverse need not be the case even for Gaussian
states. As has been recently shown in \cite{Baksova2025}, if $|\mathfrak{I}| >
1$ then a $\mathfrak{I}$-entangled Gaussian state can have the
$\mathfrak{I}$-separable CM.

The two types are the additive and the multiplicative tests. If we have two
positive quantities $Q$ and $Q'$ and we want to know which one is larger, we can
either compare their difference $Q-Q'$ with zero, or compare their ratio $Q/Q'$
with one. These tasks are trivial when $Q$ and $Q'$ are given explicitly, and
it does not really matter which type of test to use. But when $Q-Q'$ and $Q/Q'$
cannot be computed directly and are results of complicated computational
procedures, the situation becomes less clear. The efficiency and the precision of
the outcome of the comparison might differ for the two types of tests, depending
on the nature of the computational procedures used to compute the difference and
the ratio. That is why we provide both additive and multiplicative version of the
entanglement test.

\subsection{Eigenvalue test}\label{sec:IVA}

\begin{figure*}
	\includegraphics[scale=0.73]{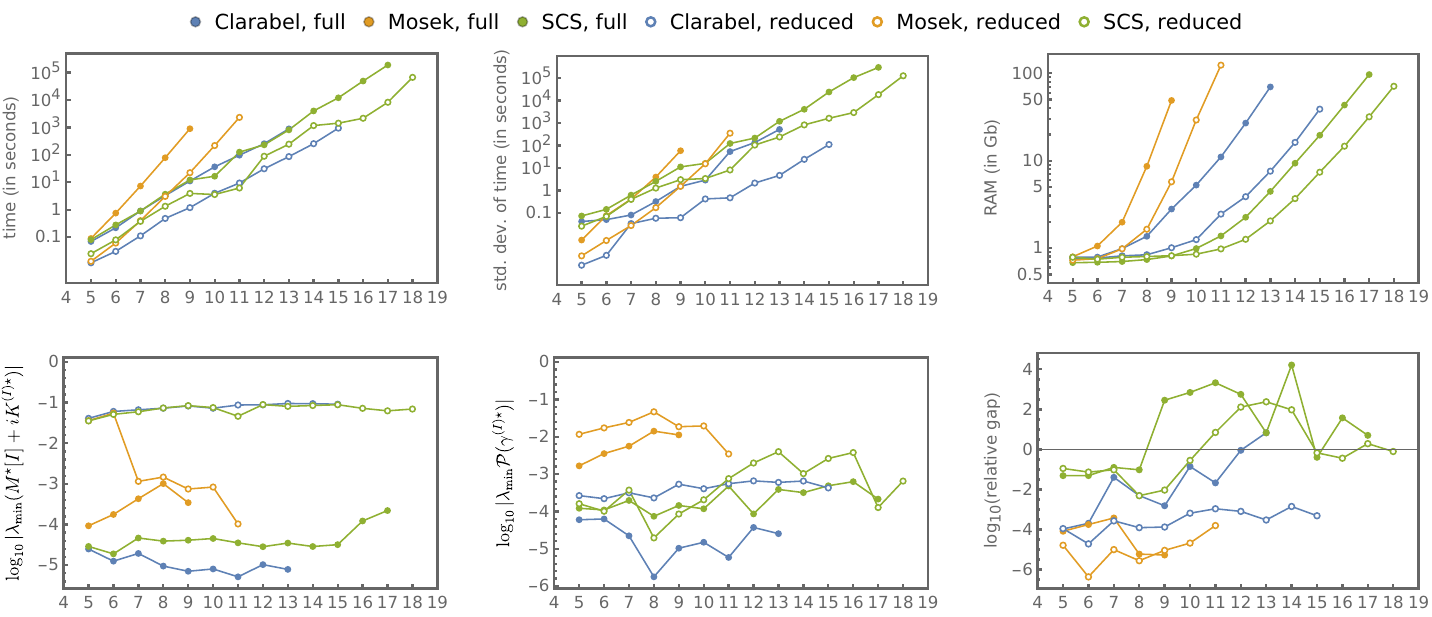}
	\caption{Comparing performance and accuracy characteristics of different
	solvers for the eigenvalue entanglement test as a function of the number of
	modes. CMs are given by Eq.~\eqref{eq:gammalm} for random $\lambda$ and
	$\mu$.}
	\label{fig:maxmineigval}
\end{figure*}

From inequality \eqref{eq:gammatest} it follows that for any $M \succcurlyeq 0$ the inequality
\begin{displaymath}
\begin{split}
	\tr(M\gamma_{\mathfrak{I}}) &\geqslant 
	\sum_{\mathcal{I} \in \mathfrak{I}} p_{\mathcal{I}} \min_{\gamma^{(I_1)}, \ldots, \gamma^{(I_l)}}
	\tr(M \gamma^{(I_1)} \oplus \ldots \oplus \gamma^{(I_l)}) \\
	&\geqslant \min_{\mathcal{I} \in \mathfrak{I}} \min_{\gamma^{(I_1)}, \ldots, \gamma^{(I_l)}} 
	\tr(M \gamma^{(I_1)} \oplus \ldots \oplus \gamma^{(I_l)}).
\end{split}
\end{displaymath}
is satisfied. The minimum on the right-hand side can be easily computed
\begin{displaymath}
\begin{split}
	\min_{\gamma^{(I_1)}, \ldots, \gamma^{(I_l)}}
	&\tr(M \gamma^{(I_1)} \oplus \ldots \oplus \gamma^{(I_l)}) \\
	&= \sum_{I \in \mathcal{I}} \min_{\gamma^{(I)}} \tr(M[I] \gamma^{(I)}) 
	= \sum_{I \in \mathcal{I}}\str(M[I]),
\end{split}
\end{displaymath}
where $M[I]$ is the submatrix of $M$ with row and column indices equal to $I
\cup (I+n)$. We thus derive the condition that the CM $\gamma_{\mathfrak{I}}$ of
any $\mathfrak{I}$-separable state satisfies
\begin{equation}
	\tr(M\gamma_{\mathfrak{I}}) \geqslant \min_{\mathcal{I} \in \mathfrak{I}} 
	\sum_{I \in \mathcal{I}} \str(M[I]),
\end{equation}
for any $M \succcurlyeq 0$. If for a given CM $\gamma$ we manage to find a
matrix $M \succcurlyeq 0$ such that
\begin{equation}\label{eq:MW}
\begin{split}
	&\tr(M\gamma) < \min_{\mathcal{I} \in \mathfrak{I}} 
	\sum_{I \in \mathcal{I}} \str(M[I]) \\
	&= \frac{1}{2}\min_{\mathcal{I} \in \mathfrak{I}}\sum_{I \in \mathcal{I}} 
	\tr\sqrt{\sqrt{M[I]}\Omega^{\mathrm{T}}_{|I|}M[I]\Omega_{|I|}\sqrt{M[I]}},
\end{split}
\end{equation}
using Eq.~\eqref{eq:strM}; then $M$ witnesses that $\gamma$ cannot be
$\mathfrak{I}$-separable. In the reduced case this inequality reads as
\begin{equation}\label{eq:entXP}
	\tr(X \gamma_{xx} + P \gamma_{pp}) < \min_{\mathcal{I} \in \mathfrak{I}}
	\sum_{I \in \mathcal{I}} \tr\sqrt{\sqrt{X[I]} P[I] \sqrt{X[I]}},
\end{equation}
where $X[I]$ and $P[I]$ are submatrices of $X$ and $P$, respectively, with row
and column indices in $I$. We now give a simple application of this test.
\begin{thrm}\label{thrm:lm} 
The states in Eq.~\eqref{eq:gammalm} are genuine
multipartite entangled provided that $\lambda \not= \mu$.
\end{thrm}
\begin{proof}
We just need to find a pair of positive semidefinite matrices $X$ and $P$ that
satisfy the inequality \eqref{eq:entXP} for $\mathfrak{I} = \mathfrak{I}_2$, the
set of all $2^{n-1}-1$ non-trivial bipartitions of $n$ parts. We take $X =
C(\alpha, \beta)$ and $P = C(\alpha', \beta')$ and find their parameters to
satisfy that condition. We have
\begin{equation}
\begin{split}
	&\tr(X \gamma^S_{xx} + P \gamma^S_{pp}) \\
	&= \frac{1}{2}\left(\frac{\alpha}{\lambda} + (n-1)\frac{\beta}{\mu} + \alpha'\lambda + (n-1)\beta'\mu\right).
\end{split}
\end{equation}
To compute the right-hand side of Eq.~\eqref{eq:entXP} note that due to symmetry
of the matrices of the form $C(\alpha, \beta)$, both $X[I]$ and $P[I]$ depend
only on $|I| = k$, so we need to take the minimum over $n-1$ values (for $k = 1,
\ldots, n-1$) instead of all $2^{n-1}-1$ non-trivial bipartitions. Every
submatrix $C(\alpha, \beta)[I] = C(\tilde{\alpha}, \tilde{\beta})$ is also of
the same form, where the new parameters are related to the original ones via
\begin{equation}
	\frac{\tilde{\alpha} + (k-1)\tilde{\beta}}{k} = \frac{\alpha + (n-1)\beta}{n}, \quad
	\frac{\tilde{\alpha} - \tilde{\beta}}{k} = \frac{\alpha - \beta}{n}.
\end{equation}
Explicitly they read as
\begin{equation}
	\tilde{\alpha} = \frac{k\alpha + (n-k)\beta}{n}, \quad \tilde{\beta} = \beta.
\end{equation}
We then have $X[I] = C(\tilde{\alpha}, \tilde{\beta})$, $P[I] =
C(\tilde{\alpha}', \tilde{\beta}')$ and 
\begin{equation}
\begin{split}
    &\tr\sqrt{\sqrt{X[I]} P[I] \sqrt{X[I]}} 
	= \sqrt{\tilde{\alpha}\tilde{\alpha}'} + (k-1) \sqrt{\tilde{\beta}\tilde{\beta'}} \\
	& = \frac{1}{n}\sqrt{(k\alpha + (n-k)\beta)(k\alpha' + (n-k)\beta')} \\
	&+ (k-1) \sqrt{\beta\beta'}.
\end{split}
\end{equation}
The sum of this expression for $k$ and $n-k$, minimized over $k = 1, \ldots,
n-1$, is the right-hand side of Eq.~\eqref{eq:entXP}. We thus need to find
$\alpha$, $\beta$, $\alpha'$ and $\beta'$ such that the following inequality is
satisfied
\begin{equation}\label{eq:abab}
\begin{split}
	&\frac{1}{2}\left(\frac{\alpha}{\lambda} + (n-1)\frac{\beta}{\mu} + \alpha'\lambda + (n-1)\beta'\mu\right) \\
	&< \frac{1}{n}\sqrt{(k\alpha + (n-k)\beta)(k\alpha' + (n-k)\beta')} \\
	&+ \frac{1}{n}\sqrt{((n-k)\alpha + k\beta)((n-k)\alpha' + k\beta')} \\
	&+ (n-2) \sqrt{\beta\beta'},
\end{split}	
\end{equation}
for all $k = 1, \ldots, n-1$. We show that for $\alpha = \lambda^2\mu^{-1}$,
$\beta = \mu$, $\alpha' = \beta' = \mu^{-1}$ these inequalities hold. Note that
$P = C(\alpha', \beta')$ is diagonal since $\alpha' = \beta'$, but strictly
positive-definite, as it must be. Substituting these values into the inequality
\eqref{eq:abab}, we can simplify it as follows:
\begin{displaymath}
	\lambda + \mu < \sqrt{\frac{k}{n}\lambda^2 + \left(1-\frac{k}{n}\right)\mu^2} 
	+\sqrt{\left(1-\frac{k}{n}\right)\lambda^2 + \frac{k}{n}\mu^2}.
\end{displaymath}
Because both sides are nonnegative, we can square this inequality and get an
equivalent one
\begin{equation}
	\lambda \mu < \sqrt{\frac{k}{n}\lambda^2 + \left(1-\frac{k}{n}\right)\mu^2} 
	\sqrt{\left(1-\frac{k}{n}\right)\lambda^2 + \frac{k}{n}\mu^2}.
\end{equation}
Squaring again and taking the difference of the left-hand side and the
right-hand side, for a bipartition $\mathcal{I} = \{I, \overline{I}\}$ with
$|I|=k$ we obtain
\begin{equation}
	-\frac{k}{n}\left(1-\frac{k}{n}\right)(\lambda^2 - \mu^2)^2 < 0,
\end{equation}
for all $k = 1, \ldots, n-1$ provided that $\lambda \not= \mu$. This verifies
that the states given by Eq.~\eqref{eq:gammalm} are genuine multipartite
entangled if the parameters $\lambda$ and $\mu$ are different.
\end{proof}

The states \eqref{eq:gammalm} have a high degree of symmetry, so it was possible
to verify their entanglement analytically. However, in general, it is almost
impossible to find a witness just by guessing it, so we need a more systematic
approach. In addition, we aim not just to find a witness for a given entangled
state, but to find the best possible one, the witness for which the difference
of the left-hand side and the right-hand side of Eq.~\eqref{eq:MW} is the
largest possible (by absolute value). Since scaling $M$ by a positive factor
also scales the difference by the same factor, we put the additional restriction
on $M$, $\tr(M) = 1$. Given a CM $\gamma$, we thus need to solve the following
optimization problem:
\begin{equation}\label{eq:EM}
	-\mathcal{E}_{\mathfrak{I}}(\gamma) = \min_{\substack{M \succcurlyeq 0 \\ \tr(M) = 1}} 
	\left[\tr(M\gamma) - \min_{\mathcal{I} \in \mathfrak{I}}
	\sum_{I \in \mathcal{I}} \str(M[I])\right].
\end{equation}
If the optimal value turns out to be negative or, equivalently,
$\mathcal{E}_{\mathfrak{I}}(\gamma) > 0$, the corresponding optimal $M^\star$ is
the optimal witness for $\mathfrak{I}$-entanglement of $\gamma$ (different
$\gamma$ are likely to have different optimal witnesses). Using the analytical
expression \eqref{eq:Mstr} for the symplectic trace together with
Eq.~\eqref{eq:strM}, the objective function has an analytic expression as a
function of $M$. This expression has a complicated structure, as we have already
seen in the proof of the theorem above, so no standard optimization technique
can be applied straightforwardly to it. To tackle this problem, we use the dual
expression instead. This approach is expressed by the theorem below.
\begin{thrm}\label{thrm:additive}
For any CM $\gamma$ the following equality takes place:
\begin{equation}\label{eq:ent-test-full}
	\begin{split}
		&\min_{M \succcurlyeq 0} \left[\tr(M\gamma) - \min_{\mathcal{I} \in \mathfrak{I}}
		\sum_{I \in \mathcal{I}} \str(M[I])\right] \\
		&= \min_{M \succcurlyeq 0, K^{(I)}} 
		\max_{\mathcal{I} \in \mathfrak{I}}\left[
		\tr(M\gamma) - \frac{1}{2}\sum_{I \in \mathcal{I}} \tr(K^{(I)}\Omega_{|I|})\right] \\
		&= \max_{\gamma^{(I)},  p_{\mathcal{I}}}
		\lambda_{\mathrm{min}}\left(\gamma - \sum_{\mathcal{I} \in \mathfrak{I}} p_{\mathcal{I}}
		\bigoplus_{I \in \mathcal{I}} \gamma^{(I)}\right),
	\end{split}
\end{equation}
the additional constraints of the primal problem being
\begin{equation}
	\tr(M) = 1, \quad M[I] + i K^{(I)} \succcurlyeq 0,
\end{equation}
and the constraints of the dual problem being
\begin{equation}\label{eq:p1}
	p_{\mathcal{I}} \geqslant 0, \quad
	\sum_{\mathcal{I} \in \mathfrak{I}} p_{\mathcal{I}} = 1.
\end{equation}
The KKT conditions read as
\begin{equation}\label{eq:KKT-1}
	\mathcal{P}(\gamma^{(I)\star})(M^\star[I] + i K^{(I)\star}) = 0, 
\end{equation}
for all $\mathcal{I} \in \mathfrak{I}$ and $I \in \mathcal{I}$ with
$p^\star_{\mathcal{I}} > 0$, and 
\begin{equation}
	M^\star\left(\gamma - \sum_{\mathcal{I} \in \mathfrak{I}}
	p^\star_{\mathcal{I}} \bigoplus_{I \in \mathcal{I}} \gamma^{(I)\star} \right)
	 = \lambda^\star_{\mathrm{min}} M^\star,
\end{equation}
where for all $\mathcal{I} \in \mathfrak{I}$ with $p^\star_{\mathcal{I}} > 0$
\begin{equation}
	\lambda^\star_{\mathrm{min}} = \tr(M^\star\gamma) - \frac{1}{2}\sum_{I \in \mathcal{I}} 
	\tr(K^{(I)\star}\Omega_{|I|})
\end{equation}
is the primal optimal objective value. Note that $M^\star[I]$ here are
submatrices of one and the same matrix $M^\star$, and $K^{(I) \star}$ are
totally independent matrices, which is reflected in different notation.

In the reduced case the duality reads as
\begin{displaymath}
	\begin{split}
		&\min_{X, P \succcurlyeq 0} \tr(X\gamma_{xx} + P\gamma_{pp}) - \min_{\mathcal{I} \in \mathfrak{I}}
		\sum_{I \in \mathcal{I}} \str(X[I], P[I]) \\
		&= \min_{X, P \succcurlyeq 0, Z^{(I)}} 
		\max_{\mathcal{I} \in \mathfrak{I}}\left[
		\tr(X\gamma_{xx} + P\gamma_{pp}) - \sum_{I \in \mathcal{I}} \tr Z^{(I)}\right] \\
		&= \max_{\gamma^{(I)}_{xx}, \gamma^{(I)}_{pp}, p_{\mathcal{I}}} 
		\lambda_{\mathrm{min}}\left(\gamma_{xx} - \sum_{\mathcal{I} \in \mathfrak{I}} p_{\mathcal{I}}
		\bigoplus_{I \in \mathcal{I}} \gamma^{(I)}_{xx}\right) \\
		&= \max_{\gamma^{(I)}_{xx}, \gamma^{(I)}_{pp}, p_{\mathcal{I}}} 
		\lambda_{\mathrm{min}}\left(\gamma_{pp} - \sum_{\mathcal{I} \in \mathfrak{I}} p_{\mathcal{I}}
		\bigoplus_{I \in \mathcal{I}} \gamma^{(I)}_{pp}\right),
	\end{split}
\end{displaymath}
the additional constraints of the primal problem being
\begin{equation}
	\tr(X + P) = 1, \quad
	\begin{pmatrix}
		X[I] & Z^{(I)} \\ 
		Z^{(I)\mathrm{T}} & P[I]
	\end{pmatrix}
	\succcurlyeq 0,
\end{equation}
and the additional constraint of the dual problem given by Eq.~\eqref{eq:p1}.
The KKT conditions are
\begin{equation}\label{eq:MMKKT}
	\mathcal{P}(\gamma^{(I)\star}_{xx}, \gamma^{(I)\star}_{pp})
	\begin{pmatrix}
		X^\star[I] & -Z^{(I)\star} \\
		-Z^{(I)\star\mathrm{T}} & P^\star[I]
	\end{pmatrix}
	= 0,
\end{equation}
for all $I \in \mathcal{I}$ with $p^\star_{\mathcal{I}} > 0$, and
\begin{equation}\label{eq:XXKKT}
	\begin{split}
		X^\star\left(\gamma_{xx} - \sum_{\mathcal{I} \in \mathfrak{I}} 
		p^\star_{\mathcal{I}} \bigoplus_{I \in \mathcal{I}} \gamma^{(I)\star}_{xx}\right)
		&= \lambda^\star_{\mathrm{min}} X^\star, \\
		P^\star\left(\gamma_{pp} - \sum_{\mathcal{I} \in \mathfrak{I}} 
		p^\star_{\mathcal{I}} \bigoplus_{I \in \mathcal{I}} \gamma^{(I)\star}_{pp}\right) 
		&= \lambda^\star_{\mathrm{min}} P^\star,
	\end{split}
\end{equation}
where for all $\mathcal{I} \in \mathfrak{I}$ with $p^\star_{\mathcal{I}} > 0$
\begin{equation}
	\lambda^\star_{\mathrm{min}} = \tr(X^\star\gamma_{xx} + P^\star\gamma_{pp}) - 
		\sum_{I \in \mathcal{I}} \tr Z^{(I)\star}
\end{equation}
is the primal optimal objective value.
\end{thrm}
The duality is directly related to the basic condition \eqref{eq:gammatest}. If
that condition holds, then there is a mixture of factorizable states such that
the difference with the given state is positive-semidefinite and thus has
non-negative eigenvalues only. But if the difference with these mixtures always
has a negative eigenvalue, then the condition \eqref{eq:gammatest} cannot hold
for any mixture, so the state under study is entangled. The duality is a
quantitative characterization of the ``best'' possible minimal eigenvalue of the
difference.

The number of variables is given by
\begin{displaymath}
	N_v = 
	\begin{cases}
		n(2n+1) + \sum\limits_{\mathcal{I} \in \mathfrak{I}} \sum\limits_{I \in \mathcal{I}} |I|(2|I| - 1) & \text{full case} \\
		n(n+1) + \sum\limits_{\mathcal{I} \in \mathfrak{I}} \sum\limits_{I \in \mathcal{I}} |I|^2 & \text{reduced case}.
	\end{cases}
\end{displaymath}
The first term is the number of different elements in the symmetric matrix $M$
(or the symmetric matrices $X$ and $P$) and the rest is the number of different
elements in the anti-symmetric matrices $K^{(I)}$ (or the general matrices
$Z^{(I)}$). For the bipartite separability $\mathfrak{I} = \mathfrak{I}_2$ we
have
\begin{equation}
	N_v = 
	\begin{cases}
		2n(n2^{n-2} + 1) & \text{full case} \\
		n[(n+1)2^{n-2} + 1] & \text{reduced case}.
	\end{cases}
\end{equation}
Note that from the KKT condition \eqref{eq:KKT-1} it follows that if
$p_{\mathcal{I}} > 0$ then for all $I \in \mathcal{I}$ the state
$\gamma^{(I)\star}$ is given by Eq.~\eqref{eq:gammaM}, $\gamma^{(I)\star} =
\gamma_{M[I]}$, provided that $M[I] \succ 0$. The optimal separable state in the
dual problem in Eq.~\eqref{eq:ent-test-full} is thus a mixture of factorizable
states, where each factor is a boundary pure state. Given $M^\star$ (or
$X^\star$ and $P^\star$ in the reduced case), we can easily check (without any
optimizations) that it really witnesses the entanglement of the CM $\gamma$ with
the inequality \eqref{eq:MW} (or \eqref{eq:entXP}, respectively). This is where
we need an explicit expression for the symplectic trace. The corresponding dual
solution and KKT conditions verify that this witness is optimal.

In Fig.~\ref{fig:maxmineigval} we show the results of testing the symmetric
states $\gamma^S(\lambda, \mu)$ for 10 different random values of $\lambda$ and
$\mu$ in the range $[0.5, 2.5]$. The average time to test an $n$-partite state
grows exponentially with $n$, as well as the standard deviation of this time,
while the memory needed grows super-exponentially. The large relative gap for
SCS is likely due to the exceeding the default number of iterations, so the
result produced is not fully converged to the solution. But this curve shows the
largest value (the worst) over the 10 runs, for some cases the solution obtained
is much more precise.

The quantity $\mathcal{E}_{\mathfrak{I}}(\gamma)$, given by Eq.~\eqref{eq:EM},
can be considered as an $\mathfrak{I}$-entanglement measure. We show that it is
nonnegative and convex, which justifies using it as a measure of entanglement.
\begin{lmn}
The measure $\mathcal{E}_{\mathfrak{I}}(\gamma) \leqslant 0$ iff $\gamma$ is
$\mathfrak{I}$-separable. In addition, it is a convex function of $\gamma$:
\begin{equation}
	\mathcal{E}_{\mathfrak{I}}(p\gamma_1 + (1-p)\gamma_2) \leqslant
	p \mathcal{E}_{\mathfrak{I}}(\gamma_1) +
	(1-p)\mathcal{E}_{\mathfrak{I}}(\gamma_2)
\end{equation}
for all CMs $\gamma_1, \gamma_2$ and all $0 \leqslant p \leqslant 1$.
\end{lmn}
\begin{proof}
The first statement directly follows from Theorem~\ref{thrm:additive}. As for
the convexity, by the definition \eqref{eq:EM} we have that
$-\mathcal{E}_{\mathfrak{I}}$ is the minimum of an infinite family of linear
functions, and thus is concave. We conclude that $\mathcal{E}_{\mathfrak{I}}$ is
convex.
\end{proof}

In fact, this measure is also convex with respect to states $\hat{\varrho}$, not
only with respect to their CMs. To prove this note that if $\gamma \succcurlyeq
\gamma'$, then $\mathcal{E}_{\mathfrak{I}}(\gamma) \leqslant
\mathcal{E}_{\mathfrak{I}}(\gamma')$. The dual form of
Eq.~\eqref{eq:ent-test-full} shows that by ``increasing'' $\gamma$ we can only
decrease its entanglement measure. Using the concavity of CM, expressed by
Eq.~\eqref{eq:CMconcave}, and the lemma above we conclude that
\begin{displaymath}
	\mathcal{E}_{\mathfrak{I}}(\gamma(p\hat{\varrho}_1 + (1-p)\hat{\varrho}_2)) \leqslant 
	p \mathcal{E}_{\mathfrak{I}}(\gamma(\hat{\varrho}_1)) +
	(1-p) \mathcal{E}_{\mathfrak{I}}(\gamma(\hat{\varrho}_2)).
\end{displaymath}
In addition, the whole set of these measures behaves in an intuitive way with
respect to $\mathfrak{I}$. We say that a partition $\mathcal{I}$ is finer than a
partition $\mathcal{J}$, $\mathcal{I} \prec \mathcal{J}$, if for any $I \in
\mathcal{I}$ there is $J \in \mathcal{J}$ such that $I \subseteq J$. It means
that a $\mathcal{I}$-factorizable state is also $\mathcal{J}$-factorizable. The
partition $1|2|\ldots|n$ is the minimum with respect to this partial order. This
order can be extended to sets of partitions. We say that $\mathfrak{I} \prec
\mathfrak{J}$ if for any $\mathcal{I} \in \mathfrak{I}$ there is $\mathcal{J}
\in \mathfrak{J}$ such that $\mathcal{I} \prec \mathcal{J}$. It means that a
$\mathfrak{I}$-separable state is also $\mathfrak{J}$-separable. In other words,
the larger $\mathfrak{I}$, the larger the set of $\mathfrak{I}$-separable states
(and the smaller the set of $\mathfrak{I}$-entangled states). It follows from
the dual problem of Eq.~\eqref{eq:ent-test-full} that by increasing the set of
separable state we also increase the maximal value of the dual objective, and
thus decrease the corresponding $\mathcal{E}_\mathfrak{I}(\gamma)$ for a fixed
$\gamma$. We have thus derived the following inequality:
\begin{equation}\label{eq:EE}
	\mathcal{E}_\mathfrak{J}(\gamma) \leqslant \mathcal{E}_\mathfrak{I}(\gamma),
\end{equation}
provided that $\mathfrak{I} \prec \mathfrak{J}$. For a given CM $\gamma$,
$\mathcal{E}_{1|\ldots|n}(\gamma)$ is the largest and
$\mathcal{E}_{\mathfrak{I}_2}(\gamma)$ is the smallest.

The measure $\mathcal{E}_{\mathfrak{I}}(\gamma)$ can be strictly negative, for
example, for the scaled vacuum state $\gamma = c \gamma_0 \oplus \ldots \oplus
\gamma_0$ with $c>1$ and 
\begin{equation}
	\gamma_0 = \frac{1}{2}
	\begin{pmatrix}
		1 & 0 \\
		0 & 1
	\end{pmatrix}.
\end{equation}
It thus makes sense to ``fix'' this measure by truncating its values below zero
\begin{equation}\label{eq:EM+}
	\mathcal{E}^+_{\mathfrak{I}}(\gamma) = \max(\mathcal{E}_{\mathfrak{I}}(\gamma), 0).
\end{equation}
This function is nonnegative, convex (as the maximum of two convex functions)
and has the following addition property: $\mathcal{E}^+_{\mathfrak{I}}(\gamma) =
0$ iff $\gamma$ is $\mathfrak{I}$-separable.

\subsection{Scaling test}\label{sec:IVB}

We now use a different approach to obtain an entanglement test from the
inequality \eqref{eq:gammatest}. Instead of the difference between the left-hand
side and the right-hand side we studied above we now study their ``ratio''. More
concretely, given a CM $\gamma$ we will find the maximal value of scaling
factors $t$ that satisfy the condition
\begin{equation}\label{eq:gammatestt}
	\gamma \geqslant t \sum_{\mathcal{I} = \{I_1, \ldots, I_l\} \in \mathfrak{I}} 
	p_{\mathcal{I}} \gamma^{(I_1)} \oplus \ldots \oplus \gamma^{(I_l)},
\end{equation}
for some $p_{\mathcal{I}}$ and $\gamma^{(I)}$. 
From Eq.~\eqref{eq:gammatest} it follows for the maximal scaling factor
$t^\star$ that $t^\star \geqslant 1$ iff $\gamma$ is $\mathfrak{I}$-separable
and thus the condition $t^\star < 1$ is necessary and sufficient for
entanglement. The dual problem is given by the theorem below.
\begin{thrm}\label{thrm:scaling}
For a CM $\gamma$ the following equality is valid:
\begin{equation}
	\max_{\gamma^{(I)}, p_{\mathcal{I}}} t = 
	\min_{M \succcurlyeq 0} \tr(M\gamma),
\end{equation}
with the additional constraint for the primal problem (on the left-hand side)
given by Eqs.~\eqref{eq:gammatestt} and \eqref{eq:p1}, and the additional
constraint for the dual problem (on the right-hand side) given by
\begin{equation}\label{eq:KO}
	M[I] + i K^{(I)} \succcurlyeq 0, \quad \frac{1}{2}\sum_{I \in \mathcal{I}} \tr(K^{(I)} \Omega_{|I|}) = 1,
\end{equation}
for all $\mathcal{I} \in \mathfrak{I}$. The KKT conditions are
Eq.~\eqref{eq:KKT-1} and
\begin{equation}
	M^\star\left(\gamma - t^\star \sum_{\mathcal{I} \in \mathfrak{I}}
	p^\star_{\mathcal{I}} \bigoplus_{I \in \mathcal{I}} \gamma^{(I)\star} \right) = 0,
\end{equation}
where $t^\star = \tr(M^\star \gamma)$ is the primal optimal objective value of the primal
objective function.

In the reduced case the duality reads as
\begin{equation}
	\max_{\gamma^{(I)}_{xx}, \gamma^{(I)}_{pp}, p_{\mathcal{I}}, t} t = 
	\min_{X, P \succcurlyeq 0, Z^{(I)}} 
	\tr(X\gamma_{xx} + P\gamma_{pp}),
\end{equation}
with the additional constraint for the primal problem (on the left-hand side)
being \eqref{eq:p1} and
\begin{equation}
	\gamma_{xx} \succcurlyeq t \sum_{\mathcal{I} \in \mathfrak{I}} p_{\mathcal{I}} 
	\bigoplus_{I \in \mathcal{I}} \gamma^{(I)}_{xx}, \quad
	\gamma_{pp} \succcurlyeq t \sum_{\mathcal{I} \in \mathfrak{I}} p_{\mathcal{I}} 
	\bigoplus_{I \in \mathcal{I}} \gamma^{(I)}_{pp},
\end{equation}
and the additional constraint for the dual problem (on the right-hand side)
being
\begin{equation}
	\begin{pmatrix} 
		X[I] & Z^{(I)} \\ 
		Z^{(I)\mathrm{T}} & P[I] 
	\end{pmatrix}
	\succcurlyeq 0, \quad
	\sum_{I \in \mathcal{I}} \tr Z^{(I)} = 1.
\end{equation}
The KKT conditions are given by Eq.~\eqref{eq:MMKKT} and
\begin{equation}\label{eq:KKTscaling2}
\begin{split}
	X^\star\left(\gamma_{xx} - t^\star \sum_{\mathcal{I} \in \mathfrak{I}} 
	p^\star_{\mathcal{I}} \bigoplus_{I \in \mathcal{I}}\gamma^{(I)\star}_{xx} \right) &= 0 \\
	P^\star\left(\gamma_{pp} - t^\star \sum_{\mathcal{I} \in \mathfrak{I}} 
	p^\star_{\mathcal{I}} \bigoplus_{I \in \mathcal{I}}\gamma^{(I)\star}_{pp} \right) &= 0,
\end{split}
\end{equation}
where $t^\star = \tr(X^\star \gamma_{xx} + P^\star \gamma_{pp})$ is the primal
optimal objective value.
\end{thrm}

\begin{figure*}
	\includegraphics[scale=0.73]{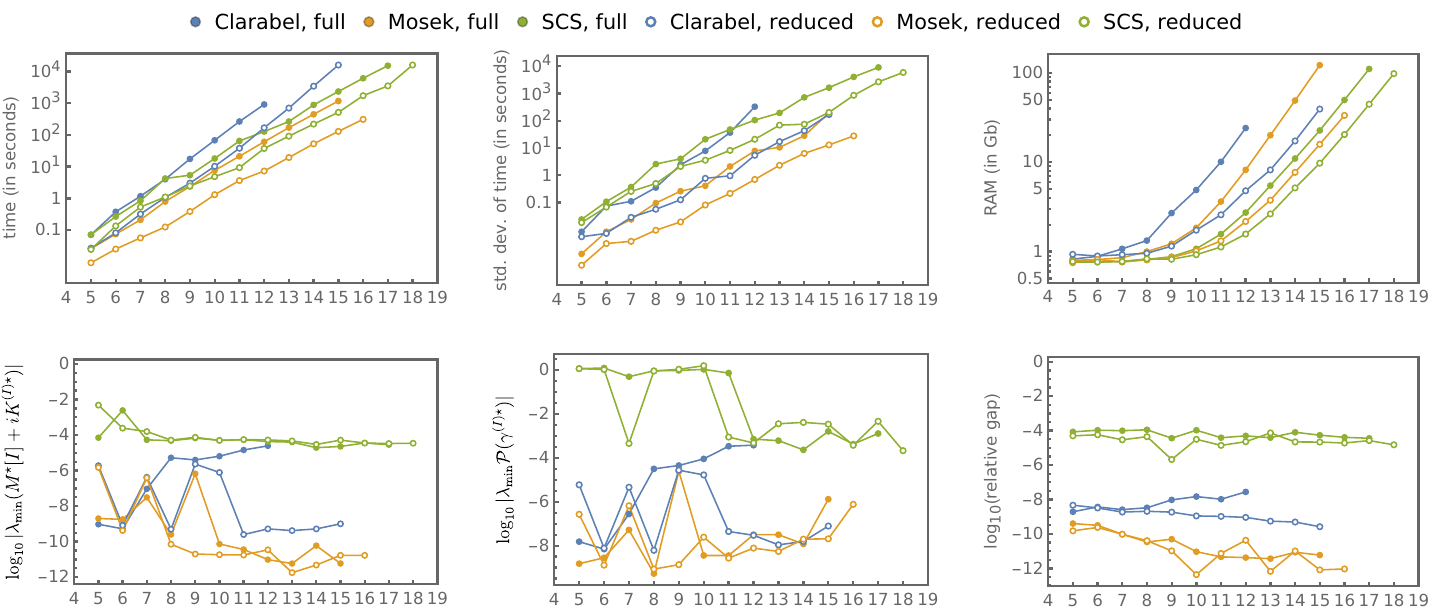}
	\caption{Comparing performance and accuracy characteristics of different solvers for the 
	scaling entanglement test as a function of the number of modes.}
	\label{fig:maxscaling}
\end{figure*}

From the first KKT condition we see that in this case the optimal states
$\gamma^{(I)\star}$ are again the pure boundary states. It is easy to see that
the optimal dual matrix $M^\star$ satisfies the inequality
\begin{equation}\label{eq:Mwp}
	\tr\left[M^\star \bigoplus_{I \in \mathcal{I}} \gamma^{(I)}\right] \geqslant 1
\end{equation}
for arbitrary CMs $\gamma^{(I)}$ and for all $\mathcal{I} \in \mathfrak{I}$. In
fact, the inequalities
\begin{equation}
\begin{split}
	\tr[\mathcal{P}(\gamma^{(I)}) &(M^\star[I] + i K^{(I)\star})] = 
	\tr(M^\star \gamma^{(I)}) \\
	&- \frac{1}{2} \tr(K^{(I)\star}\Omega_{|I|}) \geqslant 0,
\end{split}
\end{equation}
are valid for all $I$ since $M^\star[I] + i K^{(I)} \succcurlyeq 0$. Then we
have
\begin{equation}
\begin{split}
	\tr\left[M^\star \bigoplus_{I \in \mathcal{I}} \gamma^{(I)}\right] &= 
	\sum_{I \in \mathcal{I}} \tr(M^\star[I]\gamma^{(I)}) \\
	&\geqslant \frac{1}{2} \sum_{I \in \mathcal{I}} \tr(K^{(I)\star}\Omega_{|I|}) = 1,
\end{split}
\end{equation}
where we used Eq.~\eqref{eq:KO}. It follows that any $M^\star$ satisfies the
condition $\tr(M^\star\gamma_{\mathfrak{I}}) \geqslant 1$ for the CM
$\gamma_{\mathfrak{I}}$ of any $\mathfrak{I}$-separable state. Thus, if the
optimal value
\begin{equation}
	t^\star = \tr(M^\star \gamma) < 1,
\end{equation}
then $M^\star$ witnesses that $\gamma$ cannot be $\mathfrak{I}$-separable. We
have proved that $M^\star$ satisfies the property \eqref{eq:Mwp} with the help
of the related matrices $K^{(I)\star}$ provided by the solver, but given an
$M^\star$ we can always check this property directly with the explicit
expression for the symplectic trace.

The results of testing the same CMs as in the eigenvalue case are shown in
Fig.~\ref{fig:maxscaling}. The time and RAM requirements are basically the same,
but the relative gap is much smaller, so it seems that this problem can be
solved more accurately (for these test data at least).

The optimal value $t^\star(\gamma)$ is a concave function of the CM $\gamma$,
since it is the minimum of a family of linear functions. It follows that the
quantity $\mathcal{T}_{\mathfrak{I}}(\gamma)$ defined via
\begin{equation}
	\mathcal{T}_{\mathfrak{I}}(\gamma) = 1 - t^\star(\gamma)
\end{equation}
is convex and thus can be referred to as an entanglement measure. This measure
has the same properties as $\mathcal{E}_\mathfrak{I}$ --- it is convex also with
respect to the density operator and is ordered with respect to the order on
$\mathfrak{I}$. As in the previous case, this function can be negative, so we
``fix'' it in the same way 
\begin{equation}
	\mathcal{T}^+_{\mathfrak{I}}(\gamma) = \max(\mathcal{T}_{\mathfrak{I}}(\gamma), 0).
\end{equation}
This function is nonnegative, convex (as the maximum of two convex functions)
and has the following addition property: $\mathcal{T}^+_{\mathfrak{I}}(\gamma) =
0$ iff $\gamma$ is $\mathfrak{I}$-separable.

\subsection{Tests with uncertainty}\label{sec:IVC}

When the CM $\gamma$ is known with uncertainty expressed by the matrix of
standard covariances $\sigma$ it is not enough just to find an $M$ that satisfies
the inequality \eqref{eq:MW} to demonstrate entanglement of $\gamma$. We need to
establish a stronger inequality
\begin{equation}\label{eq:MWs}
	\tr(M\gamma) < \min_{\mathcal{I} \in \mathfrak{I}} 
	\sum_{I \in \mathcal{I}} \str(M[I]) - s \sigma(M),
\end{equation}
where $s$ is a sufficiently large number, e.g. $s=3$, and 
\begin{equation}
	\sigma^2(M) = \sum^{2n}_{i,j = 1} M^2_{ij} \sigma^2_{ij}.
\end{equation}
We now formulate the problem of finding the maximal $s$ such that there is an
$M$ that satisfies the inequality \eqref{eq:MWs}. Since the inequality
\eqref{eq:MWs} is invariant under the scaling of $M$, without loss of generality
we can assume that $\sigma(M) \leqslant 1$ (we cannot use strict equality here
because in this case the optimization domain will not be convex).
\begin{thrm}
For any CM $\gamma$ we have the duality
\begin{equation}
	\max_{M \succcurlyeq 0} s = \min_{\gamma^{(I)}, Q} q,
\end{equation}
with the additional constraints of the primal problem (on the left-hand side)
given by Eq.~\eqref{eq:MWs} and $\sigma(M) \leqslant 1$, and the constraints of
the dual problem (on the right-hand side) are \eqref{eq:p1} and 
\begin{equation}
	\gamma \succcurlyeq \sum_{\mathcal{I} \in \mathfrak{I}} 
	p_{\mathcal{I}} \bigoplus_{I \in \mathcal{I}} \gamma^{(I)} + \sigma \cdot Q, \quad
	\|Q\| \leqslant q,
\end{equation}
where the norm of the symmetric matrix $Q$ is defined via
\begin{equation}
	\|Q\|^2 = \tr(Q^{\mathrm{T}}Q) = \tr(Q^2).
\end{equation}
The KKT conditions read as
\begin{equation}
\begin{split}
	\tr[M^\star (\sigma \cdot Q^\star)] + q^\star &= 0 \\
	M^\star\left(\gamma - \sum_{\mathcal{I} \in \mathfrak{I}} 
	p^\star_{\mathcal{I}} \bigoplus_{I \in \mathcal{I}} \gamma^{(I)\star} 
	- \sigma \cdot Q^\star\right) &= 0.
\end{split}
\end{equation}
In the first equality we have the standard matrix product of $M^\star$ with the
Hadamard product $\sigma \cdot Q^\star$. In addition,
\begin{equation}
	\tr(M^\star \gamma) - \sum_{I \in \mathcal{I}} \str(M^\star[I]) + q^\star = 0
\end{equation}
for all $\mathcal{I} \in \mathfrak{I}$ with $p_{\mathcal{I}} > 0$.

In the reduced case the duality reads as
\begin{equation}\label{eq:XPsigma}
	\max_{X, P \succcurlyeq 0} s = \min_{\gamma^{(I)}_{xx}, \gamma^{(I)}_{pp},	Q, R} q,
\end{equation}
with the additional constraint of the primal problem (on the left-hand side)
being 
\begin{equation}
\begin{split}
	\tr(X\gamma_{xx} + P\gamma_{pp}) &- \min_{\mathcal{I} \in \mathfrak{I}}
	\sum_{I \in \mathcal{I}}\str(X[I], P[I]) \\
	&+ s\sqrt{\sigma^2(X) + \sigma^2(P)} < 0,
\end{split}
\end{equation}
where 
\begin{equation}
	\sigma^2(X) = \sum^n_{i, j = 1} X^2_{ij}\sigma^2_{xx, ij}, \quad
	\sigma^2(P) = \sum^n_{i, j = 1} P^2_{ij}\sigma^2_{pp, ij},
\end{equation}
and $\sigma^2(X) + \sigma^2(P) \leqslant 1$, and the constraints of the dual
problem (on the right-hand side) are \eqref{eq:p1} and
\begin{equation}
\begin{split}
	\gamma_{xx} \succcurlyeq \sum_{\mathcal{I} \in \mathfrak{I}} 
	p_{\mathcal{I}} \bigoplus_{I \in \mathcal{I}} \gamma^{(I)}_{xx} + \sigma_{xx} \cdot Q \\
	\gamma_{pp} \succcurlyeq \sum_{\mathcal{I} \in \mathfrak{I}},
	p_{\mathcal{I}} \bigoplus_{I \in \mathcal{I}} \gamma^{(I)}_{pp} + \sigma_{pp} \cdot R,
\end{split}
\end{equation}
where symmetric matrices $Q$ and $R$ satisfy the inequality
\begin{equation}
	\sqrt{\|Q\|^2 + \|R\|^2} \leqslant q.
\end{equation}
The KKT conditions read as
\begin{equation}
\begin{split}
	\tr[X^\star (\sigma_{xx} \cdot Q^\star)] + \tr[P^\star (\sigma_{pp} \cdot R^\star)] + q^\star &= 0, \\
	X^\star\left(\gamma_{xx} - \sum_{\mathcal{I} \in \mathfrak{I}} 
	p^\star_{\mathcal{I}} \bigoplus_{I \in \mathcal{I}} \gamma^{(I)\star}_{xx} 
	- \sigma_{xx} \cdot Q^\star\right) &= 0, \\
	P^\star\left(\gamma_{pp} - \sum_{\mathcal{I} \in \mathfrak{I}} 
	p^\star_{\mathcal{I}} \bigoplus_{I \in \mathcal{I}} \gamma^{(I)\star}_{pp} 
	- \sigma_{pp} \cdot R^\star\right) &= 0.
\end{split}
\end{equation}
In addition,
\begin{equation}
	\tr(X^\star \gamma_{xx} + P^\star\gamma_{pp}) - 
	\sum_{I \in \mathcal{I}} \str(X^\star[I], P^\star[I]) + q^\star = 0
\end{equation}
for all $\mathcal{I} \in \mathfrak{I}$ with $p_{\mathcal{I}} > 0$.
\end{thrm}

\begin{figure}
	\includegraphics[scale=1.0]{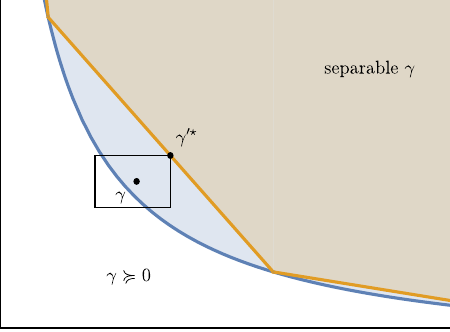}
	\caption{Closest separable state.}\label{fig:cs}
\end{figure}

\begin{figure*}
	\includegraphics[scale=0.73]{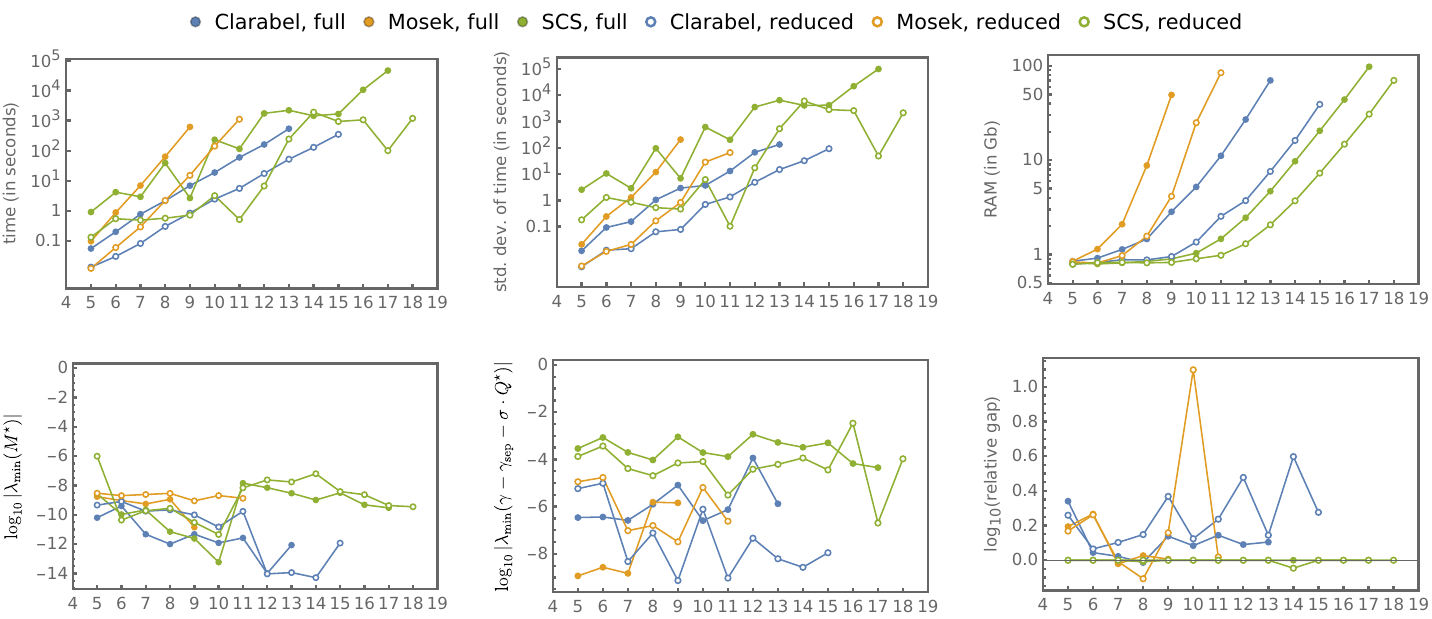}
	\caption{Comparing performance and accuracy characteristics of different solvers for the 
	problem expressed by Theorem 8 as a function of the number of modes.}
	\label{fig:maxs}
\end{figure*}

Another way to reliably test a CM $\gamma$, known with uncertainty $\sigma$, is
to determine the distance (measured in units of $\sigma$) to the closest
separable state, Fig~\ref{fig:cs}. If $\gamma$ is entangled, then we will look
for the largest $\sigma$-neighborhood $|\gamma' - \gamma| \leqslant s \sigma$
such that all its points $\gamma'$ are also entangled. An equivalent formulation
is as a smallest $\sigma$-neighborhood that contains a separable point
$\gamma'$. Such a neighborhood is illustrated in Fig.\ref{fig:cs}. The
structures in this figure are the same as those of Fig.~\ref{fig:Capprox}, and
in addition the yellow polygon denotes the set of separable states. For this
kind of test we use the scaling approach.

\begin{thrm}
For any CM $\gamma$ and the matrix of standard deviations $\sigma$ the following
duality is valid:
\begin{displaymath}
	\min_{\gamma', \gamma^{(I)}} s = 
	\max_{\Lambda, W \succcurlyeq 0, U^{\pm} \geqslant 0, K^{(I)}, \nu}  
	-\tr[\Lambda \mathcal{P}(\gamma)] - \tr(W\gamma) + \nu, 
\end{displaymath}
with the additional constrains for the primal problem (on the left-hand side)
being \eqref{eq:p1} and
\begin{displaymath}
	|\gamma' - \gamma| \leqslant s \sigma, \quad \gamma' \geqslant 
	\sum_{\mathcal{I} = \{I_1, \ldots, I_l\} \in \mathfrak{I}} 
	p_{\mathcal{I}} \gamma^{(I_1)} \oplus \ldots \oplus \gamma^{(I_l)}.
\end{displaymath}
The variables of the dual problem (on the right-hand side) are a Hermitian
matrix $\Lambda$, a real symmetric $W$, real anti-symmetric $K^{(I)}$, real
$U^{\pm}$ and a real number $\nu$. The additional constraints are $W[I] +
iK^{(I)} \succcurlyeq 0$ for all $I \in \mathcal{I}$, $\mathcal{I} \in
\mathfrak{I}$, and
\begin{equation}
\begin{split}
	U^+ - U^- + \re(\Lambda) + W &= 0, \\
	\tr[(U^+ + U^-)\sigma] &= 1, \\
	\frac{1}{2} \sum_{I \in \mathcal{I}} \tr(K^{(I)}\Omega_{|I|}) &= \nu,
\end{split}
\end{equation}
for all $\mathcal{I} \in \mathfrak{I}$. The KKT conditions read as
\begin{equation}
\begin{split}
	\mathcal{P}(\gamma^{\prime\star}) \Lambda^\star &= 0, \\
	W^\star\left(\gamma^{\prime\star} - \sum_{\mathcal{I} \in \mathfrak{I}} 
	p^\star_{\mathcal{I}} \bigoplus_{I \in \mathcal{I}} \gamma^{(I)\star}\right) &= 0, \\
	\tr[U^{\pm\star}(\gamma - \gamma^{\prime\star} \mp s^\star \sigma)] &= 0,
\end{split}
\end{equation}
where $s^\star$ is the optimal prime objective value, and
\begin{equation}
	\mathcal{P}(\gamma^{(I)\star})(W^\star[I] + i K^{(I)\star}) = 0
\end{equation}
for all $I \in \mathcal{I}$ and $\mathcal{I} \in \mathfrak{I}$ with
$p_{\mathcal{I}} > 0$.

In the reduced case the duality reads as
\begin{equation}
\begin{split}
	\min_{\substack{\gamma'_{xx}, \gamma'_{pp}, \\ \gamma^{(I)}_{xx}, \gamma^{(I)}_{pp}}} s &= 
	\max_{\substack{W, X, P \succcurlyeq 0, Z^{(I)} \\ U^{\pm}, V^{\pm} \geqslant 0, \nu}} 
	- \tr[W \mathcal{P}(\gamma_{xx}, \gamma_{pp})] \\
	&- \tr(X\gamma_{xx} + P\gamma_{pp}) + \nu, 
\end{split}
\end{equation}
with the additional constraints of the primal problem (on the left-hand side)
being \eqref{eq:p1} and
\begin{equation}
\begin{split}
	\gamma'_{xx} \geqslant \sum_{\mathcal{I} \in \mathfrak{I}} 
	p_{\mathcal{I}} \bigoplus_{I \in \mathcal{I}} \gamma^{(I)}_{xx}, \quad
	\gamma'_{pp} \geqslant \sum_{\mathcal{I} \in \mathfrak{I}} 
	p_{\mathcal{I}} \bigoplus_{I \in \mathcal{I}} \gamma^{(I)}_{pp},
\end{split}
\end{equation}
together with $|\gamma'_{xx} - \gamma_{xx}| \leqslant s \sigma_{xx}$, $|\gamma'_{pp} -
\gamma_{pp}| \leqslant s \sigma_{pp}$. The additional constraints of the dual
problem (on the right-hand side) are
\begin{equation}
\begin{split}
	U^+ - U^- + X + W_{xx} &= 0, \\
	V^+ - V^- + P + W_{pp} &= 0, \\
	\tr[(U^+ + U^-)\sigma_{xx} + (V^+ + V^-)\sigma_{pp}] &= 1, \\
	\sum_{I \in \mathcal{I}} \tr Z^{(I)} &= \nu,
\end{split}
\end{equation}
and, in addition,
\begin{equation}
	\begin{pmatrix}
		X[I] & Z^{(I)} \\
		Z^{(I)\mathrm{T}} & P[I]
	\end{pmatrix}
	\succcurlyeq 0
\end{equation}
for all $I \in \mathcal{I}$, $\mathcal{I} \in \mathfrak{I}$. The KKT conditions
are given by
\begin{equation}
\begin{split}
	\mathcal{P}(\gamma^{\prime\star}_{xx}, \gamma^{\prime\star}_{pp}) W^\star &= 0, \\
	X^\star\left(\gamma^{\prime\star}_{xx} - \sum_{\mathcal{I} \in \mathfrak{I}}
	p^\star_{\mathcal{I}} \bigoplus_{I \in \mathcal{I}} \gamma^{(I)\star}_{xx} \right) &= 0, \\
	P^\star\left(\gamma^{\prime\star}_{pp} - \sum_{\mathcal{I} \in \mathfrak{I}}
	p^\star_{\mathcal{I}} \bigoplus_{I \in \mathcal{I}} \gamma^{(I)\star}_{pp} \right) &= 0, \\
	\tr[U^{\pm\star}(\gamma_{xx} - \gamma^{\prime\star}_{xx} \mp s^\star \sigma_{xx})] &= 0, \\
	\tr[V^{\pm\star}(\gamma_{pp} - \gamma^{\prime\star}_{pp} \mp s^\star \sigma_{pp})] &= 0,
\end{split}
\end{equation}
and Eq.~\eqref{eq:KKT-1}.
\end{thrm}

Figs.~\ref{fig:maxs} and \ref{fig:maxsscaling} show the results of testing
randomly generated covariance matrices and matrices of standard deviations
according to the optimization problems given by Theorems 8 and 9, respectively.
This time the randomly generated CMs turned out not to be genuine multipartite
entangled, so the primal optimal values happen to be very close to zero. The
relative gap is thus a ratio of two nearly-zero quantities and in this case it
is not a very informative accuracy measure. It nevertheless shows that some
solvers are more accurate then the others.

The optimal value $s^\star$ produced by the these optimization problems gives a
quantitative measure of how reliable the conclusion is with the provided
experimental data. The value $s^\star \geqslant 3$ should be considered as a
reliable entanglement verification. The problem can be when the two tests,
additive and multiplicative, give significantly different results, an example of
which will be shown below. As with the best physical covariance matrix which
agrees with the measured non-physical one, there is no right approach. But from
our point of view, the multiplicative test produces a more ``correct'' result.

\begin{figure*}
	\includegraphics[scale=0.73]{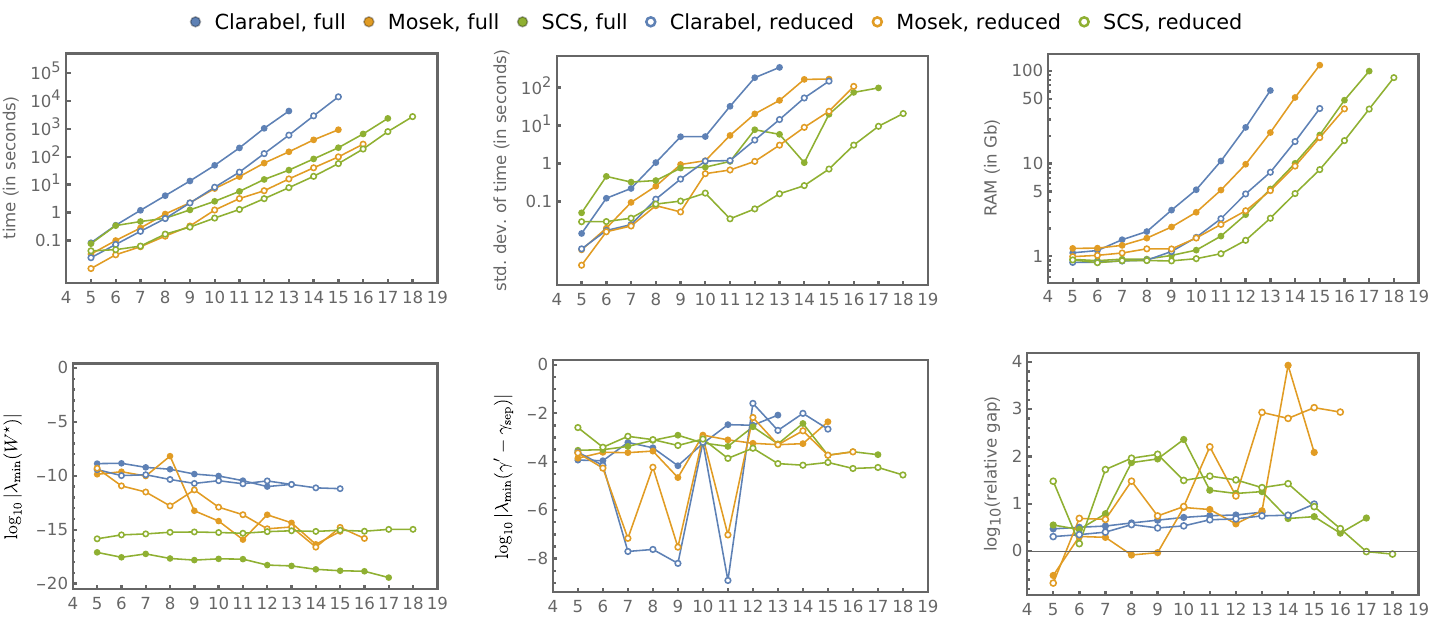}
	\caption{Comparing performance and accuracy characteristics of different solvers for the 
	problem expressed by Theorem 9 as a function of the number of modes.}
	\label{fig:maxsscaling}
\end{figure*}

\section{Refined condition}\label{sec:VIII}

In Ref.~\cite{NewJPhys.20.023030} the classical $\mathcal{I}$-separability
condition \cite{PhysRevLett.86.3658} given by Eq.~\eqref{eq:gAB}, where
$\mathcal{I} = \{I_1, \ldots, I_l\}$ is a partition of $n$-modes, was refined as
the inequality
\begin{equation}
	\gamma_{\mathcal{I}} \succcurlyeq \gamma^{(I_1)} \oplus \ldots \oplus \gamma^{(I_{l-1})}
	\oplus \frac{i}{2}\Omega_{|I_l|},
\end{equation}
that is satisfied by any $\mathcal{I}$-separable CM $\gamma_{\mathcal{I}}$. It
follows that if a CM $\gamma$ is $\mathfrak{I}$-separable, where $\mathfrak{I} =
\{\mathcal{I}_1, \ldots, \mathcal{I}_k\}$ is a collection of partitions, then
the inequality
\begin{equation}\label{eq:gammatestt3}
	\gamma \geqslant \sum_{\mathcal{I} = \{I_1, \ldots, I_l\} \in \mathfrak{I}} 
	p_{\mathcal{I}} \gamma^{(I_1)} \oplus \ldots \oplus \gamma^{(I_{l-1})} 
	\oplus \frac{i}{2}\Omega_{|I_l|},
\end{equation}
is satisfied for some $p_{\mathcal{I}}$ and $\gamma^{(I)}$. Unfortunately, this
inequality cannot be used to construct an eigenvalue-form optimization problem:
since $\tr(MK) = 0$ for any symmetric matrix $M$ and any anti-symmetric $K$ (of
compatible sizes), the information about the term $(i/2)\Omega$ on the
right-hand side is lost.

To transform this inequality into the scaling-form optimization problem, we need
the following simple statement: if $A \succcurlyeq t B$ for some $t > 0$, where
$A$ and $B$ are Hermitian matrices and $A \succcurlyeq 0$, then $A \succcurlyeq
t' B$ for all $0 \leqslant t' \leqslant t$. If $B \succcurlyeq 0$, as it is in
the classical entanglement condition, then the proof is obvious
\begin{equation}
	A - t'B = (A-tB) + (t-t')B \succcurlyeq 0,
\end{equation} 
since each term is positive semidefinite. In the new case $B$ has negative
eigenvalues due to the term $(i/2)\Omega$, so this reasoning is not applicable
here. We need a different approach,
\begin{equation}
	A - t' B = (1 - t'/t)A + (t'/t)(A - t B) \succcurlyeq 0,
\end{equation}
since both terms are positive semidefinite. We thus can look for the maximal
number $t$ such that 
\begin{equation}\label{eq:gammatestt2}
	\gamma \geqslant t\sum_{\mathcal{I} = \{I_1, \ldots, I_l\} \in \mathfrak{I}} 
	p_{\mathcal{I}} \gamma^{(I_1)} \oplus \ldots \oplus \gamma^{(I_{l-1})} 
	\oplus \frac{i}{2}\Omega_{|I_l|}.
\end{equation}
For any $\mathfrak{I}$-separable state the maximal $t^\star$ must satisfy the
inequality $t^\star \geqslant 1$. If we determine that $t^\star < 1$, then
$\gamma$ is guaranteed to be $\mathfrak{I}$-entangled. The task of finding the
maximal $t$ can be formulated as the following optimization problem.
\begin{thrm}
	For a CM $\gamma$ the following equality is valid:
	\begin{equation}
		\max_{\gamma^{(I)}, p_{\mathcal{I}}} t = 
		\min_{M + i K \succcurlyeq 0} \tr(M\gamma),
	\end{equation}
	where $\Lambda  = M + i K$ is a Hermitian matrix, with the additional
	constraint for the primal problem (on the left-hand side) is given by
	Eqs.~\eqref{eq:gammatestt2} and \eqref{eq:p1}, and the additional constraint
	for the dual problem (on the right-hand side) reads as $M[I] + i K^{(I)}
	\succcurlyeq 0$ and
	\begin{equation}\label{eq:KO1}
		\frac{1}{2}\sum_{I \in \mathcal{I}'} \tr(K^{(I)} \Omega_{|I|}) + 
		\frac{1}{2} \tr(K^{(I_l)} \Omega_{|I_l|}) = 1,
	\end{equation}
	for all $I \in \mathcal{I}'$, where $\mathcal{I}' = \{I_1, \ldots,
	I_{l-1}\}$ for $\mathcal{I} = \{I_1, \ldots, I_l\} \in \mathfrak{I}$. The
	KKT conditions are 
	\begin{equation}\label{eq:PK}
		\mathcal{P}(\gamma^{(I)\star})(M^\star[I] + i K^{(I)\star}) = 0, 
	\end{equation}
	for all $\mathcal{I} \in \mathfrak{I}$ and $I \in \mathcal{I}'$ with
$p^\star_{\mathcal{I}} > 0$ and
	\begin{equation}\label{eq:Lambdagamma}
		\overline{\Lambda^{\star}}\left(\gamma - t^\star \sum_{\mathcal{I} \in \mathfrak{I}}
		p^\star_{\mathcal{I}} \gamma^{(I_1)} \oplus \ldots \oplus \gamma^{(I_{l-1})} 
		\oplus \frac{i}{2}\Omega_{|I_l|} \right) = 0,
	\end{equation}
	where $t^\star = \tr(M^\star \gamma)$ is the primal optimal objective value
	of the primal objective function. 
\end{thrm}

As in the classical case \cite{PhysRevLett.86.3658}, the optimal value $t^\star$
is a concave function of $\gamma$, so the difference $1-t^\star$ is convex and
thus could be also called an entanglement measure. But now this value might
depend on how we arrange the parts inside the partitions. When we have a single
partition, $|\mathfrak{I}| = 1$, and the only partition $\mathcal{I} = \{I, J\}$
is a bipartition, we numerically verified that the optimal $t^\star$ value is
the same as in the classical case, irrespective of whether we replace the
variable $\gamma^{(I)}$ or $\gamma^{(J)}$ by a corresponding constant
$(i/2)\Omega$. From a computation efficiency perspective we should discard the
larger component of the bipartition. For example, when $n$ is large and $|I|=1$
and $|J|=n-1$, the difference in the performance of the optimization will be
significant.

\begin{figure*}
	\includegraphics[scale=0.73]{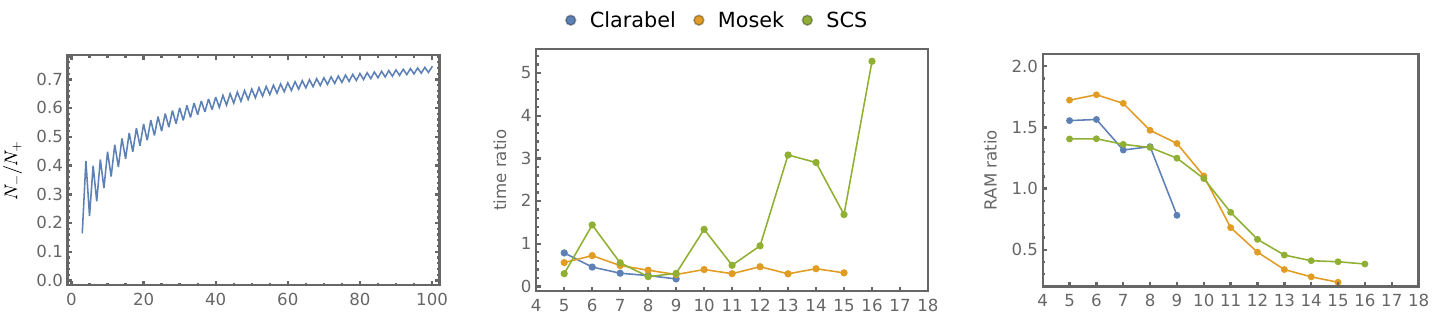}
	\caption{The ratio $N_-/N_+$ (left), the ratio of times
	to solve the new and the standard problems (middle) and the analogous ratio
	of RAMs (right) as a function of the number of modes.} \label{fig:N}
\end{figure*}

When there is more than one partition, $|\mathfrak{J}| > 1$, then we can discard
arbitrary components from $\mathcal{I} \in \mathfrak{I}$ individually. The most
optimal choice is when we discard the largest component. But the advantage will
diminish with $n$. We are most interested in genuine multipartite entanglement,
when $\mathfrak{I} = \mathfrak{I}_2$ is the set of all possible bipartitions of
$n$ modes. Let us compute the number of $K$ variables when we discard the larger
components (the most efficient way) with the same number when we discard the
smaller ones (the least efficient way). The former number is given by
\begin{equation}
	N_- = \sum^{\lfloor \frac{n}{2} \rfloor}_{k=1} \binom{n}{k} k(2k-1),
\end{equation}
and the latter is
\begin{equation}
	N_+ = \sum^{\lfloor \frac{n}{2} \rfloor}_{k=1} \binom{n}{k} (n-k)(2n-2k-1).
\end{equation}
These sums can be computed analytically:
\begin{equation}
	N_- = 
	\begin{cases}
		n^2 2^{n-2} - n(n-1) \binom{n-2}{\frac{n}{2}-1} & n\ \text{is even} \\
		n^2 2^{n-2} - n(2n-1) \binom{n-2}{\frac{n-1}{2}} & n\ \text{is odd} 
	\end{cases}
\end{equation}
and 
\begin{displaymath}
	N_+ = 
	\begin{cases}
		n^2 2^{n-2} + (3n-2)(n-1) \binom{n-2}{\frac{n}{2}-1} -n(2n-1) \\
		n^2 2^{n-2} + n(2n-1) \binom{n-2}{\frac{n-1}{2}}-n(2n-1)
	\end{cases}
\end{displaymath}
for even and odd $n$, respectively. As Fig.~\ref{fig:N} (left) shows, the ratio
$N_-/N_+$ grows and it can be shown that in the limit $n \to +\infty$ it tends
to 1. It means that for large values of $n$ it is irrelevant which components to
discard, but these large values of $n$ are currently intractable to deal with.
For the values considered in this work, $n < 20$, the new condition gives a
quite noticeable saving in RAM and also in time (but not always). Since the
refined condition is not necessary and sufficient when more than one partition
is involved, we do not study it further and do not specialize it for the reduced
case.

\section{Applications}\label{sec:V}

Here we first apply our methods to some experimental data and show the advantage
of our approach over another one, Ref.~\cite{PhysRevLett.114.050501}. That other
approach is based on genetic programming, a heuristic method that usually
produces suboptimal results. Our methods are based on convex optimization and
thus guaranteed to produce the optimal result with a certificate of optimality.
The computational resources needed to perform the genetic-optimization based
algorithms are not reported in that other work, so we cannot compare ours with
theirs. We also show that in some cases our optimization problems admit
analytical solutions.

\subsection{Experimental data}\label{sec:VIIA}

\begin{figure*}
	\includegraphics{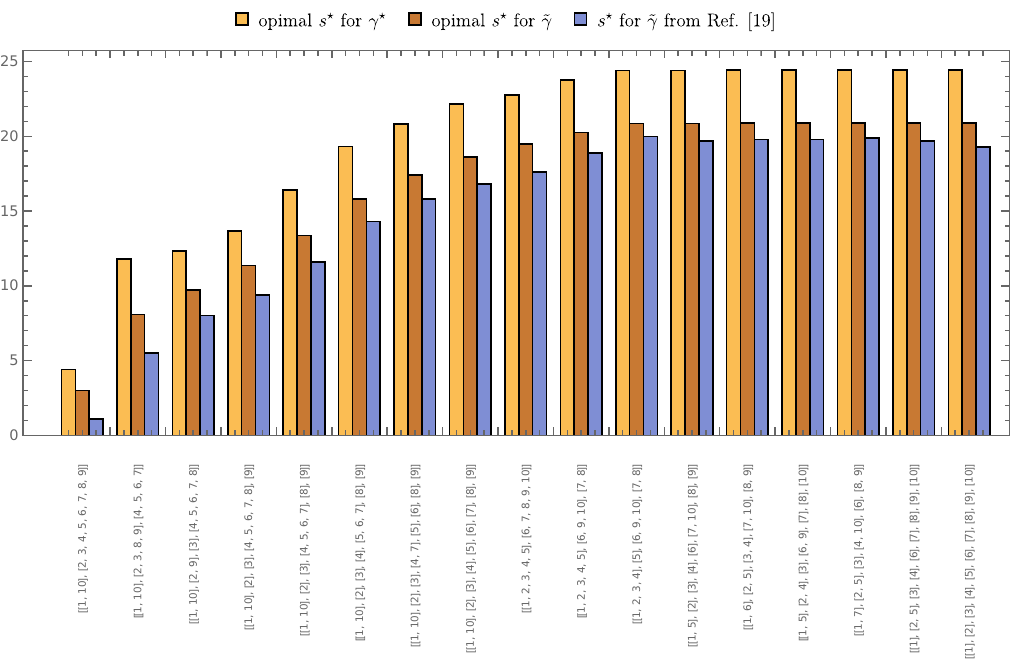}
	\caption{The optimal $s^\star$ of the optimization problem \eqref{eq:XPsigma} for the $10-$mode
	state of Ref.~\cite{PhysRevLett.114.050501}.}\label{fig:10}
\end{figure*}

To apply our approach to experimental measurements, we solve our optimization
problems for the data presented in Ref.~\cite{PhysRevLett.114.050501}. In that
work the physical CM was simply taken to be 
\begin{equation}
	\tilde{\gamma} = \gamma^\circ + 1.001 \lambda_{\mathrm{min}}(\mathcal{P}(\gamma^\circ))E_{2n}.
\end{equation}
The coefficient $1.001$ is just a number, slightly above $1$, to make the
minimal eigenvalue of $\mathcal{P}(\tilde{\gamma})$ slightly positive. For the
$4-$mode state some diagonal elements of $\tilde{\gamma}$ are more than
$10\sigma$ away from the measured averages, so $\tilde{\gamma}$ cannot be the
true physical CM. Therefore we choose not to fully explore this particular case.

Nevertheless, we apply the optimization problem of Theorem 6 to $\tilde{\gamma}$
for a few partitions of parts and compare the maximal $s^\star$ values with the
values obtained in Ref.~\cite{PhysRevLett.114.050501}. For the $4-$ and $6-$mode
states our results are pretty close, but for the $10-$mode state the advantage
of our approach can be seen in Fig.~\ref{fig:10} (orange vs blue). For
comparison we also show the optimal $s^\star$ for the most likely physical CM
$\gamma^\star$ given by Theorem 2 (yellow). Any value $s^\star > 3$ should be
enough to verify inseparability. As the first partition in Fig.~\ref{fig:10}
shows, our method clearly verifies entanglement with respect to the partition
$1,10|23456789$, while that of Ref.~\cite{PhysRevLett.114.050501} does not. The
optimal values for the most probable physical covariance matrix are even higher.

\begin{figure}
	\includegraphics{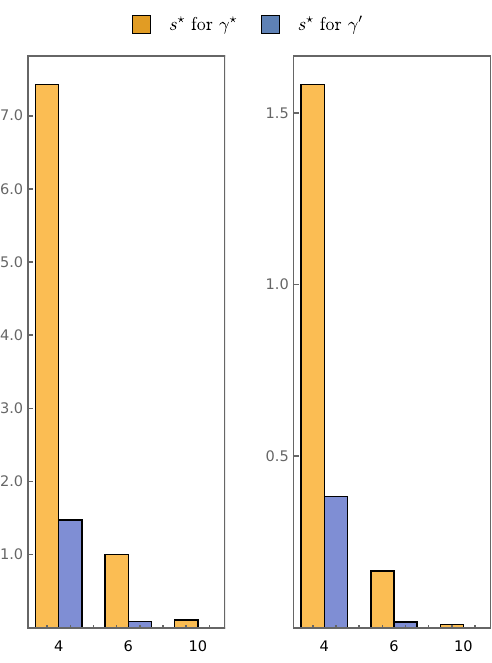}
	\caption{Optimal $s^\star$ for genuine multipartite entanglement test.}\label{fig:fig2}
\end{figure}

Testing individual partitions is not very useful if we can establish genuine
multipartite entanglement. This problem was not addressed in
Ref.~\cite{PhysRevLett.114.050501}, so we can compare our own results only. We
compare $s^\star$ for the naive physical CM $\tilde{\gamma}$ and the most
probable CM $\gamma^\star$. The results are shown in Fig.~\ref{fig:fig2}. The
left figure shows the solution given by Theorem 8 and the right figure shows the
solution given by Theorem 9. The states $\gamma^\star$ show larger values of
$s^\star$ than that of $\tilde{\gamma}$ and only the 4-partite state is verified
to be genuine multipartite entangled. The optimal solution for the 10-partite
state $\tilde{\gamma}$ is zero, which means this test is inconclusive in this
case. Here we see that $s^\star$ given by Theorems 8 and 9 can be significantly
different. In this case the latter seems to be a better measure of reliability
of the two, since it gives the largest $\sigma$-neighborhood that consists only
of entangled CMs. The former only measures the reliability of entanglement
verification of the given state.

\subsection{Bound entanglement}\label{sec:VIIB}

Here we show that in some cases the problems given by Theorems
\ref{thrm:additive} and \ref{thrm:scaling} can be solved analytically. The
method to solve them analytically relies on their numerical solution. Some
values in the solution can be identified, some can be seen to be equal (for
example, due to some symmetry in the state) or satisfy simple relations like
differing by a factor of 2. Any such observations obtained from the numerical
solution allow one to formulate the conjecture about the structure of the
analytical solution. This conjecture is then transformed to a system of
equations for the unknown parameters of the analytical solution and, if our
guess is correct, this system produces one or several analytical solutions. One
solution numerically agrees with the result produced by the solver, and this
solution is then the exact analytical solution of the problem. If our guess is
wrong, then the system would have no solution which agrees with the numerical
result or, more likely, it would have no solutions at all. The more relations we
can correctly guess from the numerical solution, the more likely that the
resulting system can be solved analytically. All the relations guessed from a
numerical solution form an overspecified system of equations, i.e. the system
with a larger number of equations than the number of variables. This
overspecification helps to verify that our conjecture on the analytical
structure of the solution is correct --- if such an overspecified system has a
solution, it is a strong indicator that our conjecture is correct.

We start to demonstrate this method on a well-known bound entangled state from
Ref.~\cite{PhysRevLett.86.3658} given by
\begin{equation}\label{eq:BBS}
	\gamma_{xx} = \frac{1}{2}
	\begin{pmatrix}
		2 & 0 & 1 & 0 \\
		0 & 2 & 0 & -1 \\
		1 & 0 & 2 & 0 \\
		0 & -1 & 0 & 2
	\end{pmatrix}, \
	\gamma_{pp} = \frac{1}{2}
	\begin{pmatrix}
		1 & 0 & 0 & -1 \\
		0 & 1 & -1 & 0 \\
		0 & -1 & 4 & 0 \\
		-1 & 0 & 0 & 4
	\end{pmatrix}.
\end{equation}
According to the condition \eqref{eq:pureReduced}, the Gaussian state with this
CM is a mixed state. It is easy to check that partial transposition of any two
modes is positive-semidefinite. We now apply Theorem~\ref{thrm:additive} and
show that the Gaussian state with this CM is $13|24$- and $14|23$-separable, but
$12|34$-entangled. Following the intentions of the authors of
Ref.~\cite{PhysRevLett.86.3658} to provide a fully analytical example without
relying on any numerical approximations, we present an analytical solution of
two of our optimization problems. The analytical solution is constructed from
the numerical one produced by the solver by guessing the exact form of the
solution and then verifying the conjecture by constructing the dual solution and
checking that they agree. The matrix elements of the optimal matrices obtained
numerically we denote without prime, and we use prime to denote their exact
analytical expressions.

\subsubsection{The $13|24$ partition, eigenvalue}

In this case the matrices $X^\star_{13|24}$ and $P^\star_{13|24}$ are easy to
recognize and they are given by
\begin{equation}\label{eq:XP1324}
\begin{split}
	X^\star_{13|24} &= \frac{1}{36}
	\begin{pmatrix}
		5 & 0 & -3 & 0 \\
		0 & 5 & 0 & 3 \\
		-3 & 0 & 2 & 0 \\
		0 & 3 & 0 & 2
	\end{pmatrix}, \\
	P^\star_{13|24} &= \frac{1}{36}
	\begin{pmatrix}
		10 & 0 & -1 & 3 \\
		0 & 10 & 3 & 1 \\
		-1 & 3 & 1 & 0 \\
		3 & 1 & 0 & 1
	\end{pmatrix}.
\end{split}
\end{equation}
Note that $\rank(X\star_{13|24}) = 4$, but $\rank(P\star_{13|24}) = 2$. The
matrices $Z^{13\star}$ and $Z^{24\star}$ are equally easy to determine as well,
\begin{equation}\label{eq:Z1324}
	Z^{13\star} = \frac{1}{36}
	\begin{pmatrix}
		7 & -1 \\
		-4 & 1
	\end{pmatrix}, \quad
	Z^{24\star} = \frac{1}{36}
	\begin{pmatrix}
		7 & 1 \\
		4 & 1
	\end{pmatrix}.
\end{equation}
The matrices satisfy the positivity condition $X^\star_{13|24}, P^\star_{13|24}
\succcurlyeq 0$ and the normalization condition
\begin{equation}
	\tr(X^\star_{13|24} + P^\star_{13|24}) = 1.
\end{equation}
In addition, we have that the matrices
\begin{displaymath}
\begin{split}
	\begin{pmatrix}
		X^\star_{13|24}[1, 3] & Z^{13\star} \\
		Z^{13\star\mathrm{T}} & P^\star_{13|24}[1, 3]
	\end{pmatrix} &= \frac{1}{36}
	\begin{pmatrix}
		5 & -3 & 7 & -1 \\
		-3 & 2 & -4 & 1 \\
		7 & -4 & 10 & -1 \\
		-1 & 1 & -1 & 1
	\end{pmatrix} \\
    \begin{pmatrix}
		X^\star_{13|24}[2, 4] & Z^{24\star} \\
		Z^{24\star\mathrm{T}} & P^\star_{13|24}[2, 4]
	\end{pmatrix} &= \frac{1}{36}
	\begin{pmatrix}
		5 & 3 & 7 & 1 \\
		3 & 2 & 4 & 1 \\
		7 & 4 & 10 & 1 \\
		1 & 1 & 1 & 1
	\end{pmatrix}
\end{split}
\end{displaymath}
are positive-semidefinite and satisfy the equalities
\begin{equation}
\begin{split}
	\str(X^\star_{13|24}[1, 3], P^\star_{13|24}[1, 3]) = \tr Z^{13\star}, \\
	\str(X^\star_{13|24}[2, 4], P^\star_{13|24}[2, 4]) = \tr Z^{24\star}.
\end{split}
\end{equation}
The primal objective reads as
\begin{equation}
	\tr(X^\star_{13|24}\gamma_{xx} + P^\star_{13|24}\gamma_{pp}) 
	- \tr Z^{13\star} - \tr Z^{24\star} = 0.
\end{equation}
To show that this value is really optimal, we need the dual solution.

The dual solution can also be recognized from the numbers produced by the solver
\begin{equation}\label{eq:gamma1324}
	\begin{alignedat}{2}
		&\gamma^{13\star}_{xx} = \frac{1}{2}
		\begin{pmatrix}
			2 & 1 \\
			1 & 2
		\end{pmatrix}, \quad
		&&\gamma^{13\star}_{pp} = \frac{1}{6}
		\begin{pmatrix}
			2 & -1 \\
			-1 & 2
		\end{pmatrix}, \\
		&\gamma^{24\star}_{xx} = \frac{1}{2}
		\begin{pmatrix}
			2 & -1 \\
			-1 & 2
		\end{pmatrix}, \quad
		&&\gamma^{24\star}_{pp} = \frac{1}{6}
		\begin{pmatrix}
			2 & 1 \\
			1 & 2
		\end{pmatrix}.
	\end{alignedat}
\end{equation}
It is easy to see that $\gamma_{xx} = \gamma^{13\star}_{xx} \oplus \gamma^{24\star}_{xx}$ and
\begin{equation}
	\gamma_{pp} \succcurlyeq
	\gamma^{13\star}_{pp} \oplus \gamma^{24\star}_{pp} = \frac{1}{6}
	\begin{pmatrix}
		2 & 0 & -1 & 0 \\
		0 & 2 & 0 & 1 \\
		-1 & 0 & 2 & 0 \\
		0 & 1 & 0 & 2
	\end{pmatrix}.
\end{equation}
It follows that the dual objective value, the minimal eigenvalue of the
differences, is also zero, which verifies the optimality of the solutions. One
can verify that the KKT conditions \eqref{eq:MMKKT} are satisfied
\begin{equation}\label{eq:KKT1324}
\begin{split}
	\mathcal{P}(\gamma^{13\star}_{xx}, \gamma^{13\star}_{pp})
	\begin{pmatrix}
		X^\star_{13|24}[1, 3] & -Z^{13\star} \\
		-Z^{13\star\mathrm{T}} & P^\star_{13|24}[1, 3]
	\end{pmatrix}
	&= 0, \\
	\mathcal{P}(\gamma^{24\star}_{xx}, \gamma^{24\star}_{pp})
	\begin{pmatrix}
		X^\star_{13|24}[2, 4] & -Z^{24\star} \\
		-Z^{24\star\mathrm{T}} & P^\star_{13|24}[2, 4]
	\end{pmatrix}
	&= 0, \\
	X^\star_{13|24}(\gamma_{xx} - \gamma^{13\star}_{xx} \oplus \gamma^{24\star}_{xx}) &= 0, \\
	P^\star_{13|24}(\gamma_{pp} - \gamma^{13\star}_{pp} \oplus \gamma^{24\star}_{pp}) &= 0.
\end{split}
\end{equation}
In terms of the entanglement measure this result can be simply stated as
\begin{equation}
	\mathcal{E}^+_{13|24}(\gamma_{xx}, \gamma_{pp}) = 0,
\end{equation}
and thus the CM given by Eq.~\eqref{eq:BBS} is $13|24$-separable.

\subsubsection{The $13|24$ partition, scaling}

The primal solution $\gamma^{13\star}_{xx}$, $\gamma^{13\star}_{pp}$,
$\gamma^{24\star}_{xx}$ and $\gamma^{24\star}_{pp}$ is given by
Eq.~\eqref{eq:gamma1324} and corresponds to the optimal value $t^\star = 1$. The
dual solution, the matrices $X^\star_{13|24}$ and $P^\star_{13|24}$ are given by
\begin{equation}
	\begin{split}
		X^\star_{13|24} = \frac{1}{16}
		\begin{pmatrix}
			5 & 0 & -3 & 0 \\
			0 & 5 & 0 & 3 \\
			-3 & 0 & 2 & 0 \\
			0 & 3 & 0 & 2
		\end{pmatrix}, \\
		P^\star_{13|24} = \frac{1}{16}
		\begin{pmatrix}
			10 & 0 & -1 & 3 \\
			0 & 10 & 3 & 1 \\
			-1 & 3 & 1 & 0 \\
			3 & 1 & 0 & 1
		\end{pmatrix}.
	\end{split}
\end{equation}
The matrices $Z^{13\star}$ and $Z^{24\star}$ read as
\begin{equation}
		Z^{13\star} = \frac{1}{16}
		\begin{pmatrix}
			7 & -1 \\
			-4 & 1
		\end{pmatrix}, \quad
		Z^{24\star} = \frac{1}{16}
		\begin{pmatrix}
			7 & 1 \\
			4 & 1
		\end{pmatrix}.
\end{equation}
These matrices just differ by a factor from the matrices \eqref{eq:XP1324} and
\eqref{eq:Z1324} of the eigenvalue version of the test. One can easily check
that 
\begin{displaymath}
	\tr Z^{13\star} + \tr Z^{24\star} = 1, 
	\quad \tr(X^\star_{13|24}\gamma_{xx} + P^\star_{13|24}\gamma_{pp}) = t^\star.
\end{displaymath}
The KKT conditions \eqref{eq:MMKKT} and \eqref{eq:KKTscaling2}, which in
this case read as \eqref{eq:KKT1324}, are obviously satisfied. In terms of the
measure $\mathcal{T}^+_{13|24}$ we have
\begin{equation}
	\mathcal{T}^+_{13|24}(\gamma_{xx}, \gamma_{pp}) = 0,
\end{equation}
which gives another demonstration that the state is $13|24$-separable.

\subsubsection{The $14|23$ partition, eigenvalue}

The matrices $X^\star_{14|23}$ and $P^\star_{14|23}$ are given by
\begin{equation}\label{eq:XP1423}
\begin{split}
	X^\star_{14|23} = \frac{1}{36}
	\begin{pmatrix}
		5 & 0 & -3 & 1 \\
		0 & 5 & 1 & 3 \\
		-3 & 1 & 2 & 0 \\
		1 & 3 & 0 & 2
	\end{pmatrix}, \\
	P^\star_{14|23} = \frac{1}{36}
	\begin{pmatrix}
		10 & 0 & 0 & 3 \\
		0 & 10 & 3 & 0 \\
		0 & 3 & 1 & 0 \\
		3 & 0 & 0 & 1
	\end{pmatrix}.
\end{split}
\end{equation}
Note that $\rank(X^\star_{14|23}) = 2$, but $\rank(P^\star_{14|23}) = 4$. The
matrices $Z^{14\star}$ and $Z^{23\star}$ read as
\begin{equation}\label{eq:Z1423}
	Z^{14\star} = \frac{1}{36}
	\begin{pmatrix}
		7 & 2 \\
		2 & 1
	\end{pmatrix}, \quad
	Z^{23\star} = \frac{1}{36}
	\begin{pmatrix}
		7 & 2 \\
		2 & 1
	\end{pmatrix}.
\end{equation}
It is easy to check that $X^\star_{14|23}, P^\star_{14|23} \succcurlyeq 0$ and
\begin{equation}
	\tr(X^\star_{14|23} + P^\star_{14|23}) = 1.
\end{equation}
In addition, we have 
\begin{displaymath}
\begin{split}
	\begin{pmatrix}
		X^\star_{14|23}[1, 4] & Z^{14\star} \\
		Z^{14\star\mathrm{T}} & P^\star_{14|23}[1, 4]
	\end{pmatrix} &= \frac{1}{36}
	\begin{pmatrix}
		5 & 1 & 7 & 2 \\
		1 & 2 & 2 & 1 \\
		7 & 2 & 10 & 3 \\
		2 & 1 & 3 & 1
	\end{pmatrix} \\
    \begin{pmatrix}
		X^\star_{14|23}[2, 3] & Z^{23\star} \\
		Z^{23\star\mathrm{T}} & P^\star_{14|23}[2, 3]
	\end{pmatrix} &= \frac{1}{36}
	\begin{pmatrix}
		5 & 1 & 7 & 2 \\
		1 & 2 & 2 & 1 \\
		7 & 2 & 10 & 3 \\
		2 & 1 & 3 & 1
	\end{pmatrix}
\end{split}
\end{displaymath}
are positive-semidefinite and
\begin{equation}
\begin{split}
	\str(X^\star_{14|23}[1, 4], P^\star_{14|23}[1, 4]) = \tr Z^{14\star}, \\
	\str(X^\star_{14|24}[2, 3], P^\star_{14|22}[2, 3]) = \tr Z^{23\star}.
\end{split}
\end{equation}
The primal objective reads as
\begin{equation}
	\tr(X^\star_{14|23}\gamma_{xx} + P^\star_{14|23}\gamma_{pp}) 
	- \tr Z^{14\star} - \tr Z^{23\star} = 0.
\end{equation}
To show that this value is really optimal, we need the dual solution.

The dual solution is given by:
\begin{equation}\label{eq:gamma1423}
\begin{alignedat}{2}
	&\gamma^{14\star}_{xx} = \frac{1}{6}
	\begin{pmatrix}
		4 & 1 \\
		1 & 1
	\end{pmatrix}, \quad
	&&\gamma^{14\star}_{pp} = \frac{1}{2}
	\begin{pmatrix}
		1 & -1 \\
		-1 & 4
	\end{pmatrix}, \\
	&\gamma^{23\star}_{xx} = \frac{1}{6}
	\begin{pmatrix}
		4 & 1 \\
		1 & 1
	\end{pmatrix}, \quad
	&&\gamma^{23\star}_{pp} = \frac{1}{2}
	\begin{pmatrix}
		1 & -1 \\
		-1 & 4
	\end{pmatrix}.
\end{alignedat}
\end{equation}
It is easy to see that $\gamma_{pp} = \gamma^{14\star}_{pp} \oplus
\gamma^{23\star}_{pp}$ and
\begin{equation}
	\gamma_{xx} \succcurlyeq
	\gamma^{14\star}_{xx} \oplus \gamma^{23\star}_{xx} = \frac{1}{6}
	\begin{pmatrix}
		4 & 0 & 0 & 1 \\
		0 & 4 & 1 & 0 \\
		0 & 1 & 1 & 0 \\
		1 & 0 & 0 & 1
	\end{pmatrix}.
\end{equation}
It follows that the dual objective value, the minimal eigenvalue of the
differences, is also zero, which verifies the optimality of the solutions. One
can verify that the KKT conditions \eqref{eq:MMKKT} are satisfied
\begin{equation}\label{eq:KKT1423}
\begin{split}
	\mathcal{P}(\gamma^{14\star}_{xx}, \gamma^{14\star}_{pp})
	\begin{pmatrix}
		X^\star_{14|23}[1, 4] & -Z^{14\star} \\
		-Z^{14\star\mathrm{T}} & P^\star_{14|23}[1, 4]
	\end{pmatrix}
	&= 0, \\
	\mathcal{P}(\gamma^{23\star}_{xx}, \gamma^{23\star}_{pp})
	\begin{pmatrix}
		X^\star_{14|23}[2, 3] & -Z^{23\star} \\
		-Z^{23\star\mathrm{T}} & P^\star_{14|23}[2, 3]
	\end{pmatrix}
	&= 0, \\
	X^\star_{14|23}(\gamma_{xx} - \gamma^{14\star}_{xx} \oplus \gamma^{23\star}_{xx}) &= 0, \\
	P^\star_{14|23}(\gamma_{pp} - \gamma^{14\star}_{pp} \oplus \gamma^{23\star}_{pp}) &= 0.
\end{split}
\end{equation}
In terms of the entanglement measure this result can be simply stated as
\begin{equation}
	\mathcal{E}^+_{14|23}(\gamma_{xx}, \gamma_{pp}) = 0,
\end{equation}
and thus the CM given by Eq.~\eqref{eq:BBS} is $14|23$-separable.

\subsubsection{The $14|23$ partition, scaling}

The primal solution $\gamma^{14\star}_{xx}$, $\gamma^{14\star}_{pp}$,
$\gamma^{23\star}_{xx}$ and $\gamma^{23\star}_{pp}$ is given by
Eq.~\eqref{eq:gamma1423} and corresponds to the optimal value $t^\star = 1$. The
dual solution, the matrices $X^\star_{14|23}$ and $P^\star_{14|23}$ are given by
\begin{equation}
	\begin{split}
		X^\star_{14|23} = \frac{1}{16}
		\begin{pmatrix}
			5 & 0 & -3 & 1 \\
			0 & 5 & 1 & 3 \\
			-3 & 1 & 2 & 0 \\
			1 & 3 & 0 & 2
		\end{pmatrix}, \\
		P^\star_{14|23} = \frac{1}{16}
		\begin{pmatrix}
			10 & 0 & 0 & 3 \\
			0 & 10 & 3 & 0 \\
			0 & 3 & 1 & 0 \\
			3 & 0 & 0 & 1
		\end{pmatrix}.
	\end{split}
\end{equation}
The matrices $Z^{14\star}$ and $Z^{23\star}$ read as
\begin{equation}
	Z^{14\star} = \frac{1}{16}
	\begin{pmatrix}
		7 & 2 \\
		2 & 1
	\end{pmatrix}, \quad
	Z^{23\star} = \frac{1}{16}
	\begin{pmatrix}
		7 & 2 \\
		2 & 1
	\end{pmatrix}.
\end{equation}
These matrices just differ by a factor from the matrices \eqref{eq:XP1423} and
\eqref{eq:Z1423} of the eigenvalue version of the test. One can easily check
that 
\begin{displaymath}
	\tr Z^{14\star} + \tr Z^{23\star} = 1, 
	\quad \tr(X^\star_{14|23}\gamma_{xx} + P^\star_{14|23}\gamma_{pp}) = t^\star.
\end{displaymath}
The KKT conditions \eqref{eq:MMKKT} and \eqref{eq:KKTscaling2}, which in
this case read as \eqref{eq:KKT1423}, are obviously satisfied. In terms of the
measure $\mathcal{T}^+_{14|23}$ we have
\begin{equation}
	\mathcal{T}^+_{14|23}(\gamma_{xx}, \gamma_{pp}) = 0,
\end{equation}
which gives another demonstration that the state is $14|23$-separable.

\subsubsection{The $12|34$ partition, eigenvalue}

The matrices $X^\star_{12|34}$ and $P^\star_{12|34}$ are given by
\begin{equation}\label{eq:BBSXP2}
	X^\star_{12|34} = 
	\begin{pmatrix}
		x & 0 & -z & 0 \\
		0 & x & 0 & z \\
		-z & 0 & y & 0 \\
		0 & z & 0 & y
	\end{pmatrix}, \quad
	P^\star_{12|34} = 
	\begin{pmatrix}
		p & 0 & 0 & r \\
		0 & p & r & 0 \\
		0 & r & q & 0 \\
		r & 0 & 0 & q
	\end{pmatrix},
\end{equation}
where the values of the parameters are given by
\begin{equation}\label{eq:xypq}
\begin{alignedat}{3}
	x &= 0.1359..., &\quad y &= 0.0706..., &\quad z &= 0.0979..., \\ 
	p &= 0.2578..., &\quad q &= 0.0356..., &\quad r &= 0.0958... .
\end{alignedat}
\end{equation}
The normalization condition is given by
\begin{equation}\label{eq:be-kkt}
	\tr(X^\star_{12|34} + P^\star_{12|34}) = 2(x + y + p + q) = 1.
\end{equation}
The matrices $Z^{12\star}$ and $Z^{34\star}$ read as
\begin{equation}
	Z^{12\star} = u
	\begin{pmatrix}
		1 & 0 \\
		0 & 1
	\end{pmatrix}, \quad
	Z^{34\star} = v
	\begin{pmatrix}
		1 & 0 \\
		0 & 1
	\end{pmatrix},
\end{equation}
where $u = 0.1872...$ and $v = 0.0501...$. This time the numbers $x$, $y$, $z$,
$p$, $q$ and $r$ are not so easy to identify. 

Both matrices
\begin{equation}
	\begin{pmatrix}
		X^\star[1, 2] & Z^{12\star} \\
		Z^{12\star\mathrm{T}} & P^\star[1, 2]
	\end{pmatrix} = 
	\begin{pmatrix}
		x & 0 & u & 0 \\
		0 & x & 0 & u \\
		u & 0 & p & 0 \\
		0 & u & 0 & p
	\end{pmatrix}
\end{equation}
and 
\begin{equation}
	\begin{pmatrix}
		X^\star[3, 4] & Z^{34\star} \\
		Z^{34\star\mathrm{T}} & P^\star[3, 4]
	\end{pmatrix} = 
	\begin{pmatrix}
		y & 0 & v & 0 \\
		0 & y & 0 & v \\
		v & 0 & q & 0 \\
		0 & v & 0 & q
	\end{pmatrix}
\end{equation}
should have two zero eigenvalues, so we have the equations
\begin{equation}\label{eq:be-kkt-1}
	x p = u^2, \quad y q = v^2.
\end{equation}
There are only two equations, not four, since the eigenvalues have multiplicity
2.

The dual solution produced by the solver is
\begin{equation}\label{eq:BBSgamma}
\begin{split}
	\gamma^{12}_{xx} = \frac{1}{2}
	\begin{pmatrix}
		a & 0 \\
		0 & a
	\end{pmatrix}, \quad
	\gamma^{12}_{pp} = \frac{1}{2}
	\begin{pmatrix}
		\frac{1}{a} & 0 \\
		0 & \frac{1}{a}
	\end{pmatrix}, \\
	\gamma^{34}_{xx} = \frac{1}{2}
	\begin{pmatrix}
		b & 0 \\
		0 & b
	\end{pmatrix}, \quad
	\gamma^{34}_{pp} = \frac{1}{2}
	\begin{pmatrix}
		\frac{1}{b} & 0 \\
		0 & \frac{1}{b}
	\end{pmatrix},
\end{split}
\end{equation}
where $a$ and $b$ read as
\begin{equation}\label{eq:BBSab}
	a = 1.37716..., \quad b = 0.71033... .
\end{equation}
The KKT conditions \eqref{eq:MMKKT}, which read as
\begin{equation}
	\begin{split}
		\mathcal{P}(\gamma^{12\star}_{xx}, \gamma^{12\star}_{pp})
		\begin{pmatrix}
			X^\star_{12|34}[1, 2] & -Z^{12\star} \\
			-Z^{12\star\mathrm{T}} & P^\star_{12|34}[1, 2]
		\end{pmatrix} 
		&= 0, \\
		\mathcal{P}(\gamma^{34\star}_{xx}, \gamma^{34\star}_{pp})
		\begin{pmatrix}
			X^\star_{12|34}[3, 4] & -Z^{34\star} \\
			-Z^{34\star\mathrm{T}} & P^\star_{12|34}[3, 4]
		\end{pmatrix} 
		&= 0,
	\end{split}
\end{equation}
produce the following system of four equations:
\begin{equation}\label{eq:be-kkt-2}
	u = a x, \quad p = a u, \quad v = b y, \quad q = b v.
\end{equation}
The KKT conditions \eqref{eq:XXKKT}, which read as
\begin{equation}
	\begin{split}
		X^\star_{12|34}(\gamma_{xx} - \gamma^{12\star}_{xx} \oplus \gamma^{34\star}_{xx}) 
		&= \lambda_{\mathrm{min}} X^\star_{12|34}, \\
		P^\star_{12|34}(\gamma_{pp} - \gamma^{12\star}_{pp} \oplus \gamma^{34\star}_{pp}) 
		&= \lambda_{\mathrm{min}} P^\star_{12|34},
	\end{split}
\end{equation}
where 
\begin{displaymath}
	\begin{split}
		\lambda_{\mathrm{min}} &= \tr(X^\star_{12|34}\gamma_{xx} + P^\star_{12|34}\gamma_{pp}) - 
		\tr(Z^{12\star} + Z^{34\star}) \\
		&= 1 - p + 2 (q - r - u - v - z),
	\end{split}
\end{displaymath}
produce the following system of eight equations:
\begin{equation}\label{eq:be-kkt-3}
\begin{split}
	4x(z+u+v-q+r) &= x(a-2p)+z, \\
	4y(z+u+v-q+r) &= y(b-2p)+z, \\
	4z(z+u+v-q+r) &= z(b-2p)+x, \\
	4z(z+u+v-q+r) &= z(a-2p)+y, \\
	4p(z+u+v-q+r) &= p(1/a+1-2p)+r, \\
	4q(z+u+v-q+r) &= q(1/b-2-2p)+r, \\
	4r(z+u+v-q+r) &= r(1/b-2-2p)+p, \\
	4r(z+u+v-q+r) &= r(1/a+1-2p)+q.
\end{split}
\end{equation}
The equations given by Eqs.~\eqref{eq:be-kkt}, \eqref{eq:be-kkt-1},
\eqref{eq:be-kkt-2} and \eqref{eq:be-kkt-3} form a system of 15 equations for 10
variables $x$, $y$, $z$, $p$, $q$, $r$, $u$, $v$, $a$ and $b$. It is easy to
derive from Eq.~\eqref{eq:be-kkt-2} that 
\begin{equation}\label{eq:uxp}
	u^\star = \sqrt{x^\star p^\star}, \quad v^\star = \sqrt{y^\star q^\star}.
\end{equation}
These equalities are equivalent to the relations
\begin{equation}
\begin{split}
	\str(X^\star_{12|34}[1, 2], P^\star_{12|34}[1, 2]) &= \tr(Z^{12\star}), \\
	\str(X^\star_{12|34}[3, 4], P^\star_{12|34}[3, 4]) &= \tr(Z^{34\star}).
\end{split}
\end{equation}
Less obvious equalities are given by
\begin{equation}\label{eq:zxy}
	z^\star = \sqrt{x^\star y^\star}, \quad r^\star = \sqrt{p^\star q^\star}.
\end{equation}
The simplest way to derive these equalities is to compute the minimal
polynomials (say, in the variable $\xi$) of the differences of the right- and
left-hand sides. The minimal polynomials turn out to be just $\xi$, which means
that these equalities hold true. The minimal polynomials of the other optimal
values are given by:
\begin{displaymath}
	\begin{split}
		541504 x^{4} - 270752 x^{3} + 38408 x^{2} - 1762 x + 25 &= 0, \\
		67688 y^{4} - 33844 y^{3} + 5362 y^{2} - 347 y + 8 &= 0, \\
		33844 p^{4} - 16922 p^{3} + 2486 p^{2} - 100 p + 1 &= 0, \\
		541504 q^{4} - 270752 q^{3} + 30008 q^{2} - 778 q + 1 &= 0, \\
		5a^4 + 11a^3 - 8 a^2 - 20 a - 4 &= 0, \\
        8b^4 + 2b^3 - 26 b^2 + 16 b - 1 &= 0.
	\end{split}
\end{displaymath}
It is possible to obtain the optimal values explicitly
\begin{widetext}
	\begin{equation}
	\begin{alignedat}{2}
		x^\star &= \frac{1}{8}\left(1+x_0-\sqrt{\frac{6179}{8461} - x^2_0 + \frac{1242}{8461 x_0}}\right), \quad
		&x^2_0 &= \frac{6179+4\sqrt{3225478}
			\cos\left(\frac{1}{3}\arccos\left(\frac{22836699733}{12901912\sqrt{3225478}}\right)\right)}{25383} \\
		y^\star &= \frac{1}{8}\left(1-y_0-\sqrt{\frac{3935}{8461} - y^2_0 - \frac{1026}{8461 y_0}}\right), \quad
		&y^2_0 &= \frac{3935+16\sqrt{1093}
			\cos\left(\frac{1}{3}\arccos\left(\frac{121301}{8744\sqrt{1093}}\right)\right)}{25383}, \\
		p^\star &= \frac{1}{8}\left(1+p_0+\sqrt{\frac{5495}{8461} - p^2_0 + \frac{234}{8461 p_0}}\right), \quad 
		&p^2_0 &= \frac{5495+16\sqrt{377431}
			\cos\left(\frac{1}{3}\arccos\left(\frac{919571683}{1509724\sqrt{377431}}\right)\right)}{25383}, \\
		q^\star &= \frac{1}{8}\left(1-q_0+\sqrt{\frac{10379}{8461} - q^2_0 - \frac{3474}{8461 q_0}}\right), \quad
		&q^2_0 &= \frac{10379+4\sqrt{4297546}
			\cos\left(\frac{1}{3}\arccos\left(\frac{26612573953}{17190184\sqrt{4297546}}\right)\right)}{25383} \\
		a^\star &= \frac{1}{20}\left(-11 + a_0 + \sqrt{683 - a^2_0 + \frac{1818}{a_0}}\right), \quad
        &a^2_0 &= \frac{683}{3} + \frac{880}{3}\cos\left[\frac{1}{3}\arccos\left(\frac{3137}{5324}\right)\right] \\
		b^\star &= \frac{1}{16}\left(-1 + b_0 - \sqrt{419 - b^2_0 - \frac{2466}{b_0}}\right), \quad
		&b^2_0 &= \frac{419}{3} + \frac{352}{3}\cos\left[\frac{1}{3}\arccos\left(\frac{3137}{5324}\right)\right],
	\end{alignedat}
	\end{equation}
\end{widetext}
where $x_0$, $y_0$, $p_0$, $q_0$ and $\nu_0$ denote the positive root of the
expression of the right-hand side. Their numerical values coincide with the
values \eqref{eq:xypq} and \eqref{eq:BBSab} produced by the solver. From
Eq.~\eqref{eq:be-kkt-2} it follows that
\begin{equation}\label{eq:abid}
	a^\star = \sqrt{\frac{p^\star}{x^\star}}, \quad 
	b^\star = \sqrt{\frac{q^\star}{y^\star}},
\end{equation}
and just by looking at the explicit expressions above it seems like a very
difficult task to validate these relations. Computing the minimal polynomials of
the differences quickly verifies that these identities are valid.

We now derive the optimal solution in another way. From the equalities
\eqref{eq:uxp} and \eqref{eq:zxy} it follows that the primal objective value is
of the form
\begin{displaymath}
	\lambda_{\mathrm{min}} = 
	1 + 2q - p -2(\sqrt{x}+\sqrt{q})(\sqrt{y}+\sqrt{p}).
\end{displaymath}
We show that the numbers $x^\star$, $y^\star$, $p^\star$ and $q^\star$ are a
minimizing point of this function, i.e. they are a solution of the following
minimization problem:
\begin{equation}
	\min_{x, y, p, q} 1 + 2q - p - 2(\sqrt{x}+\sqrt{q})(\sqrt{y}+\sqrt{p}),
\end{equation}
subject to $x+y+p+q=1/2$. To verify this conjecture, we find an analytical
solution of this problem. The Lagrangian is given by
\begin{equation}
\begin{split}
	\mathcal{L} &= 1 + 2q - p - 2(\sqrt{x}+\sqrt{q})(\sqrt{y}+\sqrt{p}) \\
	&+ \nu\left(x+y+p+q-\frac{1}{2}\right).
\end{split}
\end{equation}
The KKT conditions produce a system of 5 equations for 5 variables:
\begin{displaymath}
\begin{split}
		&\frac{\sqrt{y} + \sqrt{p}}{\sqrt{x}} = \nu, \quad
		\frac{\sqrt{x} + \sqrt{q}}{\sqrt{y}} = \nu, \quad
		\frac{\sqrt{x} + \sqrt{q}}{\sqrt{p}} + 1 = \nu, \\
		&\frac{\sqrt{y} + \sqrt{p}}{\sqrt{q}} - 2 = \nu, \quad
		x + y + p + q - \frac{1}{2} = 0.
\end{split}
\end{displaymath}
To verify that the numbers $x^\star$, $y^\star$, $p^\star$, $q^\star$ do satisfy
this system, we just need to substitute them into it and check that the
resulting $\nu^\star$ is the same in the first four equations. This can be done
by computing the minimal polynomials of the differences of the left-hand sides.
An explicit expression for $\nu^\star$ reads as
\begin{equation}
\begin{split}
	\nu^\star &= \frac{1}{4}\left(-1+2\nu_0+2\sqrt{\frac{51}{4} - \nu^2_0 - \frac{9}{4 \nu_0}}\right), \\
	\nu^2_0 &= \frac{17}{4}+2\sqrt{\frac{22}{3}}
		\cos\left(\frac{1}{3}\arccos\left(\frac{27}{22}\sqrt{\frac{3}{22}}\right)\right).
\end{split}
\end{equation}
The optimal objective value $\lambda^\star_{\mathrm{min}}$ is given by
\begin{equation}\label{eq:lambdamin}
\begin{split}
	\lambda^\star_{\mathrm{min}} &= 
	\frac{1}{8}\left(9 - \lambda_0 - \sqrt{51 - \lambda^2_0 - \frac{18}{\lambda_0}}\right), \\
	\lambda^2_0 &= 17 + 
	8\sqrt{\frac{22}{3}}\cos\left[\frac{1}{3}\arccos\left(\frac{27}{22}\sqrt{\frac{3}{22}}\right)\right].
\end{split}
\end{equation}
The key point to efficiently work with such numbers is representing them in an
explicit algebraic form, for example, we have 
\begin{displaymath}
	\cos\left[\frac{1}{3}\arccos\left(\frac{27}{22}\sqrt{\frac{3}{22}}\right)\right] = 
	\frac{\re\sqrt[3]{81+i\sqrt{25383}}}{\sqrt[6]{3}\sqrt{22}}.
\end{displaymath}
In fact, this is how all starred quantities are dealt with here. The minimal
polynomial of this optimal value is
\begin{equation}
	8\lambda^4 - 36 \lambda^3 + 48 \lambda^2 - 18 \lambda - 1 = 0.
\end{equation}
Further simplification of the expression \eqref{eq:lambdamin} does not seem to
be possible. This lengthy construction shows that 
\begin{equation}
	\mathcal{E}^+_{12|34}(\gamma) = -\lambda^\star_{\mathrm{min}} = 0.0489333...,
\end{equation}
and thus the CM given by Eq.~\eqref{eq:BBS} is $12|34$-entangled. Note that the
matrices $X^\star_{12|34}$ and $P^\star_{12|34}$ are degenerate, but of the same
rank, $\rank(X^\star_{12|34}) = \rank(P^\star_{12|34}) = 2$.

\subsubsection{The $12|34$ partition, scaling}

The primal solution produced by the solver is easily identified,
\begin{equation}\label{eq:gamma1234}
	\begin{alignedat}{2}
		&\gamma^{12\star}_{xx} = \frac{1}{\sqrt{2}}
		\begin{pmatrix}
			1 & 0 \\
			0 & 1
		\end{pmatrix}, \quad
		&&\gamma^{12\star}_{pp} = \frac{1}{2\sqrt{2}}
		\begin{pmatrix}
			1 & 0 \\
			0 & 1
		\end{pmatrix}, \\
		&\gamma^{34\star}_{xx} = \frac{1}{2\sqrt{2}}
		\begin{pmatrix}
			1 & 0 \\
			0 & 1
		\end{pmatrix}, \quad
		&&\gamma^{34\star}_{pp} = \frac{1}{\sqrt{2}}
		\begin{pmatrix}
			1 & 0 \\
			0 & 1
		\end{pmatrix}.
	\end{alignedat}
\end{equation}
The analytical expression of the optimal value $t^\star$ can be found from the
requirement that the minimums of the eigenvalues of the differences
\begin{equation}
	\gamma_{xx} - t^\star \gamma^{12\star}_{xx} \oplus \gamma^{34\star}_{xx}, \quad
	\gamma_{pp} - t^\star \gamma^{12\star}_{pp} \oplus \gamma^{34\star}_{pp}
\end{equation}
are zero. Equating the minimal eigenvalues of the differences to zero, we derive
that $t^\star$ is the minimum of the solutions of the following two equations:
\begin{equation}
\begin{split}
	8 - 3\sqrt{2}t^\star - \sqrt{16 + t^{\star 2}} &= 0, \\
	10 - 3\sqrt{2}t^\star - \sqrt{52 - 12\sqrt{2}t^\star + 2t^{\star 2}} &= 0.
\end{split}
\end{equation}
It turns out that the roots of these equations coincide:
\begin{equation}\label{eq:tstar}
	t^\star = \frac{3\sqrt{2} - \sqrt{6}}{2} = \sqrt{6 - 3\sqrt{3}}.  
\end{equation}

We could construct and solve the system of KKT equations to find the dual
solution, but in this case the matrices $X^\star_{12|34}$ and $P^\star_{12|34}$
are easily recognizable to be in the form
\begin{equation}
	\begin{split}
		X^\star_{12|34} = 
		\begin{pmatrix}
			u & 0 & -\sqrt{2uv} & 0 \\
			0 & u & 0 & \sqrt{2uv} \\
			-\sqrt{2uv} & 0 & 2v & 0 \\
			0 & \sqrt{2uv} & 0 & 2v
		\end{pmatrix}, \\
		P^\star_{12|34} = 
		\begin{pmatrix}
			2u & 0 & 0 & \sqrt{2uv} \\
			0 & 2u & \sqrt{2uv} & 0 \\
			0 & \sqrt{2uv} & v & 0 \\
			\sqrt{2uv} & 0 & 0 & v
		\end{pmatrix},
	\end{split}
\end{equation}
where the numerical values of $u$ and $v$ are
\begin{equation}
	u = 0.2788..., \quad v = 0.0747... .
\end{equation}
The matrices $Z^{12\star}$ and $Z^{34\star}$ read as
\begin{equation}
	Z^{12\star} = \sqrt{2}u
	\begin{pmatrix}
		1 & 0 \\
		0 & 1
	\end{pmatrix}, \quad
	Z^{34\star} = \sqrt{2}v
	\begin{pmatrix}
		1 & 0 \\
		0 & 1
	\end{pmatrix}.
\end{equation}
Our conjecture is that the exact values of $u$ and $v$ are the solution
$u^\star$ and $v^\star$ of the following optimization problem:
\begin{equation}
\begin{split}
	&\min_{\tr(Z^{12\star} + Z^{34\star})=1} \tr(X^\star_{12|34}\gamma_{xx} + P^\star_{12|34}\gamma_{pp}) \\
	&= \min_{u+v=1/(2\sqrt{2})} 4(u - \sqrt{2}\sqrt{uv} + 2v).
\end{split}
\end{equation}
This problem can be solved much easier than a similar problem in the eigenvalue
case, and the solution reads as
\begin{equation}
	u^\star = \frac{3\sqrt{2} + \sqrt{6}}{24}, \quad
	v^\star = \frac{3\sqrt{2} - \sqrt{6}}{24}.
\end{equation}
The optimal value is exactly the value \eqref{eq:tstar}. The normalization
condition
\begin{equation}
	\tr Z^{12\star} + \tr Z^{34\star} = 1,
\end{equation}
and the KKT conditions \eqref{eq:MMKKT} and \eqref{eq:KKTscaling2} are
satisfied:
\begin{equation}
	\begin{split}
		\mathcal{P}(\gamma^{12\star}_{xx}, \gamma^{12\star}_{pp})
		\begin{pmatrix}
			X^\star_{12|34}[1, 2] & -Z^{12\star} \\
			-Z^{12\star\mathrm{T}} & P^\star_{12|34}[1, 2]
		\end{pmatrix} 
		&= 0, \\
		\mathcal{P}(\gamma^{34\star}_{xx}, \gamma^{34\star}_{pp})
		\begin{pmatrix}
			X^\star_{12|34}[3, 4] & -Z^{34\star} \\
			-Z^{34\star\mathrm{T}} & P^\star_{12|34}[3, 4]
		\end{pmatrix} 
		&= 0, \\
		X^\star_{12|34}(\gamma_{xx} - t^\star \gamma^{12\star}_{xx} \oplus \gamma^{34\star}_{xx}) &= 0, \\
		P^\star_{12|34}(\gamma_{pp} - t^\star \gamma^{12\star}_{pp} \oplus \gamma^{34\star}_{pp}) &= 0.
	\end{split}
\end{equation}
In terms of the measure $\mathcal{T}^+_{12|34}$ we have
\begin{displaymath}
	\mathcal{T}^+_{12|34}(\gamma_{xx}, \gamma_{pp}) = 1 - t^\star 
	= 1 - \sqrt{6 - 3\sqrt{3}} = 0.1034...,
\end{displaymath}
which gives another demonstration that the state is $12|34$-entangled. This
number was obtained in Ref.~\cite{NewJPhys.8.51} in a purely numerical way, here
we derived an analytical expression for this value. Note that the matrices
$X^\star_{12|34}$ and $P^\star_{12|34}$ are degenerate, but of the same rank,
$\rank(X^\star_{12|34}) = \rank(P^\star_{12|34}) = 2$.

\subsubsection{The $12|34$ partition, refined condition}

Here we solve the full problem, not only the reduced one as we did up to now.
The primal solution $\gamma^{12\star}$ produced by the solver is
\begin{equation}
	\gamma^{12\star} = \frac{1}{2\sqrt{2}}
	\begin{pmatrix}
		2 & 0 & 0 & 0 \\
		0 & 2 & 0 & 0 \\
		0 & 0 & 1 & 0 \\
		0 & 0 & 0 & 1
	\end{pmatrix}.
\end{equation}
The dual solution, the matrix $M$ is of the form
\begin{equation}
	M^\star = 
	\begin{pmatrix}
		x & 0 & -y & 0 & 0 & 0 & 0 & 0 \\
		0 & x & 0 & y & 0 & 0 & 0 & 0 \\
		-y & 0 & 2z & 0 & 0 & 0 & 0 & 0 \\
		0 & y & 0 & 2z & 0 & 0 & 0 & 0 \\
		0 & 0 & 0 & 0 & 2x & 0 & 0 & y \\
		0 & 0 & 0 & 0 & 0 & 2x & y & 0 \\
		0 & 0 & 0 & 0 & 0 & y & z & 0 \\
		0 & 0 & 0 & 0 & y & 0 & 0 & z
	\end{pmatrix},
\end{equation}
and matrix $K$ is of the form
\begin{equation}
	K^\star = 
	\begin{pmatrix}
		0 & 0 & 0 & 0 & 0 & u & v & 0 \\ 
		0 & 0 & 0 & 0 & -u & 0 & 0 & -v \\
		0 & 0 & 0 & 0 & 0 & -2v & -w & 0 \\
		0 & 0 & 0 & 0 & -2v & 0 & 0 & -w \\
		0 & u & 0 & 2v & 0 & 0 & 0 & 0 \\
		-u & 0 & 2v & 0 & 0 & 0 & 0 & 0 \\
		-v & 0 & w & 0 & 0 & 0 & 0 & 0 \\
		0 & v & 0 & w & 0 & 0 & 0 & 0
	\end{pmatrix}.
\end{equation}
The matrix $K^{12\star}$ is of the form
\begin{equation}
	K^{12\star} = 
	\begin{pmatrix}
		0 & 0 & -u & 0 \\
		0 & 0 & 0 & -u \\
		u & 0 & 0 & 0 \\
		0 & u & 0 & 0
	\end{pmatrix}.
\end{equation}
The numerical values of the parameters are given by
\begin{equation}\label{eq:params}
	\begin{alignedat}{3}
		x &= 0.2788..., &\quad y &= 0.2041..., &\quad z &= 0.0747..., \\ 
		u &= 0.3943..., &\quad v &= 0.1443..., &\quad w &= 0.1056... .
	\end{alignedat}
\end{equation}
The normalization condition \eqref{eq:KO1} reads as
\begin{equation}
	\frac{1}{2} \tr(K^{12\star}\Omega_2) + 
	\frac{1}{2} \tr(K^\star[3] \Omega_2) = 1,
\end{equation}
where, as in Sec.~\ref{sec:IVA}, in the full case we denote by $K[I]$ the
submatrix of $K$ with row and column indices in the set $I \cup (I+n)$. In this
case the full notation would be $K^\star[3, 7]$. This condition produces
the equation
\begin{equation}\label{eq:new-1}
	2(u+w) = 1.
\end{equation}
The KKT condition \eqref{eq:PK} is
\begin{equation}
	\mathcal{P}(\gamma^{12\star})(M^\star[1, 2] + i K^{12\star}) = 0.
\end{equation}
According to the note above, the full notation for $M[1, 2]$ would be $M[1, 2,
5, 6]$. This condition produces a single equation
\begin{equation}\label{eq:new-2}
	u = \sqrt{2}x.
\end{equation}
The  KKT condition \eqref{eq:Lambdagamma} reads as
\begin{equation}
	(M^\star - i K^\star)(\gamma - t \gamma^{12\star} \oplus \frac{i}{2}\Omega_2) = 0,
\end{equation}
where 
\begin{equation}
	t = \tr(M^\star \gamma) = 4(x - y + 2z).
\end{equation}
This condition produces a system of equations, which contains 
\begin{equation}
\begin{split}
	x - 2 y + 4 v (x - y + 2 z) &= 0, \\
	v + 2 w - 4 z (x - y + 2 z) &= 0, \\
	v + u - 2\sqrt{2}u(x - y + 2 z) & = 0,
\end{split}
\end{equation}
and several other similar equations. Combining these equations with
Eqs.~\eqref{eq:new-1} and \eqref{eq:new-2} and solving the resulting system, we
obtain two solutions:
\begin{displaymath}
\begin{alignedat}{3}
	x^\star &= \frac{3\sqrt{2} \pm \sqrt{6}}{24}, &\quad y^\star &= \pm \frac{1}{2\sqrt{6}}, &\quad z^\star &= \frac{3\sqrt{2} \mp \sqrt{6}}{24}, \\ 
	u^\star &= \frac{3 \pm \sqrt{3}}{12}, &\quad v^\star &= \pm\frac{1}{4\sqrt{3}}, &\quad w^\star &= \frac{3 \mp \sqrt{3}}{12}.
\end{alignedat}
\end{displaymath}
The solution \eqref{eq:params} produced by the solver corresponds to the first
solution, where we choose the upper sign, though both choices are valid
solutions. The corresponding primal objective values are given by
\begin{equation}
	t^\star = \sqrt{6 \mp 3\sqrt{3}},
\end{equation}
so only the first choice gives the minimal value for $t$. This value coincides
with the value \eqref{eq:tstar} obtained with the scaling optimization problem
and shows that the CM is $12|34$-entangled. Note that $M^\star$ is degenerate,
$\rank(M^\star) = 4$.

\subsection{Genuine 3-partite entanglement}\label{sec:VIIC}

Here we consider a more complicated case of a 3-partite genuine entangled state
treated numerically in \cite{NewJPhys.8.51}. That state is also reduced with the
CM given by
\begin{equation}
	\gamma_{xx} = \frac{1}{4}
	\begin{pmatrix}
		2 & 1 & 1 \\
		1 & 2 & 1 \\
		1 & 1 & 2
	\end{pmatrix}, \quad
	\gamma_{pp} = \frac{1}{4}
	\begin{pmatrix}
		3 & -1 & -1 \\
		-1 & 3 & -1 \\
		-1 & -1 & 3
	\end{pmatrix}.
\end{equation}
It is an instance of the symmetric states \eqref{eq:gammalm} with $\lambda = 2$
and $\mu = 1/2$ and thus is a pure state. It has been numerically demonstrated
that this state is genuine entangled. Here we again construct analytical
solutions of the two optimizations problems. These solutions are a bit more
complicated then in the previous case, where only single bipartitions are
involved.

\subsubsection{Eigenvalue problem}

The solver produces the following optimal matrices
\begin{equation}
	X^\star = 
	\begin{pmatrix}
		2x & -x & -x \\
		-x & 2x & -x \\
		-x & -x & 2x
	\end{pmatrix}, \quad
	P^\star = 
	\begin{pmatrix}
		p & q & q \\
		q & p & p \\
		q & q & p
	\end{pmatrix},
\end{equation}
where the numerical values of $x$, $p$ and $q$ are given by
\begin{equation}
	x = 0.0905..., \quad p = 0.1522..., \quad q = 0.1308....
\end{equation}
Because the state is symmetric, it is enough to consider only one bipartition,
for example, $12|3$. The $Z$ matrices read as
\begin{equation}
	Z^{12\star} = 
	\begin{pmatrix}
		u & v \\
		v & u
	\end{pmatrix}, \quad
	Z^{3\star} = 
	\begin{pmatrix}
		w
	\end{pmatrix},
\end{equation}
where the numerical values of $u$, $v$ and $w$ are given by
\begin{equation}
	u = 0.1181..., \quad v = 0.0419..., \quad w = 0.1660....
\end{equation}
Two eigenvalues of the matrix
\begin{equation}
	\begin{pmatrix}
		X^\star[1, 2] & Z^{12\star} \\
		Z^{12\star\mathrm{T}} & P^\star[1, 2]
	\end{pmatrix} = 
	\begin{pmatrix}
		2x & -x & u & v \\
		-x & 2x & v & u \\
		u & v & p & q \\
		v & u & q & p
	\end{pmatrix}
\end{equation}
must be zero, so we have two equations
\begin{equation}\label{eq:eq-1}
	(u + v)^2 = x(p + q), \quad (u - v)^2 = 3x(p - q).
\end{equation}
One eigenvalue of the matrix 
\begin{equation}
	\begin{pmatrix}
		X^\star[3] & Z^{3\star} \\
		Z^{3\star\mathrm{T}} & P^\star[3]
	\end{pmatrix} = 
	\begin{pmatrix}
		2x & w \\
		w & p 
	\end{pmatrix}
\end{equation}
must be zero, so we have another equation
\begin{equation}\label{eq:eq-2}
	w^2 = 2xp.
\end{equation}
Normalization condition reads as
\begin{equation}\label{eq:eq-3}
	\tr(X^\star + P^\star) = 6x + 3p = 1.
\end{equation}
The primal objective function is given by
\begin{displaymath}
\begin{split}
	\lambda_{\mathrm{min}} &= \tr(X^\star \gamma_{xx} + P^\star \gamma_{pp}) 
	-\tr(Z^{12\star}) - \tr(Z^{3\star}) \\
	&= \frac{1}{4}(9p - 6q - 8u- 4w + 6x).
\end{split}
\end{displaymath}
The dual solution produced by the solver is given by
\begin{displaymath}
\begin{split}
	&\gamma^{12\star}_{xx} = \gamma^{13\star}_{xx} = \gamma^{23\star}_{xx} = 
	\begin{pmatrix}
		a & b \\
		b & a
	\end{pmatrix}, \quad \gamma^{3}_{xx} = \gamma^{2}_{xx} = \gamma^{1}_{xx} = 
	\begin{pmatrix}
		c
	\end{pmatrix} \\
	&\gamma^{12\star}_{pp} = \gamma^{13\star}_{pp} = \gamma^{23\star}_{pp} = 
	\begin{pmatrix}
		d & e \\
		e & d
	\end{pmatrix}, \quad \gamma^{3}_{pp} = \gamma^{2}_{pp} = \gamma^{1}_{pp} = 
	\begin{pmatrix}
		f
	\end{pmatrix},
\end{split}
\end{displaymath}
where the numerical values of the parameters read as
\begin{equation}
	\begin{alignedat}{3}
		a &= 0.5121..., \quad &b &= 0.3719..., \quad &c &= 0.4584..., \\
		d &= 1.0327..., \quad &e &= -3/4, \quad &f &= 0.5453....
	\end{alignedat}
\end{equation}
Though all $xx$- and $pp$-matrices are the same, they
compose in different ways, for example
\begin{displaymath}
	\gamma^{12\star}_{xx} \oplus \gamma^{3\star}_{xx} = 
	\begin{pmatrix}
		a & b & 0 \\
		b & a & 0 \\
		0 & 0 & c
	\end{pmatrix}, \quad
	\gamma^{13\star}_{xx} \oplus \gamma^{2\star}_{xx} = 
	\begin{pmatrix}
		a & 0 & b \\
		0 & c & 0 \\
		b & 0 & a
	\end{pmatrix}.
\end{displaymath}
The coefficients of the mixture are equal
\begin{equation}
	p^\star_{12|3} = p^\star_{13|2} = p^\star_{23|1} = \frac{1}{3}.
\end{equation}
Two eigenvalues of the matrix $\mathcal{P}(\gamma^{12\star}_{xx},
\gamma^{12\star}_{pp})$ must be zero, which reduces to the following equations:
\begin{equation}\label{eq:eq-4}
	(a+b)(4d-3) = 1, \quad (a-b)(4d+3) = 1.
\end{equation}
Two eigenvalues of the matrix $\mathcal{P}(\gamma^{3\star}_{xx},
\gamma^{3\star}_{pp})$ must be zero, which reduces to the following equation:
\begin{equation}\label{eq:eq-5}
	4cf = 1.
\end{equation}
The KKT condition
\begin{equation}
	\mathcal{P}(\gamma^{12\star}_{xx}, \gamma^{12\star}_{pp})
	\begin{pmatrix}
		X^\star[1, 2] & -Z^{12\star} \\
		-Z^{12\star\mathrm{T}} & P^\star[1, 2]
	\end{pmatrix}
	= 0
\end{equation}
produces the following system of equations:
\begin{equation}\label{eq:eq-6}
\begin{alignedat}{4}
	&2x(2a-b) &&= u, &&\quad\quad
	2x(2b-a) &&= v, \\
	&2(au + bv) &&= p, &&\quad\quad
	2(bu + av) &&= q, \\
	&3v + 4x &&= 4du, &&\quad\quad
	3u - 2x &&= 4dv, \\
	&3q + 2u &&= 4dp, &&\quad\quad
	3p + 2v &&= 4dq.
\end{alignedat}
\end{equation}
The KKT condition
\begin{equation}\label{eq:kkt-2}
	\mathcal{P}(\gamma^{3\star}_{xx}, \gamma^{3\star}_{pp})
	\begin{pmatrix}
		X^\star[3] & -Z^{3\star} \\
		-Z^{3\star\mathrm{T}} & P^\star[3]
	\end{pmatrix}
	= 0
\end{equation}
produces the following system of equations:
\begin{equation}\label{eq:eq-7}
	w = 4cx, \quad	p = 2cw, \quad x = fw, \quad w = 2fp.
\end{equation}
The KKT conditions yield that the matrices satisfy
\begin{displaymath}
\begin{split}
	&X^\star\left(\gamma_{xx} - \frac{1}{3}\gamma^{12\star}_{xx}\oplus\gamma^{3\star}_{xx}
	- \ldots \right) = \lambda_{\mathrm{min}}X^\star \\
	&P^\star\left(\gamma_{pp} - \frac{1}{3}\gamma^{12\star}_{pp}\oplus\gamma^{3\star}_{pp}
	- \ldots \right) = \lambda_{\mathrm{min}}P^\star, 
\end{split}
\end{displaymath}
where the dots stand for the combination of the bipartitions, which produces the
following system of equations:
\begin{equation}\label{eq:eq-8}
\begin{split}
	&18x - 12w -24u -18q + 27p + 8d +4f = 9, \\
	&8 a - 4 b + 4 c - 8 d - 4 f = -6, \\
	&3 p + 4 a p - 2 b p + 2 c p - 4 d p - 2 f p = 0, \\
	&6 q + 8 a q - 4 b q + 4 c q - 8 d q - 4 f q = 0.
\end{split}
\end{equation}
Combining the equations given by Eqs.~\eqref{eq:eq-1}, \eqref{eq:eq-2},
\eqref{eq:eq-3}, \eqref{eq:eq-4}, \eqref{eq:eq-5}, \eqref{eq:eq-6},
\eqref{eq:eq-7} and \eqref{eq:eq-8}, we obtain a system of 21 equations for 11
variables $x$, $p$, $q$, $u$, $v$, $w$, $a$, $b$, $c$, $d$ and $f$. This system
has several solutions, one of which agrees with the values produced by the
solver. This solution consists of algebraic numbers of degree 7, the minimal
polynomials are given by
\begin{widetext}
\begin{displaymath}
\begin{split}
	2033478432x^{7} - 1355652288x^{6} +360282816x^{5}
	-48215520x^{4} +3371688x^{3} -114528x^{2} +146x -3 &= 0, \\
	21182067p^{7} -21182067p^{6} +7951095p^{5}
	-1388655p^{4} +116613p^{3} -5461p^{2} +285p -9 &= 0, \\
	15441726843q^{7} +6291073899q^{6} -401977647q^{5}
	-261154287q^{4} -1979829q^{3} +565821q^{2} +448227q -13633 &= 0, \\
	36602611776u^{7} -4066956864u^{6} -27226368u^{5}
	+30483216u^{4} -7018344u^{3} +38700u^{2} +35296u +851 &= 0, \\
	36602611776v^{7} +8133913728v^{6} -272310336v^{5}
	-124092000v^{4} +649656v^{3} +492948v^{2} -14480v +59 &= 0, \\
	7060689w^{7} -296802w^{5} +15552w^{4}
	+4138w^{3} -504w^{2} -36w +2 &= 0, \\
	6912a^{7} +20736a^{6} +7168a^{5}
	-7008a^{4} -2000a^{3} +144a^{2} -56a +51 &= 0, \\
	20736b^{7} -48384b^{6} -7296b^{5}
	+24736b^{4} +96b^{3} -1728b^{2} -432b +81 &= 0, \\
	768c^{7} -2048c^{6} -4416c^{5}
	+160c^{4} +1104c^{3} -56c^{2} -12c +9 &= 0, \\
	16384d^{7} -36864d^{6} +1024d^{5}
	+46848d^{4} -29504d^{3} -8112d^{2} +14508d -4131 &= 0, \\
	576f^{7} -192f^{6} -224f^{5}
	+1104f^{4} +40f^{3} -276f^{2} -32f +3 &= 0.
\end{split}
\end{displaymath}
\end{widetext}
The indices of the roots of these polynomials read as
\begin{displaymath}
\begin{alignedat}{6}
	x^\star &: 4, \quad p^\star &&: 2, \quad q^\star &&: 4, \quad u^\star &&: 5, \quad v^\star &&: 4, 
	\quad w^\star &&: 5, \\ 
	a^\star &: 5, \quad b^\star &&: 3, \quad c^\star &&: 4, \quad d^\star &&: 4, \quad f^\star &&: 5.
\end{alignedat}
\end{displaymath}
Different computer algebra systems might use different enumeration schemes for
roots, these indices are valid in the \textsl{Mathematica} nomenclature.
Standard arithmetic roots of real numbers are also definite roots of some simple
polynomials, so this representation is no less analytical than an expression
with explicit radicals.

It turns out that the numbers $x^\star$, $p^\star$ and $q^\star$ produce the
minimum of the function 
\begin{displaymath}
\begin{split}
	&f(x, p, q) = \tr(X^\star \gamma_{xx} + P^\star \gamma_{pp}) - \str_{12|3}(X^\star, P^\star) \\
	&\frac{1}{4}\left(1+6(p-q)-4\sqrt{x}(\sqrt{2p}+\sqrt{3(p-q)}+\sqrt{p+q})\right).
\end{split}
\end{displaymath}
Because $p^\star > q^\star$ we can omit the term containing $p-q$ in the
Lagrangian of this problem, 
\begin{equation}
	\mathcal{L} = f(x, p, q) + \nu\left(2x+p-\frac{1}{3}\right).
\end{equation}
To prove the optimality of the point $(x^\star, p^\star, q^\star)$ we need to
show that it is a solution of the following system of equations:
\begin{displaymath}
\begin{split}
	2\frac{\partial \mathcal{L}}{\partial x} &= 4\nu^\star - 
	\sqrt{\frac{2p^\star}{x^\star}}+\sqrt{\frac{3(p^\star-q^\star)}{x^\star}}+
	\sqrt{\frac{p^\star+q^\star}{x^\star}} = 0, \\
	2\frac{\partial \mathcal{L}}{\partial p} &= 2\nu^\star + 3 - \sqrt{\frac{2x^\star}{p^\star}}
	- \sqrt{\frac{3x^\star}{p^\star-q^\star}} - \sqrt{\frac{x^\star}{p^\star+q^\star}} = 0, \\
	2\frac{\partial \mathcal{L}}{\partial q} &= 
	-3 + \sqrt{\frac{3x^\star}{p^\star-q^\star}} - \sqrt{\frac{x^\star}{p^\star+q^\star}} = 0.
\end{split}
\end{displaymath}
By computing the minimal polynomials, one can check that the numbers $x^\star$,
$p^\star$ and $q^\star$ satisfy the equations
\begin{equation}
	\frac{\partial \mathcal{L}}{\partial q} = 0, \quad 
	\frac{\partial \mathcal{L}}{\partial x} - 2\frac{\partial \mathcal{L}}{\partial p} = 0,
\end{equation}
and thus are a solution of the system above, where $\nu^\star$ is the 5th root
of the equation
\begin{equation}
\begin{split}
	576 \nu^7 &+ 3264\nu^6 +2368\nu^5 -7728\nu^4 -1744\nu^3 \\
	&+3060\nu^2 -936\nu +81 = 0.
\end{split}
\end{equation}
For $\lambda^\star_{\mathrm{min}}$ we have
\begin{equation}
	\lambda^\star_{\mathrm{min}} = \frac{1}{4}(6x^\star + 9p^\star - 6q^\star - 8u^\star - 4w^\star),
\end{equation}
which is the 1st root of the equation
\begin{equation}
\begin{split}
	11943936 \lambda^7 &-43462656\lambda^6 +54973440\lambda^5 \\
	&-28567296\lambda^4 +5711552\lambda^3 \\
	&-119952\lambda^2 -91556\lambda +8163 = 0.
\end{split}
\end{equation}
We derive that
\begin{equation}
	\mathcal{E}^+_{\mathfrak{I}_2}(\gamma_{xx}, \gamma_{pp}) = -\lambda^\star_{\mathrm{min}} = 0.1202...,
\end{equation}
and thus the state is genuine 3-partite entangled. Note that $X$ and $P$ are of
different rank: $\rank(X^\star) = 2$, but $\rank(P^\star) = 3$.

\subsubsection{Scaling}

The primal optimal solution reads as
\begin{displaymath}
	\begin{split}
		&\gamma^{12\star}_{xx} = \gamma^{13\star}_{xx} = \gamma^{23\star}_{xx} = 
		\begin{pmatrix}
			a & b \\
			b & a
		\end{pmatrix}, \quad \gamma^{3}_{xx} = \gamma^{2}_{xx} = \gamma^{1}_{xx} = 
		\begin{pmatrix}
			c
		\end{pmatrix}, \\
		&\gamma^{12\star}_{pp} = \gamma^{13\star}_{pp} = \gamma^{23\star}_{pp} = 
		\begin{pmatrix}
			d & e \\
			e & d
		\end{pmatrix}, \quad \gamma^{3}_{pp} = \gamma^{2}_{pp} = \gamma^{1}_{pp} = 
		\begin{pmatrix}
			f
		\end{pmatrix},
	\end{split}
\end{displaymath}
where the values of the parameters read as
\begin{equation}
\begin{alignedat}{3}
	a &= 0.5100..., \quad &b &= 0.4073..., \quad &c &= 0.4672..., \\
	d &= 1.3526..., \quad &e &= -1.0800..., \quad &f &= 0.5350....
\end{alignedat}
\end{equation}
Two eigenvalues of the matrix $\mathcal{P}(\gamma^{12\star}_{xx},
\gamma^{12\star}_{pp})$ must be zero, which reduces to the following equations:
\begin{equation}\label{eq:eq2-1}
	4(a+b)(d+e) = 1, \quad 4(a-b)(d-e) = 1.
\end{equation}
Two eigenvalues of the matrix $\mathcal{P}(\gamma^{3\star}_{xx},
\gamma^{3\star}_{pp})$ must be zero, which reduces to the same equation as in
the eigenvalue case
\begin{equation}\label{eq:eq2-2}
	4cf = 1.
\end{equation}
The dual solution has the same form as in the previous case
\begin{equation}
	X^\star = 
	\begin{pmatrix}
		2x & -x & -x \\
		-x & 2x & -x \\
		-x & -x & 2x
	\end{pmatrix}, \quad
	P^\star = 
	\begin{pmatrix}
		p & q & q \\
		q & p & p \\
		q & q & p
	\end{pmatrix}.
\end{equation}
The $Z$ matrices read as
\begin{equation}
	Z^{12\star} = 
	\begin{pmatrix}
		u & v \\
		v & u
	\end{pmatrix}, \quad
	Z^{3\star} = 
	\begin{pmatrix}
		w
	\end{pmatrix}.
\end{equation}
The numerical values of all the parameters read as
\begin{equation}
	\begin{alignedat}{3}
		x &= 0.2314..., \quad &p &= 0.4042..., \quad &q &= 0.3749..., \\
		u &= 0.2836..., \quad &v &= 0.1409..., \quad &w &= 0.4326....
	\end{alignedat}
\end{equation}
Not surprisingly, we have the same equations for these parameters as in in the
eigenvalue case,
\begin{equation}\label{eq:eq2-3}
	(u + v)^2 = x(p + q), \ (u - v)^2 = 3x(p - q), \ w^2=2xp.
\end{equation}
The normalization condition is different and reads as
\begin{equation}\label{eq:eq2-4}
	\tr(Z^{12\star}) + \tr(Z^{3\star}) = 2u + w = 1.
\end{equation}
The KKT conditions \eqref{eq:kkt-1} and \eqref{eq:kkt-2} produce the following
system of equations:
\begin{equation}\label{eq:eq2-5}
	\begin{alignedat}{4}
		&2x(2a-b) &&= u, &&\quad\quad
		2x(2b-a) &&= v, \\
		&2(au + bv) &&= p, &&\quad\quad
		2(bu + av) &&= q, \\
		&du + ev &&= x, &&\quad\quad
		2(eu + dv) &&= -x, \\
		&2(dp + eq) &&= u, &&\quad\quad
		2(ep + dq) &&= v,
	\end{alignedat}
\end{equation}
and
\begin{equation}\label{eq:eq2-6}
	w = 4cx, \quad	p = 2cw, \quad x = fw, \quad w = 2fp.
\end{equation}
The KKT condition \eqref{eq:KKTscaling2} in this case state that 
\begin{equation}
	\begin{split}
		X^\star\left(\gamma_{xx} - \frac{t}{3}\gamma^{12\star}_{xx}\oplus\gamma^{3\star}_{xx}
		- \ldots \right) &= 0, \\
		P^\star\left(\gamma_{pp} - \frac{t}{3}\gamma^{12\star}_{pp}\oplus\gamma^{3\star}_{pp}
		- \ldots \right) &= 0,
	\end{split}
\end{equation}
where the dots stand for the combination of the other two bipartitions and
\begin{equation}
	t = \tr(X^\star \gamma_{xx} + P^\star \gamma_{pp}) = \frac{3}{4}(2x+3p-2q).
\end{equation}
These conditions produce the following additional equations:
\begin{equation}\label{eq:eq2-7}
\begin{split}
	(2 a - b + c) (3 p - 2 q + 2 x) &= 1, \\
	(2 q - 3 p - 2 x) (2 d p + f p + 2 e q) &= 2 q - 3 p, \\
	(2 q - 3 p - 2 x) (e p + 2 d q + e q + f q) &= p - 2 q.
\end{split}
\end{equation}
Combining the equations given by Eqs.~\eqref{eq:eq2-1}, \eqref{eq:eq2-2},
\eqref{eq:eq2-3}, \eqref{eq:eq2-4}, \eqref{eq:eq2-5}, \eqref{eq:eq2-6},
\eqref{eq:eq2-7}, we obtain a system of 22 equations for 12 variables $a$, $b$,
$c$, $d$, $e$, $f$, $x$, $p$, $q$, $u$, $v$, and $w$. This system has several
solutions, one of which agrees with the values produced by the solver. The
minimal polynomials of this solution read as
\begin{widetext}
\begin{equation}
\begin{alignedat}{2}
	14400 a^8 - 2880 a^6 - 396 a^4 + 29 a^2 + 4 &= 0, &\quad
	9216 x^8 - 11968 x^6 + 4224 x^4 - 660 x^2 + 25 &= 0, \\
	230400 b^8 - 31680 b^6 + 3024 b^4 - 736 b^2 + 9 &= 0, &\quad 
	14400 p^8 - 8635 p^6 + 879 p^4 + 18 p^2 + 1 &= 0, \\
	384 c^4  + 25600 c^8 - 7360 c^6 + 384 c^4 - 4 c^2 + 1 &= 0, &\quad 
	3686400 q^8 - 3271360 q^6 + 933744 q^4 - 81300 q^2 + 625 &= 0, \\
	1600 d^8 - 4240 d^6 + 2616 d^4 - 427 d^2 + 64 &= 0, &\quad 
	192 u^4 - 424 u^3 + 324 u^2 - 108 u + 13 &= 0, \\
	64 f^8 - 16 f^6 + 96 f^4 - 115 f^2 +25 &= 0, &\quad 
	96 v^4 + 88 v^3 + 36 v^2 - 1 &= 0, \\
	1600 e^8 - 2640 e^6 + 1056 e^4 - 187 e^2 + 9 &= 0, &\quad 
	12 w^4 + 5 w^3 - 6 w^2 + 3 w - 1 &= 0.
\end{alignedat}
\end{equation} 
\end{widetext}
The indices of the roots of these polynomials read as
\begin{displaymath}
\begin{alignedat}{6}
	a^\star &: 4, \quad b^\star &&: 4, \quad c^\star &&: 4, \quad d^\star &&: 4, \quad e^\star &&: 1, 
	\quad f^\star &&: 3, \\ 
	x^\star &: 3, \quad p^\star &&: 3, \quad q^\star &&: 4, \quad u^\star &&: 1, \quad v^\star &&: 2, 
	\quad w^\star &&: 2.
\end{alignedat}
\end{displaymath}
It is possible to give explicit expression for these numbers in radicals, here
we show $x^\star$, $p^\star$ and $q^\star$ only:
\begin{widetext}
\begin{equation}
\begin{alignedat}{2}
	x^{\star 2} &= \frac{1}{576}\left(187 + x_0 - \sqrt{28875 - x^2_0 + \frac{2281862}{x_0}}\right), \quad
	&x^2_0 &= 9625 - 432(x_+ + x_-), \\
	p^{\star 2} &= \frac{1}{11520}\left(1727 + p_0 - \sqrt{4897155 - p^2_0 + \frac{2828801662}{p_0}}\right), \quad 
	&p^2_0 &= 1632385 + 2304(p_+ + p_-), \\
	q^{\star 2} &= \frac{1}{46080} \left(10223 - q_0 + \sqrt{44610915 - q^2_0 +
	\frac{72877799522}{q_0}}\right), \quad
	&q^2_0 &= 14870305 + 13824 (q_+ + q_-),
\end{alignedat}
\end{equation}
\end{widetext}
where all square roots are understood in the arithmetic sense and
\begin{equation}
\begin{split}
	x^3_{\pm} &= \frac{4873 \pm 395 \sqrt{237}}{2}, \\
	p^3_{\pm} &= \frac{141323729 \pm 4374115 \sqrt{237}}{2}, \\
	q^3_{\pm} &= \frac{3212429769 \pm 111739125 \sqrt{237}}{2},
\end{split}
\end{equation}
and here again cubic roots are taken as real roots of real numbers. The optimal
value 
\begin{equation}
	t^\star = \frac{3}{4}(2x^\star+3p^\star-2q^\star)
\end{equation}
has minimal polynomial
\begin{equation}
	1024 t^8 - 11968 t^6 + 38016 t^4 - 53460 t^2 + 18225 = 0
\end{equation}
and explicitly reads as
\begin{equation}
	t^{\star 2} = \frac{1}{64}\left(187 + t_0 - \sqrt{28875 - t^2_0 + \frac{2281862}{t_0}}\right),
\end{equation}
where 
\begin{displaymath}
	t^2_0 = 9625 - 432 (t_+ + t_-), \quad t^3_{\pm} = \frac{4873 \pm 395 \sqrt{237}}{2}.
\end{displaymath}
We derive that
\begin{equation}
	\mathcal{T}^+_{\mathfrak{I}_2}(\gamma_{xx}, \gamma_{pp}) = 1-t^\star = 0.3056...,
\end{equation}
which proves that the state is genuine 3-partite entangled. Note that the values
for $x^\star$, $p^\star$, $q^\star$ and $t^\star$ were obtained in
Ref.~\cite{NewJPhys.8.51} in a purely numerical way and without any evidence of
their optimality. Here we constructed analytical expressions for these numbers
together with an analytical validation of their optimality.

\subsubsection{Refined condition}

We first consider the case of $\mathfrak{I}_2 = \{1|23, 2|13, 3|12\}$, so the
second, larger, components are discarded and the smaller ones remain. The primal
solution in this case is given by
\begin{equation}
	\gamma^{1\star} = \gamma^{2\star} = \gamma^{3\star} = \frac{1}{2}
	\begin{pmatrix}
		1 & 0 \\
		0 & 1
	\end{pmatrix}.
\end{equation}
The optimal $M^\star$ read as
\begin{equation}\label{eq:Mstar}
	M^\star = 
	\begin{pmatrix}
		x & -y & -y & 0 & 0 & 0 \\
		-y & x & -y & 0 & 0 & 0 \\
		-y & -y & x & 0 & 0 & 0 \\
		0 & 0 & 0 & p & q & q \\
		0 & 0 & 0 & q & p & q \\
		0 & 0 & 0 & q & q & p
	\end{pmatrix},
\end{equation}
and the optimal $K^\star$ is given by
\begin{equation}
	K^\star = 
	\begin{pmatrix}
		0 & 0 & 0 & -u & -v & -v \\
		0 & 0 & 0 & -v & -u & -v \\
		0 & 0 & 0 & -v & -v & -u \\
		u & v & v & 0 & 0 & 0 \\
		v & u & v & 0 & 0 & 0 \\
		v & v & u & 0 & 0 & 0
	\end{pmatrix}.
\end{equation}
The matrices $K^{j\star}$ read as
\begin{equation}
	K^{1\star} = K^{2\star} = K^{13\star} = x
	\begin{pmatrix}
		0 & -1 \\
		1 & 0
	\end{pmatrix}.
\end{equation}
The numerical values of the parameters produced by the solver are
where the numerical values of the parameters read as
\begin{equation}
	\begin{alignedat}{3}
		x &= 0.4469..., \quad &y &= 0.1516..., \quad &p &= 0.4469..., \\
		q &= 0.3750..., \quad &u &= 0.2765..., \quad &v &= 0.0691....
	\end{alignedat}
\end{equation}
The normalization condition \eqref{eq:KO1} reads as
\begin{equation}
	\frac{1}{2} \tr(K^{1\star}\Omega_1) + 
	\frac{1}{2} \tr(K^\star[2, 3] \Omega_2) = 1,
\end{equation}
and produces the equation
\begin{equation}\label{eq:new2-1}
	2u+w = 1.
\end{equation}
The KKT condition \eqref{eq:PK} is
\begin{equation}
	\mathcal{P}(\gamma^{1\star})(M^\star[1] + i K^{1\star}) = 0,
\end{equation}
which produces a single equation 
\begin{equation}\label{eq:new2-2}
	x = p.
\end{equation}
The  KKT condition \eqref{eq:Lambdagamma} reads as
\begin{equation}
\begin{split}
	(M^\star &- i K^\star)\left(\gamma - \frac{t}{3} \gamma^{1\star} \oplus \frac{i}{2}\Omega^{23}_2 \right.\\
	&- \left.\frac{t}{3} \gamma^{2\star} \oplus \frac{i}{2}\Omega^{13}_2
	- \frac{t}{3} \gamma^{3\star} \oplus \frac{i}{2}\Omega^{12}_2\right) = 0,
\end{split}
\end{equation}
where 
\begin{equation}
	t = \tr(M^\star \gamma) = \frac{3}{4}(2x-2y+3p-2q).
\end{equation}
This condition produces a system of equations, which contains 
\begin{equation}
\begin{split}
	(2 u + x) (2 x - 2y +3 p - 2 q) &= 4(x-y), \\
	(2 v - y) (2 x - 2y +3 p - 2 q) &= 2(x-3y),
\end{split}
\end{equation}
and several others. Combining these equations with Eqs.~\eqref{eq:new2-1} and
\eqref{eq:new2-2} and solving the resulting system, we obtain the following
solution:
\begin{displaymath}
\begin{alignedat}{2}
	x^\star &= \frac{-73+20\sqrt{73}}{219}, &\quad y^\star &= \frac{73+7\sqrt{73}}{876}, \\
	p^\star &= \frac{-73+20\sqrt{73}}{219}, &\quad q^\star &= \frac{-73+47\sqrt{73}}{876}, \\
	u^\star &= \frac{2(73-5\sqrt{73})}{219}, &\quad v^\star &= \frac{73-5\sqrt{73}}{438}.
\end{alignedat}
\end{displaymath}
The optimal $t^\star$ reads as
\begin{equation}
	t^\star = \frac{-5+\sqrt{73}}{4} = 0.8860....
\end{equation}
Since $t^\star < 1$, it shows that the state is genuine 3-partite entangled.
Note that the matrix $M^\star$ is degenerate and $\rank(M^\star) = 3$.

Now we consider the case of $\mathfrak{I}_2 = \{23|1, 13|2, 12|3\}$, so the
smaller components are discarded and the larger ones remain. The primal solution
is given by
\begin{equation}
	\gamma^{12\star} = \gamma^{13\star} = \gamma^{23\star} = 
	\begin{pmatrix}
		a & b & 0 & 0 \\
		b & a & 0 & 0 \\
		0 & 0 & c & -d \\
		0 & 0 & -d & c
	\end{pmatrix}.
\end{equation}
The dual solution is of the form given by Eq.~\eqref{eq:Mstar} and $K^\star$
reads as
\begin{equation}
	K^\star = 
	\begin{pmatrix}
		0 & 0 & 0 & -u & v & v \\
		0 & 0 & 0 & v & -u & v \\
		0 & 0 & 0 & v & v & -u \\
		u & -v & -v & 0 & 0 & 0 \\
		-v & u & -v & 0 & 0 & 0 \\
		-v & -v & u & 0 & 0 & 0
	\end{pmatrix}.
\end{equation}
The matrices $K^{jk\star}$ are given by
\begin{equation}
	K^{12\star} = K^{13\star} = K^{23\star} = 
	\begin{pmatrix}
		0 & 0 & -z & -w \\
		0 & 0 & -w & -z \\
		z & w & 0 & 0 \\
		w & z & 0 & 0
	\end{pmatrix}.
\end{equation}
The numerical values of the parameters parameters produced by the solver are
\begin{equation}
	\begin{alignedat}{3}
		a &= 0.4484..., \quad &b &= 0.2522..., \quad &c &= 0.8154..., \\
		d &= 0.4586..., \quad &x &= 0.5416..., \quad &y &= 0.2443..., \\
		p &= 0.3525..., \quad &q &= 0.2314..., \quad &u &= 0.2748..., \\
		v &= 0.0335..., \quad &z &= 0.3625..., \quad &w &= 0.0541....
	\end{alignedat}
\end{equation}
The normalization condition reads as
\begin{equation}
	\frac{1}{2} \tr(K^{12\star}\Omega_2) + \frac{1}{2} \tr(K^\star[3]\Omega_1) = 1,
\end{equation}
and produces the equation
\begin{equation}\label{eq:new3-1}
	u+2z=1.
\end{equation}
The KKT condition \eqref{eq:PK} is
\begin{equation}
	\mathcal{P}(\gamma^{12\star})(M^\star[1, 2] + i K^{12\star}) = 0,
\end{equation}
which produces the equations
\begin{equation}\label{eq:new3-2}
\begin{alignedat}{2}
	2(ax - by) &= z, \quad &2(bx-ay) &= w, \\
	2(az+bw) &= p, \quad &2(aw+bz) &= q, \\
	2(cz-dw) &= x, \quad &2(dz-cw) &= y, \\
	2(cp-dq) &= z, \quad &2(cq-dp) &= w.
\end{alignedat}
\end{equation}
The  KKT condition \eqref{eq:Lambdagamma} reads as
\begin{equation}
\begin{split}
	(M^\star &- i K^\star)\left(\gamma - \frac{t}{3} \gamma^{12\star} \oplus \frac{i}{2}\Omega^{3}_1 \right.\\
	&- \left.\frac{t}{3} \gamma^{13\star} \oplus \frac{i}{2}\Omega^{2}_1
	- \frac{t}{3} \gamma^{23\star} \oplus \frac{i}{2}\Omega^{1}_1\right) = 0,
\end{split}
\end{equation}
where 
\begin{equation}
	t = \tr(M^\star \gamma) = \frac{3}{4}(2x-2y+3p-2q).
\end{equation}
This relation produces several larger equations, not presented here. Combining
these larger equations with Eqs.~\eqref{eq:new3-1} and \eqref{eq:new3-2} and
solving the resulting system, we obtain the following analytical expressions for
the parameters:
\begin{displaymath}
	\begin{alignedat}{3}
		a^\star &= \frac{4}{5}\sqrt{\frac{11}{35}}, \quad &b^\star &= \frac{9}{20}\sqrt{\frac{11}{35}}, \quad &c^\star &= \frac{16}{\sqrt{385}}, \\
		d^\star &= \frac{9}{\sqrt{385}}, \quad &x^\star &= \frac{491}{231}\sqrt{\frac{5}{77}}, \quad &y^\star &= \frac{443}{462}\sqrt{\frac{5}{77}}, \\
		p^\star &= \frac{83}{12\sqrt{385}}, \quad &q^\star &= \frac{109}{24\sqrt{385}}, \quad &u^\star &= \frac{127}{462}, \\
		v^\star &= \frac{31}{924}, \quad &z^\star &= \frac{335}{924}, \quad &w^\star &= \frac{25}{462}.
	\end{alignedat}
\end{displaymath}
The optimal value $t^\star$ is given by
\begin{equation}
	t^\star = \frac{1}{2}\sqrt{\frac{35}{11}} = 0.8918....
\end{equation}
Since $t^\star < 1$, it shows the state is genuine 3-partite entangled. Note
that in this case $M^\star$ is also degenerate and $\rank(M^\star) = 3$.

\subsection{Genuine 4-partite entanglement}\label{sec:VIID}

Our last example is a little bit more complicated case of a 4-partite reduced
state with the following CM:
\begin{equation}\label{eq:G4}
	\gamma_{xx} = 
	\begin{pmatrix}
		\alpha & \beta & \beta & 0 \\
		\beta & \alpha & 0 & \beta \\
		\beta & 0 & \alpha & -\beta \\
		0 & \beta & -\beta & \alpha
	\end{pmatrix}, \quad
	\gamma_{pp} = 
	\begin{pmatrix}
		\alpha & -\beta & -\beta & 0 \\
		-\beta & \alpha & 0 & -\beta \\
		-\beta & 0 & \alpha & \beta \\
		0 & -\beta & \beta & \alpha
	\end{pmatrix},
\end{equation}
where the parameters $\alpha$ and $\beta$ are given by
\begin{equation}\label{eq:Galphabeta}
\begin{split}
	\alpha &= \frac{1}{4}\left(\frac{1}{2\sqrt[3]{2}} + 10\sqrt[3]{2}\right) = 3.2490..., \\
	\beta &= \frac{1}{4\sqrt{2}}\left(\frac{1}{2\sqrt[3]{2}} - 10\sqrt[3]{2}\right) = -2.1570....
\end{split}
\end{equation}
This a mixed state. Such states occur in an intermediate step of a
continuous-variable entanglement swapping, as introduced in
Refs.~\cite{PhysRevA.60.2752, PhysRevA.61.010302}. This somewhat uncommon choice
of parameters was proposed in Ref.~\cite{NewJPhys.8.51}. As usually, we
construct the solutions of the two optimization problems.

\subsubsection{Eigenvalue}

The optimal matrices produced by the solver have the following form:
\begin{equation}
\begin{split}
	X^\star &= \frac{1}{8}
	\begin{pmatrix}
		1 & 8x & 8x & 0 \\
		8x & 1 & 0 & 8x \\
		8x & 0 & 1 & -8x \\
		0 & 8x & -8x & 1
	\end{pmatrix}, \\
	P^\star &= \frac{1}{8}
	\begin{pmatrix}
		1 & -8x & -8x & 0 \\
		-8x & 1 & 0 & -8x \\
		-8x & 0 & 1 & 8x \\
		0 & -8x & 8x & 1
	\end{pmatrix}.
\end{split}
\end{equation}
The primal numerical solution shows that all $2\times2$ partitions are inactive,
and all $1\times3$ partitions are active, so we show only $Z$ matrices of the
latter partitions:
\begin{displaymath}
\begin{alignedat}{2}
	Z^{123\star} &= 
	\begin{pmatrix}
		u & 0 & 0 \\
		0 & v & -w \\
		0 & -w & v
	\end{pmatrix}, &\quad
	Z^{124\star} &= 
	\begin{pmatrix}
		v & 0 & -w \\
		0 & u & 0 \\
		-w & 0 & v
	\end{pmatrix}, \\
	Z^{134\star} &= 
	\begin{pmatrix}
		v & 0 & w \\
		0 & u & 0 \\
		w & 0 & v
	\end{pmatrix}, &\quad
	Z^{234\star} &= 
	\begin{pmatrix}
		v & w & 0 \\
		w & v & 0 \\
		0 & 0 & u
	\end{pmatrix}.
\end{alignedat}
\end{displaymath}
All single-partite $Z$ matrices are equal
\begin{equation}
	Z^{1\star} = Z^{2\star} = Z^{3\star} = Z^{4\star} = \frac{1}{8}
	\begin{pmatrix}
		1
	\end{pmatrix}.
\end{equation}
The numerical values of the parameters of these matrices read as:
\begin{equation}
\begin{alignedat}{2}
	x &= 0.0880..., &\quad u &= 0.0102..., \\
	v &= 0.0676..., &\quad w &= 0.0573....
\end{alignedat}
\end{equation}
Three eigenvalues of the matrix
\begin{equation}
	\begin{pmatrix}
		X^\star[1, 2, 3] & Z^{123\star} \\
		Z^{123\star\mathrm{T}} & P^\star[1, 2, 3]
	\end{pmatrix}
\end{equation}
must be zero, so we have three equations
\begin{equation}\label{eq:eq3-1}
\begin{split}
	8(v+w) &= 1, \quad	u + w = v, \quad
	64(2x^2 + u^2) = 1.
\end{split}
\end{equation}
The primal objective value reads as
\begin{equation}
	\lambda_{\mathrm{min}} = 16x\beta + \alpha -u - 2v - \frac{1}{8}.
\end{equation}
The dual solution is given by
\begin{equation}\label{eq:psol}
\begin{alignedat}{2}
	\gamma^{123\star}_{xx} &= 
	\begin{pmatrix}
		b & -c & -c \\
		-c & a & d \\
		-c & d & a
	\end{pmatrix}, &\quad
	\gamma^{123\star}_{pp} &= 
	\begin{pmatrix}
		b & c & c \\
		c & a & d \\
		c & d & a
	\end{pmatrix}, \\
	\gamma^{124\star}_{xx} &= 
	\begin{pmatrix}
		a & -c & d \\
		-c & b & -c \\
		d & -c & a
	\end{pmatrix}, &\quad
	\gamma^{124\star}_{pp} &= 
	\begin{pmatrix}
		a & c & d \\
		c & b & c \\
		d & c & a
	\end{pmatrix}, \\
	\gamma^{134\star}_{xx} &= 
	\begin{pmatrix}
		a & -c & -d \\
		-c & b & c \\
		-d & c & a
	\end{pmatrix}, &\quad
	\gamma^{134\star}_{pp} &= 
	\begin{pmatrix}
		a & c & -d \\
		c & b & -c \\
		-d & -c & a
	\end{pmatrix}, \\
	\gamma^{234\star}_{xx} &= 
	\begin{pmatrix}
		a & -d & -c \\
		-d & a & c \\
		-c & c & b
	\end{pmatrix}, &\quad
	\gamma^{234\star}_{pp} &= 
	\begin{pmatrix}
		a & -d & c \\
		-d & a & -c \\
		c & -c & b
	\end{pmatrix},
\end{alignedat}
\end{equation}
where the numerical values of the parameters read as
\begin{equation}
	\begin{alignedat}{2}
		a &= 3.3107..., &\quad b &= 6.1215..., \\
		c &= 4.3141..., &\quad d &= 2.81073....
	\end{alignedat}
\end{equation}
The single-partite matrices $\gamma^{i\star}_{xx}$ and $\gamma^{i\star}_{pp}$
are all equal to the CM of the vacuum state, $(\begin{smallmatrix} 1/2
\end{smallmatrix}$), for $i = 1, \ldots, 4$. The 2$\times$2 partitions do not
participate in the optimal solution and the 1$\times$3 partitions have equal
contribution to it,
\begin{equation}
\begin{split}
	&p^\star_{12|34} = p^\star_{13|24} = p^\star_{14|23} = 0, \\
	&p^\star_{1|234} = p^\star_{2|134} = p^\star_{3|124} = p^\star_{4|123} = \frac{1}{4}.
\end{split}
\end{equation}
The KKT conditions that the matrices 
\begin{equation}
	\mathcal{P}(\gamma^{123\star}_{xx}, \gamma^{123\star}_{pp}), \quad
	\mathcal{P}(\gamma^{4\star}_{xx}, \gamma^{4\star}_{pp})
\end{equation}
have 3 and 1 zero eigenvalues, respectively, give the following equations:
\begin{equation}\label{eq:eq3-2}
	2(a-d) = 1, \quad a + d = b, \quad 4(b^2-2c^2) = 1.
\end{equation}
The KKT condition
\begin{equation}\label{eq:kkt-1}
	\mathcal{P}(\gamma^{123\star}_{xx}, \gamma^{123\star}_{pp})
	\begin{pmatrix}
		X^\star[1, 2, 3] & -Z^{123\star} \\
		-Z^{123\star\mathrm{T}} & P^\star[1, 2, 3]
	\end{pmatrix}
	= 0
\end{equation}
produces the equations
\begin{equation}\label{eq:eq3-3}
\begin{alignedat}{3}
	4(4cx + u) &= b, &\quad 8bx &= c, &\quad 16bu &= 1, \\
	\quad 2c(v-w) &= x,	&\quad 4(2cx+v) &= a, &\quad d+4w &= 8cx, \\
	2cu &= x, &\quad 16(av - dw) &= 1, &\quad dv-aw &= 0.
\end{alignedat}
\end{equation}
The KKT condition
\begin{equation}
	\mathcal{P}(\gamma^{4\star}_{xx}, \gamma^{4\star}_{pp})
	\begin{pmatrix}
		X^\star[4] & -Z^{4\star} \\
		-Z^{4\star\mathrm{T}} & P^\star[4]
	\end{pmatrix}
	= 0
\end{equation}
is already satisfied, so it gives no new equations. Due to some kind of
symmetry, the other three states give exactly the same equations. The KKT
condition
\begin{equation}
	X^\star \left(\gamma_{xx} - \frac{1}{4}\gamma^{123\star}_{xx} \oplus 
	\gamma^{4\star}_{xx} - \ldots\right) = \lambda_{\mathrm{min}} X^\star,
\end{equation}
where dots denote the sum of the other three 1$\times$3 combinations, produces
the following two equations:
\begin{equation}\label{eq:eq3-4}
\begin{split}
	4(u + 2 v + 8 c x) &= 2 a + b, \\
	4 x (2 a + b + 64 x \beta) &= c + 2 (8 u x + 16 v x + \beta).
\end{split}
\end{equation}
The similar condition with $P^\star$ and $pp$-parts gives exactly the same two
equations again. The combined system of equations given by
Eqs.~\eqref{eq:eq3-1}, \eqref{eq:eq3-2}, \eqref{eq:eq3-3} and \eqref{eq:eq3-4}
has two solutions, one of which numerically coincides with the numbers produced
by the solver. The explicit expressions are simpler than in the previous
examples, the primal solution reads as:
\begin{equation}
\begin{split}
	x^{\star 2} &= \frac{1023971841 - 5120288\sqrt[3]{2} -230404\sqrt[3]{4}}{64\times 2047972354}, \\
	u^\star &= \frac{2\sqrt{2}}{31999}(40 + \sqrt[3]{2} + 800\sqrt[3]{4}) x^\star, \\
	v^\star &= \frac{\sqrt{2}}{31999}(40 + \sqrt[3]{2} + 800\sqrt[3]{4}) x^\star + \frac{1}{16}. \\
\end{split}
\end{equation}
The dual solution is given by
\begin{equation}
\begin{split}
	b^{\star2} &= \frac{1}{16}\left(400\sqrt[3]{4} + \frac{1}{\sqrt[3]{4}} - 36\right), \\
	a^\star &= \frac{2b^\star + 1}{4}, \quad
	d^\star = \frac{2b^\star - 1}{4}, \\
	c^\star &= \frac{1}{\sqrt[6]{2}}\left(5 - \frac{1}{4\sqrt[3]{4}}\right).
\end{split}
\end{equation}
The optimal value reads as
\begin{displaymath}
	-16\lambda^\star_{\mathrm{min}} = \sqrt{2}\sqrt{800\sqrt[3]{4} + \sqrt[3]{2} - 72}
	-\sqrt[3]{4}- 40\sqrt[3]{2} + 4.
\end{displaymath}
We also need to check that the 2$\times$2 partitions are indeed inactive. Due to the equality
\begin{equation}
	\str_{12|34}(X^\star, P^\star) = \str_{13|24}(X^\star, P^\star),
\end{equation}
we need to check only two of the three quantities. We have
\begin{equation}
\begin{split}
	\tr(X^\star\gamma_{xx} &+ P^\star\gamma_{pp}) - \str_{12|34}(X^\star, P^\star) \\
	&= \alpha + 16\beta x^\star - \frac{1}{2}\sqrt{1-64x^{\star2}}, \\
	\tr(X^\star\gamma_{xx} &+ P^\star\gamma_{pp}) - \str_{14|23}(X^\star, P^\star) \\
	&= \alpha + 16\beta x^\star - \frac{1}{2}.
\end{split}
\end{equation}
\textsl{Mathematica} can symbolically prove that these two quantities are
smaller than $\lambda^\star_{\mathrm{min}}$, and thus the 2$\times$2 partitions
are indeed inactive. One can prove this fact ``by hand'' by squaring and
subtracting the left-hand side from the right-hand side several times, but this task
is what CASs are good at. We arrive to the following result:
\begin{equation}
	\mathcal{E}^+_{\mathfrak{I}_2}(\gamma_{xx}, \gamma_{pp}) = -\lambda^\star_{\mathrm{min}}
	 = 0.0618...,
\end{equation}
which verifies that the state is genuine multipartite entangled. Note that the
matrices $X^\star$ and $P^\star$ are non-degenerate and thus have rank 4.

\subsubsection{Scaling}

The primal solution has exactly the same form as given by Eq.~\eqref{eq:psol},
but with different values for the parameters
\begin{equation}
	\begin{alignedat}{2}
		a &= 4.2229..., &\quad b &= 7.9458..., \\
		c &= 5.6074..., &\quad d &= 3.7229...,
	\end{alignedat}
\end{equation}
so the first set of equations is given by Eq.~\eqref{eq:eq3-2}. The dual
solutions read as
\begin{equation}\label{eq:sXY}
	X^\star = 
	\begin{pmatrix}
		x & y & y & 0 \\
		y & x & 0 & y \\
		y & 0 & x & -y \\
		0 & y & -y & x
	\end{pmatrix}, \quad
	P^\star = 
	\begin{pmatrix}
		x & -y & -y & 0 \\
		-y & x & 0 & -y \\
		-y & 0 & x & y \\
		0 & -y & y & x
	\end{pmatrix},
\end{equation}
where the parameters are given by
\begin{equation}
	x = 0.4704..., \quad y = 0.3319....
\end{equation}
The $Z$ matrices read as
\begin{equation}\label{eq:Z1+3-3}
\begin{split}
	Z^{123\star} &= \frac{1}{4}
	\begin{pmatrix}
		4u & 0 & 0 \\
		0 & 1 & -4v \\
		0 & -4v & 1
	\end{pmatrix}, \\
	Z^{124\star} &= \frac{1}{4}
	\begin{pmatrix}
		1 & 0 & -4v \\
		0 & 4u & 0 \\
		-4v & 0 & 1
	\end{pmatrix}, \\
	Z^{134\star} &= \frac{1}{4}
	\begin{pmatrix}
		1 & 0 & 4v \\
		0 & 4u & 0 \\
		4v & 0 & 1
	\end{pmatrix}, \\
	Z^{123\star} &= \frac{1}{4}
	\begin{pmatrix}
		1 & 4v & 0 \\
		4v & 1 & 0 \\
		0 & 0 & 4u
	\end{pmatrix},
\end{split}
\end{equation}
and 
\begin{equation}\label{eq:Z1+3-1}
	Z^{1\star} = Z^{2\star} = Z^{3\star} = Z^{4\star} = 
	\begin{pmatrix}
		x
	\end{pmatrix},
\end{equation}
where the values of the parameters are
\begin{equation}
	u = 0.0295..., \quad v = 0.2204....
\end{equation}
The matrices $Z^{ij\star}$ are all the same and read as
\begin{equation}\label{eq:Z2+2}
	Z^{ij\star} = \frac{1}{4} E_2.
\end{equation}
Clearly, we have 
\begin{equation}
	\tr Z^{12\star} + \tr Z^{34\star} = 1,
\end{equation}
and two other similar equalities for $2+2$ bipartitions. The condition
\begin{equation}
	\tr Z^{123\star} + \tr Z^{4\star} = 1
\end{equation}
and similar for $1+3$ bipartitions produce the equation
\begin{equation}\label{eq:tr1}
	u + x + \frac{1}{2} = 1.
\end{equation}
The condition that the matrix
\begin{equation}
	\begin{pmatrix}
		X^\star[1, 2, 3] & Z^{123\star} \\
		Z^{123\star\mathrm{T}} & P^\star[1, 2, 3]
	\end{pmatrix}
\end{equation}
has three zero eigenvalues produces the equations
\begin{equation}\label{eq:eq4-2}
\begin{split}
	4(x-v) = 1, \quad 4(u+v) &= 1, \\
	x - 2 u - 4 v x + 4 x^2 - 8 y^2 &= 0.
\end{split}
\end{equation}
The other three 1$\times$3 partitions produce identical equations. The KKT
condition \eqref{eq:kkt-1} gives the following equations:
\begin{equation}\label{eq:eq4-3}
\begin{alignedat}{2}
	2(bx-2cy) &= u, &\quad cx - by &= 0, \\
	2(2cv + y) &= c, &\quad 2bu -x &= 0, \\
	8(ax-cy) &= 1, &\quad (a+d)y &= cx, \\
	2(cy-dx) &= v, &\quad 2cu-y &= 0, \\
	2(2dv+x) &= a, &\quad 4av - d &= 0.
\end{alignedat}
\end{equation}
The KKT condition \eqref{eq:kkt-2} is already satisfied. The KKT condition
\eqref{eq:KKTscaling2} in this case reads as 
\begin{displaymath}
	\begin{split}
		&X^\star\left(\gamma_{xx} - \frac{t}{4}\gamma^{123\star}_{xx}\oplus\gamma^{4\star}_{xx}
		- \ldots \right) = 0, \\
		&P^\star\left(\gamma_{pp} - \frac{t}{4}\gamma^{123\star}_{pp}\oplus\gamma^{4\star}_{pp}
		- \ldots\right) = 0,
	\end{split}
\end{displaymath}
where the dots stand for the other three 1$\times$3 partitions and
\begin{equation}
	t = \tr(X^\star \gamma_{xx} + P^\star \gamma_{pp}) = 8(\alpha x + 2 \beta y).
\end{equation}
These conditions produce the following additional equations:
\begin{equation}\label{eq:eq4-4}
\begin{split}
	4 a x + 2 b x - 8 c y + x - 1 &= 0, \\
	(2 (1 + 4 a + 2 b) y^2-x - 8 c x y)\beta &= (8 c y^2-4 c x^2) \alpha.
\end{split}
\end{equation}
The equations \eqref{eq:eq3-2}, \eqref{eq:tr1}, \eqref{eq:eq4-2},
\eqref{eq:eq4-3} and \eqref{eq:eq4-4} together form a system that has a solution
which numerically agrees with the values produced by the solver. The primal
solution reads as
\begin{equation}
\begin{split}
	x^\star &= \frac{341290667 - 17065600\sqrt[3]{2} + 853320\sqrt[3]{4}}{682709334}, \\
	y^\star &= \sqrt[6]{2}\frac{80 + 31999\sqrt[3]{2} - 1600\sqrt[3]{4}}{128004}, \\
	u^\star &= 20\frac{1600 + 426640\sqrt[3]{2} -21333\sqrt[3]{4}}{341354667}, \\
	v^\star &= \frac{341226667 - 34131200\sqrt[3]{2} +1706640\sqrt[3]{4}}{1365418668}.
\end{split}
\end{equation}
The dual solution is simpler
\begin{equation}
\begin{split}
	a^\star &= \frac{80 + \sqrt[3]{2} + 800\sqrt[3]{4}}{320}, \\
	b^\star &= 5\sqrt[3]{2} + \frac{1}{80\sqrt[3]{2}}, \quad
	c^\star = 5\sqrt[6]{2} - \frac{1}{160\sqrt[6]{2}}, \\
	d^\star &= \frac{-80 + \sqrt[3]{2} + 800\sqrt[3]{4}}{320}.
\end{split}
\end{equation}
The optimal scaling factor $t^\star$ is equal to
\begin{displaymath}
	t^\star = u^\star + x^\star + \frac{1}{2} = 20\frac{-40+\sqrt[3]{2}+800\sqrt[3]{2}}{32001} = 0.7694....
\end{displaymath}
We arrive to the following result:
\begin{equation}
	\mathcal{T}^+_{\mathfrak{I}_2}(\gamma_{xx}, \gamma_{pp}) = 1-t^\star_{\mathrm{min}}
	 = 0.2305...,
\end{equation}
which verifies that the state is genuine multipartite entangled. The values for
$x^\star$, $y^\star$ and $t^\star$ were obtained in Ref.~\cite{NewJPhys.8.51} in
a purely numerical way and without any evidence of their optimality. Here we
again constructed analytical expressions for these numbers together with an
analytical validation of their optimality. Note that the matrices $X^\star$ and
$P^\star$ are non-degenerate and thus have rank 4.

\subsection{Parametric solution}\label{sec:secPar}

The state with the CM given by Eq.~\eqref{eq:G4} can, in fact, be solved for
arbitrary $\alpha$ and $\beta$, not only for the values given by
Eq.~\eqref{eq:Galphabeta}. The eigenvalues of $\mathcal{P}(\gamma_{xx},
\gamma_{pp})$ are given by
\begin{equation}
	\alpha \pm \frac{1}{2}\sqrt{1+8\beta^2}
\end{equation}
with multiplication 4, so this CM is physical iff
\begin{equation}\label{eq:alphabeta}
	\beta^2 \leqslant \frac{1}{2}\left(\alpha^2-\frac{1}{4}\right).
\end{equation}
The condition $\alpha \geqslant 1/2$ is then automatically satisfied. This state
is pure exactly when the inequality \eqref{eq:alphabeta} becomes equality. For
$\alpha = 1/2$ the only possible value for $\beta$ is zero and this CM is then
the CM of the four-mode vacuum state. We derive an entanglement condition on
$\alpha$ and $\beta$ for the general state. As in the previous parts for
concrete states, we construct an analytical solution of both problems.

\subsubsection{Eigenvales}

All the machinery of constructing and solving the system of KKT equations has
already been developed before, so here we just present the final result. The
primal solution reads as
\begin{equation}
	\begin{split}
		X^\star &= \frac{1}{8}
		\begin{pmatrix}
			1 & 8x^\star & 8x^\star & 0 \\
			8x^\star & 1 & 0 & 8x^\star \\
			8x^\star & 0 & 1 & -8x^\star \\
			0 & 8x^\star & -8x^\star & 1
		\end{pmatrix}, \\
		P^\star &= \frac{1}{8}
		\begin{pmatrix}
			1 & -8x^\star & -8x^\star & 0 \\
			-8x^\star & 1 & 0 & -8x^\star \\
			-8x^\star & 0 & 1 & 8x^\star \\
			0 & -8x^\star & 8x^\star & 1
		\end{pmatrix},
	\end{split}
\end{equation}
where $x^\star$ is given by
\begin{eqnarray}
	x^\star = -\frac{\beta}{2\sqrt{1+32\beta^2}}.
\end{eqnarray}
The eigenvalues of these matrices are the same and are given by
\begin{equation}
	\frac{1}{8} \pm \frac{\beta}{\sqrt{2(1+32\beta^2)}} > 0,
\end{equation}
with multiplicity 2, so $X^\star$ and $P^\star$ are positive-definite and
obviously satisfy the normalization condition $\tr(X^\star + P^\star) = 1$. 

In the dual solution only the 1$\times$3 bipartitions are active, so that 
\begin{equation}
	p^\star_{123|4} = p^\star_{124|3} = p^\star_{134|2} = p^\star_{234|1} = \frac{1}{4},
\end{equation}
and coefficients of the other bipartitions are zero. The states corresponding to
non-zero coefficients read as
\begin{equation}\label{eq:gamma123}
	\begin{alignedat}{2}
		\gamma^{123\star}_{xx} &= 
		\begin{pmatrix}
			b^\star & -c^\star & -c^\star \\
			-c^\star & a^\star & d^\star \\
			-c^\star & d^\star & a^\star
		\end{pmatrix}, &\ 
		\gamma^{123\star}_{pp} &= 
		\begin{pmatrix}
			b^\star & c^\star & c^\star \\
			c^\star & a^\star & d^\star \\
			c^\star & d^\star & a^\star
		\end{pmatrix}, \\
		\gamma^{124\star}_{xx} &= 
		\begin{pmatrix}
			a^\star & -c^\star & d^\star \\
			-c^\star & b^\star & -c^\star \\
			d^\star & -c^\star & a^\star
		\end{pmatrix}, &\ 
		\gamma^{124\star}_{pp} &= 
		\begin{pmatrix}
			a^\star & c^\star & d^\star \\
			c^\star & b^\star & c^\star \\
			d^\star & c^\star & a^\star
		\end{pmatrix}, \\
		\gamma^{134\star}_{xx} &= 
		\begin{pmatrix}
			a^\star & -c^\star & -d^\star \\
			-c^\star & b^\star & c^\star \\
			-d^\star & c^\star & a^\star
		\end{pmatrix}, &\  
		\gamma^{134\star}_{pp} &= 
		\begin{pmatrix}
			a^\star & c^\star & -d^\star \\
			c^\star & b^\star & -c^\star \\
			-d^\star & -c^\star & a^\star
		\end{pmatrix}, \\
		\gamma^{234\star}_{xx} &= 
		\begin{pmatrix}
			a^\star & -d^\star & -c^\star \\
			-d^\star & a^\star & c^\star \\
			-c^\star & c^\star & b^\star
		\end{pmatrix}, &\ 
		\gamma^{234\star}_{pp} &= 
		\begin{pmatrix}
			a^\star & -d^\star & c^\star \\
			-d^\star & a^\star & -c^\star \\
			c^\star & -c^\star & b^\star
		\end{pmatrix},
	\end{alignedat}
\end{equation}
where the parameters are 
\begin{equation}
\begin{split}
	a^\star &= \frac{1}{4}\left(\sqrt{1+32\beta^2} + 1\right), \quad
	b^\star = \frac{1}{2}\sqrt{1+32\beta^2}, \\
	c^\star &= -2\beta, \quad d^\star = \frac{1}{4}\left(\sqrt{1+32\beta^2} - 1\right).
\end{split}
\end{equation}
The single-partite matrices $\gamma^{i\star}_{xx}$ and $\gamma^{i\star}_{pp}$
are all equal to the CM of the vacuum state. Eigenvalues of
$\mathcal{P}(\gamma^{123\star}_{xx}, \gamma^{123\star}_{pp})$ and of the other
three 1$\times$3 partitions read as
\begin{equation}
	(0, 0, 0, 1, \sqrt{1+32\beta^2}, \sqrt{1+32\beta^2}),
\end{equation}
so these are physical CMs.

It is easy to see that 
\begin{displaymath}
\begin{split}
	\lambda^\star_{\mathrm{min}} &\equiv \tr(X^\star\gamma_{xx} + P^\star\gamma_{pp}) 
	-\str(X^\star[1, 2, 3], P^\star[1, 2, 3]) \\
	&-\str(X^\star[4], P^\star[4]) = -\frac{1}{4}(1-4\alpha + \sqrt{1+32\beta^2}),
\end{split}
\end{displaymath}
and the other three 1$\times$3 bipartitions give the same value. For the
2$\times$2 bipartitions we have different values. For $12|34$ we have
\begin{displaymath}
	\begin{split}
		\tr(X^\star\gamma_{xx} &+ P^\star\gamma_{pp}) 
		-\str(X^\star[1, 2], P^\star[1, 2]) \\
		&-\str(X^\star[3, 4], P^\star[3, 4]) \\
		&= \alpha-\frac{16\beta^2 +\sqrt{1+16\beta^2}}{2\sqrt{1+32\beta^2}} 
		< \lambda^\star_{\mathrm{min}}.
	\end{split}
\end{displaymath}
The case of partition $13|24$ is identical to $12|34$, and for $14|23$ we have
\begin{displaymath}
	\begin{split}
		\tr(X^\star\gamma_{xx} &+ P^\star\gamma_{pp}) 
		-\str(X^\star[1, 4], P^\star[1, 4]) \\
		&-\str(X^\star[2, 3], P^\star[2, 3]) \\
		&= -\frac{1}{2}\left(1-2\alpha+\frac{16\beta^2}{\sqrt{1+32\beta^2}}\right) 
		< \lambda^\star_{\mathrm{min}}.
	\end{split}
\end{displaymath}
We see that $\lambda^\star_{\mathrm{min}}$ is the maximal value among the values
obtained with all bipartitions. On the other hand, for the dual solution we have
\begin{equation}
\begin{split}
	\gamma_{xx} - \frac{1}{4}\gamma^{123\star}_{xx} \oplus \gamma^{4\star}_{xx} - \ldots
	&= \lambda^\star_{\mathrm{min}} E_4, \\
	\gamma_{pp} - \frac{1}{4}\gamma^{123\star}_{pp} \oplus \gamma^{4\star}_{pp} - \ldots
	&= \lambda^\star_{\mathrm{min}} E_4, 
\end{split}
\end{equation}
where the dots stand for the combinations corresponding to the rest of
1$\times$3 bipartitions. These equalities demonstrate strong duality between the
constructed primal and dual solutions, verifying that they are indeed optimal.

For the entanglement measure, we have
\begin{equation}
	\mathcal{E}_{\mathfrak{I}_2}(\gamma) = -\lambda^\star_{\mathrm{min}} = 
	\frac{1}{4}(1-4\alpha + \sqrt{1+32\beta^2}).
\end{equation}
Combining the condition $\mathcal{E}^+(\gamma) > 0$ with the physicality
condition \eqref{eq:alphabeta}, we derive the following condition for the
entanglement:
\begin{equation}\label{eq:eab}
	\frac{1}{4} \leqslant \alpha^2 - 2\beta^2 < \frac{\alpha}{2}.
\end{equation}
The larger $\beta^2$, the larger the entanglement measure, and for the maximal
allowed value given by Eq.~\eqref{eq:alphabeta} we have
\begin{equation}
	\mathcal{E}_{\mathfrak{I}_2}(\gamma) = \frac{1}{4}(1-4\alpha + \sqrt{16\alpha^2-3}).
\end{equation}
The plot of this quantity as a function of $\alpha$ is shown in
Fig.~\ref{fig:Ea}. We see that this pure state is entangled for all $\alpha >
1/2$ and its entanglement measure tends to $1/4$ when $\alpha\to+\infty$. Note
that the matrices $X^\star$ and $P^\star$ are always non-degenerate and thus
have rank 4.

\begin{figure}
	\includegraphics[scale=0.9]{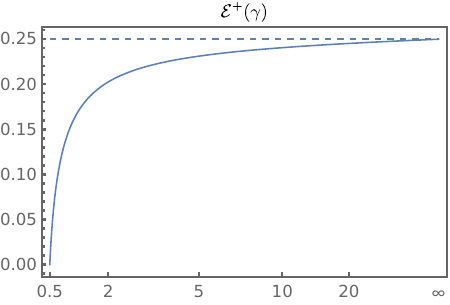}
	\caption{$\mathcal{E}^+$ of the pure state \eqref{eq:G4} as a function of
	$\alpha$ with $\beta$ given by the right-hand side of \eqref{eq:alphabeta}.}
	\label{fig:Ea}
\end{figure}

\subsubsection{Scaling}

\begin{figure}
	\includegraphics[scale=0.9]{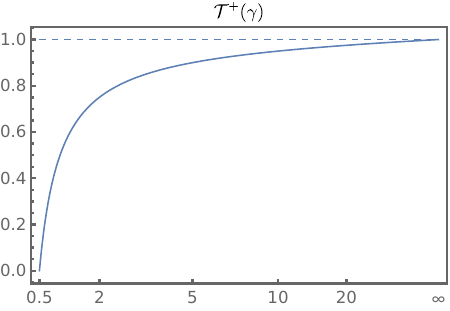}
	\caption{$\mathcal{T}^+$ of the pure state \eqref{eq:G4} as a function of
	$\alpha$ with $\beta$ given by the right-hand side of \eqref{eq:alphabeta}.}
	\label{fig:Ta}
\end{figure}

\begin{figure*}[t!]
	\hspace*{-8mm}
	\begin{subfigure}[t]{0.65\textwidth}
	    \includegraphics[scale=0.9]{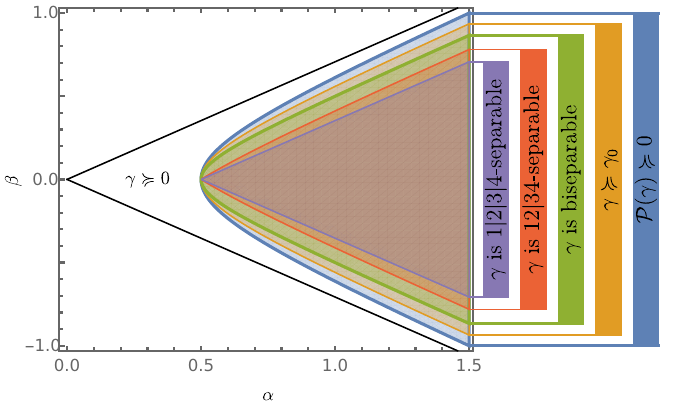}
	    \caption{The regions defined by Eqs.~\eqref{eq:region-1}-\eqref{eq:region-6}.}\label{fig:regions}
	\end{subfigure}
    \hspace*{-10mm}
	\begin{subfigure}[t]{0.42\textwidth}
	    \includegraphics[scale=0.85]{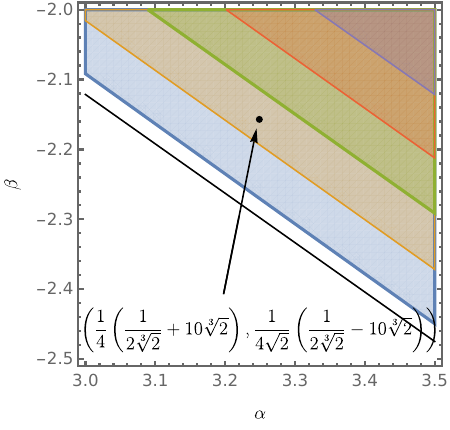}
	    \caption{The location of the point given by Eq.~\eqref{eq:Galphabeta}.}\label{fig:regions2}
	\end{subfigure}
	\caption{The physicality region and the regions of various separability properties of the state given by Eq.~\eqref{eq:G4}.}
\end{figure*}

The primal solution in this case has the form of Eq.~\eqref{eq:gamma123}, but
the parameters are different and read as
\begin{equation}
\begin{alignedat}{2}
	a^\star &= \frac{\alpha^2}{2(\alpha^2 - 2\beta^2)}, &\quad b^\star &= \frac{\alpha^2 + 2\beta^2}{2(\alpha^2 - 2\beta^2)}, \\
	c^\star &= -\frac{\alpha\beta}{\alpha^2 - 2\beta^2}, &\quad d^\star &= \frac{\beta^2}{\alpha^2 - 2\beta^2}.
\end{alignedat}
\end{equation}
Note that from the physicality of the state, the first inequality of
Eq.~\eqref{eq:eab}, it follows that the denominators in these expressions are
strictly positive. The dual solution is of the form given by Eq.~\eqref{eq:sXY}
where the parameters read as
\begin{equation}
	x^\star = \frac{\beta^2}{2\alpha^2} + \frac{1}{4}, \quad y^\star = -\frac{\beta}{2\alpha}.
\end{equation}
The dual solution is given by Eqs.~\eqref{eq:Z1+3-3}, \eqref{eq:Z1+3-1} and
\eqref{eq:Z2+2}, where 
\begin{equation}
	u^\star = \frac{1}{4} - \frac{\beta^2}{2\alpha^2}, \quad
	v^\star = \frac{\beta^2}{2\alpha^2}.
\end{equation}
The optimal scaling parameter is given by
\begin{equation}
	t^\star = \tr(X^\star \gamma_{xx} + P^\star \gamma_{pp}) = 2\frac{\alpha^2-2\beta^2}{\alpha}.
\end{equation}
Simple algebraic manipulations show that the entanglement condition $t^\star <
1$, together with the physicality condition, again produce Eq.~\eqref{eq:eab}.
For the measure $\mathcal{T}^+$ we have
\begin{equation}
	\mathcal{T}_{\mathfrak{I}_2}(\gamma) = 1 - t^\star = 1 - 2\alpha + 4\frac{\beta^2}{\alpha}.
\end{equation}
This measure is maximal when the state is pure and is given by
\begin{equation}
	\mathcal{T}_{\mathfrak{I}_2}(\gamma) = 1 - \frac{1}{2\alpha}.
\end{equation}
The plot of this quantity is shown in Fig.~\ref{fig:Ta} and looks very similar
to the plot of $\mathcal{E}_{\mathfrak{I}_2}(\gamma)$ for the same state, Fig.~\ref{fig:Ea}.

We now show that the refined condition \eqref{eq:gammatestt3} is not necessary and
sufficient when $|\mathfrak{I}| > 1$. As a counter example we take the CM
\eqref{eq:G4} and $\mathfrak{I} = \{1|234, 2|134, 3|124, 4|123\}$. As
single-mode states we take the vacuum 
\begin{equation}
	\gamma^{1} = \gamma^{2} = \gamma^{3} = \gamma^{4} = \frac{1}{2}
	\begin{pmatrix}
		1 & 0 \\
		0 & 1
	\end{pmatrix},
\end{equation}
and as the coefficients we take $p_1 = p_2 = p_3 = p_4 = 1/4$. We then have
\begin{equation}
\begin{split}
	\gamma_0 &= p_1 \gamma^{1} \oplus \frac{i}{2}\Omega^{234}_3 + \ldots + p_4 \gamma^{4} \oplus \frac{i}{2}\Omega^{123}_3 \\
	&= \frac{1}{8}
	\begin{pmatrix}
		1 & 0 & 0 & 0 & 3i & 0 & 0 & 0 \\
		0 & 1 & 0 & 0 & 0 & 3i & 0 & 0 \\
		0 & 0 & 1 & 0 & 0 & 0 & 3i & 0 \\
		0 & 0 & 0 & 1 & 0 & 0 & 0 & 3i \\
		-3i & 0 & 0 & 0 & 1 & 0 & 0 & 0 \\
		0 & -3i & 0 & 0 & 0 & 1 & 0 & 0 \\
		0 & 0 & -3i & 0 & 0 & 0 & 1 & 0 \\
		0 & 0 & 0 & -3i & 0 & 0 & 0 & 1
	\end{pmatrix}.
\end{split}
\end{equation}
The eigenvalues of the difference are
\begin{equation}
	\lambda_\pm = \frac{1}{8}(8\alpha-1 \pm \sqrt{128\beta^2+9})
\end{equation}
with multiplicity 4. It follows that
\begin{equation}
	\gamma \succcurlyeq \gamma_0 \ \Leftrightarrow\ \lambda_- \geqslant 0 \ 
	\Leftrightarrow\ \beta^2 \leqslant \frac{\alpha^2}{2} - \frac{2\alpha+1}{16}.
\end{equation}
We thus derived three conditions: the condition for physicality of $\gamma$
reads as
\begin{equation}\label{eq:region-1}
	\beta^2 \leqslant \frac{\alpha^2}{2} - \frac{1}{8},
\end{equation}
the condition for biseparability of $\gamma$ is given by
\begin{equation}\label{eq:region-2}
	\beta^2 \leqslant \frac{\alpha^2}{2} - \frac{\alpha}{4},
\end{equation}
and the condition for the inequality $\gamma \succcurlyeq \gamma_0$ is
\begin{equation}\label{eq:region-3}
	\beta^2 \leqslant \frac{\alpha^2}{2} - \frac{2\alpha+1}{16}.
\end{equation}
In addition, one can also derive the conditions for full $1|2|3|4$-separability
and for $12|34$-separability. We present here only the final results: the state
is fully separable iff
\begin{equation}\label{eq:region-4}
	|\beta| \leqslant \frac{1}{\sqrt{2}} \left(\alpha - \frac{1}{2}\right),
\end{equation}
and it is $12|34$-separable iff
\begin{equation}\label{eq:region-5}
	|\beta| \leqslant \frac{1}{4}(\sqrt{8\alpha^2-1} - 1).
\end{equation}
We also add the condition $\gamma \succcurlyeq 0$ here:
\begin{equation}\label{eq:region-6}
	|\beta| \leqslant \frac{\alpha}{\sqrt{2}}.
\end{equation}
These six regions are shown in Fig.~\ref{fig:regions}. The ``strip'' between the
blue and the green lines is the domain where $\gamma$ is genuine multipartite
entangled. We emphasize that this figure is not schematic like a Venn diagram,
which shows different sets for pure illustrative purpose only, it shows real
sets of different kinds of separability. Figures \ref{fig:Capprox} and
\ref{fig:cs} are also not purely abstract and correspond to a realistic
situation. If we work with states of the form given by Eq.~\eqref{eq:G4}, then
the former figure illustrates finding the best physical point $(\alpha^\star,
\beta^\star)$ corresponding to a given non-physical one $(\alpha^\circ,
\beta^\circ)$, and the latter figure illustrates the fully separable state
closest to a given state.

As Fig.~\ref{fig:regions} illustrates, this domain intersects with the region
where the inequality $\gamma \succcurlyeq \gamma_0$ holds true, the intersection
being the ``strip'' between the yellow and the green lines. The point given by
Eq.~\eqref{eq:Galphabeta} is located exactly in this ``gray'' zone.
Consequently, the new optimization problem does not say that the corresponding
CM is entangled:
\begin{equation}
	t^\star = 1.2847... > 1.
\end{equation}
In the domain where the new problem shows a violation of the new separability
condition, the ``strip'' between the blue and the yellow lines, the state is
definitely genuine entangled.

\subsection{Tripartite symmetric Gaussian and non-Gaussian states}\label{sec:X}

Here we consider a reduced tripartite state with the following CM:
\begin{equation}\label{eq:CM3S}
    \gamma_{xx} = 
    \begin{pmatrix}
        \alpha & \beta & \beta \\
        \beta & \alpha & \beta \\
        \beta & \beta & \alpha
    \end{pmatrix}, \quad
    \gamma_{pp} = 
    \begin{pmatrix}
        \alpha & -\beta & -\beta \\
        -\beta & \alpha & -\beta \\
        -\beta & -\beta & \alpha
    \end{pmatrix}.
\end{equation}
This matrix is positive semidefinite, $\gamma \succcurlyeq 0$, iff 
\begin{equation}
    |\beta| \leqslant \frac{\alpha}{2},
\end{equation}
and it is a physical CM provided that
\begin{equation}
    |\beta| \leqslant \frac{1}{4}\sqrt{4\alpha^2-1}.
\end{equation}
This covariance matrix is symmetric, so it is enough to consider only one
bipartition, for example $1|23$, and any result will be equally valid for the
other two bipartitions. Below we will also consider combinations of two
bipartitions, and it is enough to study only one combination, say $\{1|23,
2|13\}$. Here we present just the final results obtained by solving the
eigenvalue optimization problem with the approach presented in details in the
previous sections. Changing $\beta \to -\beta$ just exchanges $\gamma_{xx}
\leftrightarrow \gamma_{pp}$, so it suffices to consider just the case of $\beta
\geqslant 0$. Because the presentation of the results is too short to put them
into separate subsections, we use lemmas to clearly separate where each
individual result starts and ends.

Our first result is 
\begin{lmn}
The Gaussian state with the CM given by Eq.~\eqref{eq:CM3S} is $1|2|3$-separable
iff the parameters $\alpha$ and $\beta$ satisfy the relation
\begin{equation}\label{eq:1|2|3-sep}
    3|\beta| + \sqrt{\beta^2 + 1} \leqslant 2\alpha.
\end{equation}
\end{lmn}
\begin{proof}
If $\beta \geqslant 0$ then the optimal matrices $X^\star$ and $P^\star$ read as
\begin{equation}
    X^\star = 
    \begin{pmatrix}
        x^\star & y^\star & y^\star \\
        y^\star & x^\star & y^\star \\
        y^\star & y^\star & x^\star
    \end{pmatrix}, \quad
    P^\star =
    \begin{pmatrix}
        p^\star & p^\star & p^\star \\
        p^\star & p^\star & p^\star \\
        p^\star & p^\star & p^\star
    \end{pmatrix}, 
\end{equation}
where
\begin{equation}
\begin{split}
    x^\star &= \frac{1}{6}\left(1 - \frac{|\beta|}{\sqrt{\beta^2+1}}\right), \\
    y^\star &= -\frac{1}{12}\left(1 - \frac{|\beta|}{\sqrt{\beta^2+1}}\right), \\
    p^\star &= \frac{1}{6}\left(1 + \frac{|\beta|}{\sqrt{\beta^2+1}}\right).
\end{split}
\end{equation}
If $\beta < 0$ then $X^\star$ and $P^\star$ are swapped. The eigenvalues of
$X^\star$ are 0 and
\begin{equation}
    \frac{1}{4}\left(1-\frac{|\beta|}{\sqrt{\beta^2+1}}\right) > 0,
\end{equation}
with multiplicity 2. The eigenvalues of $P^\star$ are 0 with multiplicity 2 and 
\begin{equation}
    \frac{1}{2}\left(1+\frac{|\beta|}{\sqrt{\beta^2+1}}\right) > 0,
\end{equation}
so $X^\star, P^\star \succ 0$.

The optimal dual solution read as
\begin{equation}
\begin{split}
    \gamma^{1\star}_{xx} &= \gamma^{2\star}_{xx} = \gamma^{3\star}_{xx} = 
    \begin{pmatrix}
        a^\star
    \end{pmatrix}, \\ 
    \gamma^{1\star}_{pp} &= \gamma^{2\star}_{pp} = \gamma^{3\star}_{pp} = 
    \begin{pmatrix}
        b^\star
    \end{pmatrix},
\end{split}
\end{equation}
where
\begin{equation}
    a^\star = \frac{1}{2}(\sqrt{\beta^2+1} + \beta), \quad
    b^\star = \frac{1}{2}(\sqrt{\beta^2+1} - \beta).
\end{equation}
These states are single-mode squeezed vacuum states with the degree of squeezing
$r$ given by
\begin{equation}
    r = \frac{1}{2} \ln(\sqrt{\beta^2+1} + \beta).
\end{equation}

It is straightforward to verify that 
\begin{equation}\label{eq:abmin}
\begin{split}
    \tr(X^\star\gamma_{xx} &+ P^\star\gamma_{pp}) - 3\str(X^\star[1], P^\star[1]) \\
    &= \alpha - \frac{1}{2}(3|\beta| + \sqrt{\beta^2+1}).
\end{split}
\end{equation}
On the other hand, we have that distinct eigenvalues of both differences
\begin{equation}
    \gamma_{xx} - \gamma^{1\star}_{xx} \oplus \gamma^{2\star}_{xx} \oplus \gamma^{3\star}_{xx}, \quad
    \gamma_{pp} - \gamma^{1\star}_{pp} \oplus \gamma^{2\star}_{pp} \oplus \gamma^{3\star}_{pp},
\end{equation}
are given by
\begin{equation}
    \alpha - \frac{1}{2}(\pm 3\beta + \sqrt{\beta^2 + 1}).
\end{equation}
It follows that the minimal eigenvalue is given by the right-hand side of
Eq.~\eqref{eq:abmin}, which proves the optimality of the presented solutions.

We have just derived that
\begin{equation}
    \mathcal{E}_{1|2|3}(\gamma_{xx}, \gamma_{pp}) = \frac{1}{2}(3|\beta| + \sqrt{\beta^2+1}) - \alpha.
\end{equation}
The $1|2|3$-separability is expressed by the inequality
$\mathcal{E}_{1|2|3}(\gamma_{xx}, \gamma_{pp}) \leqslant 0$, which is exactly
the statement of the lemma, Eq.~\eqref{eq:1|2|3-sep}.
\end{proof}

\begin{lmn}
The Gaussian state with the CM given by Eq.~\eqref{eq:CM3S} is $1|23$-separable
iff its partial transposition with respect to the first mode is positive
semidefinite. This happens provided that the condition
\begin{equation}\label{eq:pos1|23}
    10\beta^2 + 2|\beta|\sqrt{9\beta^2+8} + 1 \leqslant 4\alpha^2.
\end{equation}
holds true.
\end{lmn}
\begin{proof}
The optimal witness matrices read as
\begin{equation}
    X^\star = 
    \begin{pmatrix}
        x^\star & y^\star & y^\star \\
        y^\star & z^\star & z^\star \\
        y^\star & z^\star & z^\star
    \end{pmatrix}, \quad
    P^\star = 
    \begin{pmatrix}
        p^\star & q^\star & q^\star \\
        q^\star & r^\star & r^\star \\
        q^\star & r^\star & r^\star
    \end{pmatrix},
\end{equation}
and the dual solution is given by
\begin{equation}
\begin{split}
    &\gamma^{1\star}_{xx} = 
    \begin{pmatrix}
        a^\star
    \end{pmatrix}, \quad
    \gamma^{23\star}_{xx} = 
    \begin{pmatrix}
        b^\star & c^\star \\
        c^\star & b^\star
    \end{pmatrix}, \\
    &\gamma^{1\star}_{pp} = 
    \begin{pmatrix}
        d^\star
    \end{pmatrix}, \quad
    \gamma^{23\star}_{pp} = 
    \begin{pmatrix}
        e^\star & f^\star \\
        f^\star & e^\star
    \end{pmatrix}.
\end{split}
\end{equation}
The matrix elements are now given by much more complicated expressions. For the
single-mode part we have $a^\star d^\star = 1/4$, so it is a squeezed vacuum
state. With respect to the two-mode part the problem is degenerate and multiple
solutions exist. It is possible to find a pure state among them, i.e. a solution
that satisfies the relation
\begin{equation}
    \gamma^{23\star}_{xx} \gamma^{23\star}_{pp} = \frac{1}{4}\Omega_2.
\end{equation}
The minimal eigenvalue, as produced by \textsl{Mathematica} for $\beta \geqslant
0$, reads as
\begin{displaymath}
\begin{split}
    \lambda^\star_{\mathrm{min}} &= \alpha - 
    \frac{32\beta^3 - 11\beta - \sqrt{9\beta^2+8}}{8\sqrt{2}(2\beta^2-1)} \\
    &\times \sqrt{\frac{\beta(16\beta^2-3)\sqrt{9\beta^2+8}+80\beta^4-23\beta^2+4}{64\beta^4-12\beta^2+1}}.
\end{split}
\end{displaymath}
Though \textsl{Mathematica} does not further simplify it, it turns out that it
is possible to simplify the last term on the right-hand side (which depends on
$\beta$ only).

The CM of the state partially transposed with respect to mode 1 is given by
\begin{equation}
    \gamma^{\mathrm{PT}}_{xx} = 
    \begin{pmatrix}
        \alpha & \beta & \beta \\
        \beta & \alpha & \beta \\
        \beta & \beta & \alpha
    \end{pmatrix}, \quad
    \gamma^{\mathrm{PT}}_{pp} = 
    \begin{pmatrix}
        \alpha & \beta & \beta \\
        \beta & \alpha & -\beta \\
        \beta & -\beta & \alpha
    \end{pmatrix}.
\end{equation}
The eigenvalues of $\mathcal{P}(\gamma^{\mathrm{PT}}_{xx},
\gamma^{\mathrm{PT}}_{pp})$ are
\begin{displaymath}
    \alpha \pm \frac{1}{2}\sqrt{4\beta^2+1}, \quad
    \alpha \pm \frac{1}{2}\sqrt{10\beta^2 \pm 2|\beta|\sqrt{9\beta^2+8}+1}.
\end{displaymath}
We see that the partially transposed state is a quantum state iff
\begin{equation}\label{eq:posPT}
    \alpha - \frac{1}{2}\sqrt{10\beta^2 + 2|\beta|\sqrt{9\beta^2+8}+1} \geqslant 0.
\end{equation}
This expression coincides with $\lambda^\star_{\mathrm{min}}$ since
\begin{displaymath}
\begin{split}
    &\sqrt{10\beta^2 + 2\beta\sqrt{9\beta^2+8}+1} = 
    \frac{32\beta^3 - 11\beta - \sqrt{9\beta^2+8}}{4\sqrt{2}(2\beta^2-1)} \\
    &\times \sqrt{\frac{\beta(16\beta^2-3)\sqrt{9\beta^2+8}+80\beta^4-23\beta^2+4}{64\beta^4-12\beta^2+1}},
\end{split}
\end{displaymath}
for $\beta \geqslant 0$. The entanglement measure reads as
\begin{displaymath}
    \mathcal{E}_{1|23}(\gamma_{xx}, \gamma_{pp}) = 
    \frac{1}{2}\sqrt{10\beta^2 + 2|\beta|\sqrt{9\beta^2+8}+1} - \alpha.
\end{displaymath}
We conclude that $1|23$-separability coincides with the positivity of partial
transposition, which is expressed by the inequality
$\mathcal{E}_{1|23}(\gamma_{xx}, \gamma_{pp}) \leqslant 0$ or \eqref{eq:posPT},
which is the same as the inequality \eqref{eq:pos1|23}. 
\end{proof}

\begin{lmn}
A state with the CM given by Eq.~\eqref{eq:CM3S} is $\{1|23, 2|13\}$-entangled
if the parameters $\alpha$ and $\beta$ satisfy the relation
\begin{equation}\label{eq:1|23,2|13-sep}
    \alpha < \frac{2}{3}|\beta| + \frac{\sqrt{28\beta^2+3}}{3}\cos\left[\frac{1}{3}
    \arccos\left(\frac{18|\beta|-80|\beta|^3}{\sqrt{(28\beta^2+3)^3}}\right)\right].
\end{equation}
\begin{proof}
The optimal solution for $\beta \geqslant 0$ reads as
\begin{equation}
    X^\star = 
    \begin{pmatrix}
        x^\star & -x^\star & 0 \\
        -x^\star & x^\star & 0 \\
        0 & 0 & 0
    \end{pmatrix}, \quad
    P^\star = 
    \begin{pmatrix}
        p^\star & p^\star & q^\star \\
        p^\star & p^\star & q^\star \\
        q^\star & q^\star & r^\star
    \end{pmatrix},
\end{equation}
and the dual solution is given by
\begin{equation}
    p^\star_{1|23} = p^\star_{2|13} = \frac{1}{2},
\end{equation}
and
\begin{equation}
\begin{split}
    \gamma^{1\star}_{xx} = \gamma^{2\star}_{xx} = 
    \begin{pmatrix}
        a^\star
    \end{pmatrix}, \quad
    \gamma^{23\star}_{xx} = \gamma^{13\star}_{xx} = 
    \begin{pmatrix}
        a^\star & b^\star \\
        b^\star & c^\star
    \end{pmatrix}, \\
    \gamma^{1\star}_{pp} = \gamma^{2\star}_{pp} = 
    \begin{pmatrix}
        d^\star
    \end{pmatrix}, \quad
    \gamma^{23\star}_{pp} = \gamma^{13\star}_{pp} = 
    \begin{pmatrix}
        e^\star & f^\star \\
        f^\star & g^\star
    \end{pmatrix},
\end{split}
\end{equation}
where the matrix elements are defined by even more complicated expressions than
in the previous case. The single-mode parts of the dual solution are squeezed
vacuum and the two-mode parts are pure. 

If we substitute these large expressions to
$\lambda^\star_{\mathrm{min}}$ we obtain
\begin{equation}
    \lambda^\star_{\mathrm{min}} = \alpha - \ldots,
\end{equation}
where $\ldots$ stands for another large expression of $\beta$ only. There is a
way to represent this expression more compactly. If we add the equation
$\lambda^\star_{\mathrm{min}} = 0$ to the system of KKT equations and eliminate
all other variables except $\alpha$ and $\beta$, we obtain the following
relation between $\alpha$ and $\beta$ when $\lambda^\star_{\mathrm{min}}$ is
zero:
\begin{equation}
    4\alpha^3 - 8\alpha^2 \beta - (4\beta^2 + 1)\alpha + 8\beta^3 = 0.
\end{equation}
It follows that $\ldots$ must be one of the roots of this polynomial of $\alpha$
for a fixed $\beta$. This root can be easily identified and for $\beta \geqslant
0$ we have
\begin{displaymath}
\begin{split}
    &\lambda^\star_{\mathrm{min}} = \alpha - \frac{2}{3}\beta \\
    &-\frac{1}{3}\re\sqrt[3]{18\beta-80\beta^3+3i\sqrt{48(36\beta^6+23\beta^4+\beta^2)+3}}.
\end{split}
\end{displaymath}
The latter term can be rewritten in a manifestly real form as it has been done
in a previous section, and we derive
\begin{displaymath}
\begin{split}
    &\mathcal{E}_{\{1|23, 2|13\}}(\gamma_{xx}, \gamma_{pp}) = -\alpha \\
    &+\frac{2}{3}|\beta| + \frac{\sqrt{28\beta^2+3}}{3}\cos\left[\frac{1}{3}
    \arccos\left(\frac{18|\beta|-80|\beta|^3}{\sqrt{(28\beta^2+3)^3}}\right)\right].
\end{split}
\end{displaymath}
The entanglement condition $\mathcal{E}_{\{1|23, 2|13\}}(\gamma_{xx},
\gamma_{pp}) > 0$ then becomes the inequality \eqref{eq:1|23,2|13-sep}.
\end{proof}
\end{lmn}

All the entanglement measures above were expressed by a single condition. As we
will see below, this does not have to be always the case.
\begin{lmn}
A state with the CM given by Eq.~\eqref{eq:CM3S} is $\{1|23, 2|13,
3|12\}$-entangled (genuine entangled) if the parameters $\alpha$ and $\beta$
satisfy the relation
\begin{equation}\label{eq:1|23,2|13,3|12-sep}
    \alpha < 
    \begin{cases}
        \frac{1}{3}\sqrt{36\beta^2 + 1} + \frac{1}{6} & |\beta| \leqslant \frac{1}{6\sqrt{3}} \\
        R_1(\beta) & |\beta| \geqslant \frac{1}{6\sqrt{3}}
    \end{cases},
\end{equation}
where $R_1(\beta)$ is the largest root of the following equation of the variable
$\alpha$:
\begin{equation}\label{eq:alphabetaeq}
\begin{split}
    &186624 \alpha^7 \beta^2 - 1119744 \alpha^5 \beta^4 - 57024 \alpha^5 \beta ^2 + 144 \alpha^5 \\
    &+186624 \alpha^3 \beta^6 + 98496 \alpha^3 \beta^4 + 1584 \alpha^3 \beta^2 - 40 \alpha^3 \\
    &+5225472 \alpha  \beta^8 + 435456 \alpha \beta^6 + 1728 \alpha \beta^4 + 76 \alpha \beta^2 \\
    &+\alpha -373248 \alpha^6 |\beta|^3 - 5184 \alpha^6 |\beta| + 2985984 \alpha^4 |\beta|^5 \\
   &+139968 \alpha^4 |\beta|^3 + 1584 \alpha^4 |\beta| - 5598720 \alpha^2 |\beta|^7 \\
   &-487296 \alpha^2 |\beta|^5 - 6768 \alpha^2 |\beta|^3 - 76 \alpha^2 |\beta| + 40 |\beta|^3 \\
   &-1492992 |\beta|^9 - 124416 |\beta|^7 + 1728 |\beta|^5 + |\beta| = 0.
\end{split}
\end{equation}
\end{lmn}
\begin{proof}
The optimal solution for $\beta \geqslant 0$ is of the form
    \begin{equation}\label{eq:bisepprimal}
        X^\star = 
        \begin{pmatrix}
            x^\star & y^\star & y^\star \\
            y^\star & x^\star & y^\star \\
            y^\star & y^\star & x^\star
        \end{pmatrix}, \quad
        P^\star =
        \begin{pmatrix}
            p^\star & q^\star & q^\star \\
            q^\star & p^\star & q^\star \\
            p^\star & q^\star & q^\star
        \end{pmatrix}, 
    \end{equation}
and the dual solution reads as
\begin{equation}
    p^\star_{1|23} = p^\star_{2|13} = p^\star_{3|12} = \frac{1}{3},
\end{equation}
and
\begin{equation}\label{eq:bisepdual}
\begin{split}
    &\gamma^{1\star}_{xx} = \gamma^{2\star}_{xx} = \gamma^{3\star}_{xx} = 
    \begin{pmatrix}
        a^\star
    \end{pmatrix}, \\
    &\gamma^{23\star}_{xx} = \gamma^{13\star}_{xx} = \gamma^{12\star}_{xx} = 
    \begin{pmatrix}
        b^\star & c^\star \\
        c^\star & b^\star
    \end{pmatrix}, \\
    &\gamma^{1\star}_{pp} = \gamma^{2\star}_{pp} = \gamma^{3\star}_{pp} = 
    \begin{pmatrix}
        d^\star
    \end{pmatrix}, \\
    &\gamma^{23\star}_{pp} = \gamma^{13\star}_{pp} = \gamma^{12\star}_{pp} = 
    \begin{pmatrix}
        e^\star & f^\star \\
        f^\star & e^\star
    \end{pmatrix}.
\end{split}
\end{equation}
This time the values of the optimal parameters cannot be found analytically, but
we can find an equation for $\alpha$ and $\beta$ when
$\lambda^\star_{\mathrm{min}}$ becomes zero and thus we can determine the region
of biseparability. This can be done as in the previous case by adding the
equation $\lambda^\star_{\mathrm{min}} = 0$ to the system of KKT conditions end
eliminating the elements of the primal and dual solutions, but this way we do
not get the expression for $\lambda^\star_{\mathrm{min}}$. So, we write
\begin{equation}
    \lambda^\star_{\mathrm{min}} = \alpha - R,
\end{equation}
and try to determine $R$. Adding this equation to the KKT conditions and
eliminating all variables except $\alpha$, $\beta$ and $R$, we obtain that $R$
is a zero of a product of several terms:
\begin{equation}
\begin{split}
    &(4R^2-4\beta^2-1)(4R^2-16\beta^2-1)(9R^2-9\beta^2-1) \\
    &(36R^2-36\beta^2-1)(36R^2-144\beta^2-1),
\end{split}
\end{equation}
\begin{equation}\label{eq:fac1}
    12R^2-48\beta^2\pm 4R-1,
\end{equation}
and
\begin{equation}\label{eq:fac2}
    \begin{split}
        &186624 R^7 \beta^2 - 1119744 R^5 \beta^4 - 57024 R^5 \beta ^2 + 144 R^5 \\
        &+186624 R^3 \beta^6 + 98496 R^3 \beta^4 + 1584 R^3 \beta^2 - 40 R^3 \\
        &+5225472 R  \beta^8 + 435456 R \beta^6 + 1728 R \beta^4 + 76 R \beta^2 \\
        &+R \mp 373248 R^6 \beta^3 \mp 5184 R^6 \beta \pm 2985984 R^4 \beta^5 \\
       &\pm139968 R^4 \beta^3 \pm 1584 R^4 \beta \mp 5598720 R^2 \beta^7 \\
       &\mp487296 R^2 \beta^5 \mp 6768 R^2 \beta^3 \mp 76 R^2 \beta \pm 40 \beta^3 \\
       &\mp1492992 \beta^9 \mp 124416 \beta^7 \pm 1728 \beta^5 \pm \beta,
    \end{split}
\end{equation}
where either the upper or lower sign is taken in all ambiguous terms. We see
that $R = R(\beta)$ depends only on $\beta$, but we need to determine of which
of these terms $R$ is a root. It might even be that this term is different for
different values of $\beta$. 

The difference $g = g(x, y, p, q)$ defined via
\begin{equation}
\begin{split}
    g = \tr(X\gamma_{xx} + P\gamma_{pp}) &- \str(X[1], P[1]) \\
    & - \str(X[2, 3], P[2, 3])
\end{split}
\end{equation}
has already a familiar form
\begin{equation}
\begin{split}
    g &= 3(x+p)\alpha - 6(q+y)\beta -\sqrt{xp} \\
    &-\sqrt{(x+y)(p-q)} - \sqrt{(x-y)(p+q)},
\end{split}
\end{equation}
so our conjecture is that the optimal $x^\star$, $y^\star$, $p^\star$ and
$q^\star$ are also optimal for this function. The optimization of $g$ has less
variables, just 4, and thus is simpler than the original problem. Setting the
derivatives of the Lagrangian
\begin{equation}
    \mathcal{L} = g(x, y, p, q) + \nu[3(x+p)-1],
\end{equation}
to zero together with the trace normalization condition
\begin{equation}
    \tr(X + P) = 3(x+p)-1 = 0,
\end{equation}
produces the following solution:
\begin{equation}
    x^\star = p^\star = \frac{1}{6}, \quad y^\star = -q^\star = -\frac{\beta}{\sqrt{36\beta^2+1}}.
\end{equation}
We then substitute these expressions into the original problem to find the dual solution:
\begin{equation}
    a^\star = d^\star = \frac{1}{2}, \quad 
    b^\star = e^\star = \frac{1}{2}\sqrt{36\beta^2+1}, \quad
    c^\star = f^\star = 3\beta.
\end{equation}
It follows that $\gamma^{1\star}$ is ths CM of the vacuum state and
$\gamma^{23\star}$ is the CM of the two-mode squeezed vacuum state. The two-mode
squeezed vacuum is defined as
\begin{equation}
    |\zeta\rangle = \hat{S}(\zeta)|00\rangle = 
    \frac{1}{\cosh(r)} \sum^{+\infty}_{n=0} \tanh^n(r) e^{-i \varphi n} |nn\rangle,
\end{equation}
where the squeezing operator $\hat{S}(\zeta)$ reads as
\begin{equation}
    \hat{S}(\zeta) = \exp(-\zeta \hat{a}^\dagger \hat{b}^\dagger + \zeta^* \hat{a}\hat{b}).
\end{equation}
In general, the squeezing parameter can be an arbitrary complex number, $\zeta =
r e^{i\varphi}$, but here $\zeta = r$ is a real number determined from $\beta$
according to the relation
\begin{equation}\label{eq:br}
    r = \frac{1}{2}\arcsin(6\beta).
\end{equation}
The eigenvalues of $X^\star$ are 
\begin{equation}
    \sqrt{36\beta^2+1} + 6\beta, \quad \sqrt{36\beta^2+1} - 12\beta,
\end{equation}
the first with multiplicity 2. The eigenvalues of $P^\star$ are
\begin{equation}
    \sqrt{36\beta^2+1} - 6\beta, \quad \sqrt{36\beta^2+1} + 12\beta,
\end{equation}
the first with multiplicity 2 (we consider the case of $\beta \geqslant 0$, for
$\beta \leqslant 0$ these matrices should be swapped). We see that $P^\star$ is
always positive semidefinite, and $X^\star$ is positive semidefinite iff
\begin{equation}
    \beta \leqslant \beta^* = \frac{1}{6\sqrt{3}}.
\end{equation}
This value corresponds to the squeezing parameter given by the following simple
expression:
\begin{equation}
    r^* = \frac{\ln(3)}{4}.
\end{equation}
For these parameters we have
\begin{equation}
    g^\star = g(x^\star, y^\star, p^\star, q^\star) = \alpha - \frac{1}{3}\sqrt{36\beta^2+1}-\frac{1}{6}.
\end{equation}
On the other hand, we also have
\begin{equation}
    \begin{split}
        &\gamma_{xx} - \frac{1}{3}\left(\gamma^{1\star}_{xx} \oplus \gamma^{23\star}_{xx} 
        + \gamma^{2\star}_{xx} \oplus \gamma^{13\star}_{xx} 
        + \gamma^{3\star}_{xx} \oplus \gamma^{12\star}_{xx}\right) \\
        &=\gamma_{pp} - \frac{1}{3}\left(\gamma^{1\star}_{pp} \oplus \gamma^{23\star}_{pp} 
        + \gamma^{2\star}_{xx} \oplus \gamma^{13\star}_{pp} 
        + \gamma^{3\star}_{xx} \oplus \gamma^{12\star}_{pp}\right) \\
        &= \lambda^\star_{\mathrm{min}}E,
    \end{split}
\end{equation}
where
\begin{equation}
    \lambda^\star_{\mathrm{min}} = \alpha - R \equiv \alpha - \frac{1}{3}\sqrt{36\beta^2+1}-\frac{1}{6}.
\end{equation}
This result shows that these parameters are valid primal and dual solutions, but
this construction works only provided that $|\beta|\leqslant \beta^*$, for
larger $\beta$ either $X^\star$ or $P^\star$ is not positive-semidefinite. This
proves the first part of the lemma.

Note that for $|\beta|\leqslant \beta^*$ the equality
$\lambda^\star_{\mathrm{min}} = 0$ holds iff 
\begin{equation}
    12\alpha^2-48\beta^2- 4\alpha-1 = 0,
\end{equation}
which means that $\alpha$ is a root of the factor of Eq.~\eqref{eq:fac1} when
the upper sign is taken. For $|\beta| \leqslant \beta^*$ the condition
$\lambda^\star_{\mathrm{min}} = 0$ must produce a root of some other term. By
substitution $\beta \to \beta^*$ to all the terms we find that the only term
that is also zero is given by Eq.~\eqref{eq:fac2} with all upper signs, so $R$
in this case is a root of that term. By plotting the roots of the term, we
conclude that $R$ must be the largest root, which proves the second part of the
lemma. The explicit expressions for the primal and dual solution does not seem
to exist. We thus obtain the following expression for the entanglement measure:
\begin{equation}\label{eq:GEM}
    \mathcal{E}_{\mathfrak{I}_2}(\gamma_{xx}, \gamma_{pp}) = R(\beta) - \alpha,
\end{equation}
where $R(\beta)$ is defined by different expressions depending on $\beta$. 
\end{proof}

\begin{figure*}
    \includegraphics{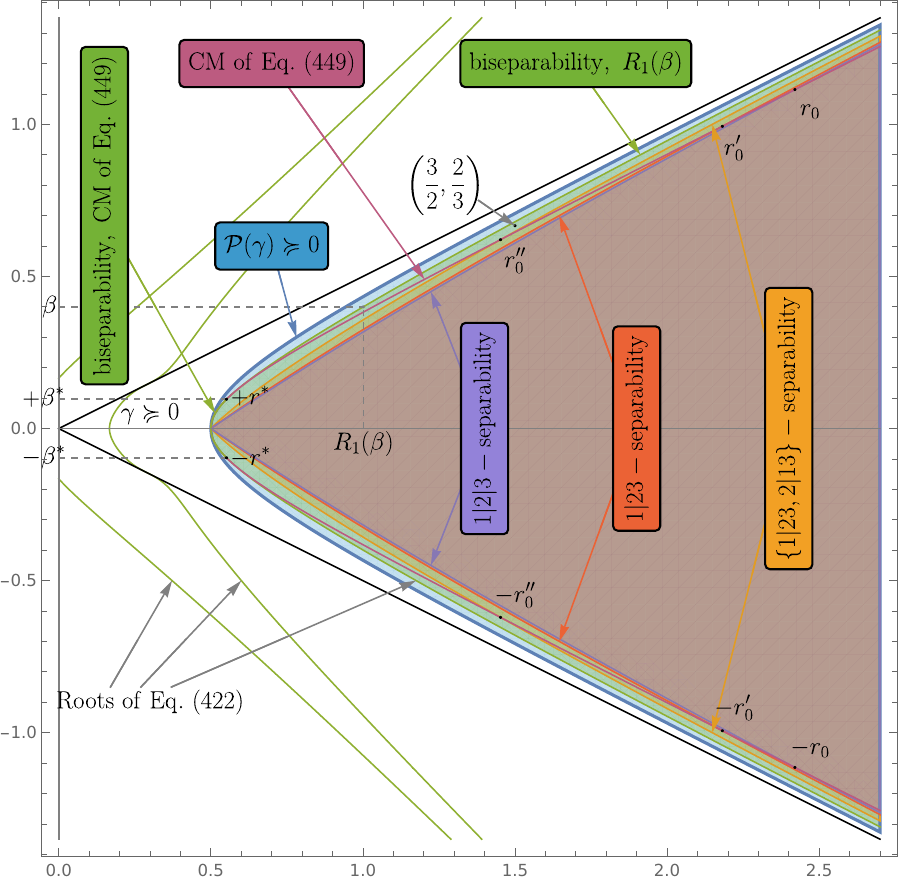}
    \caption{The regions of different kinds of separability of the state with
    the CM given by Eq.~\eqref{eq:CM3S}.}\label{fig:3sepregions}
\end{figure*}

All the obtained regions are shown in Fig.~\ref{fig:3sepregions}. The very thin
region between the red and purple lines correspond to CMs that are
simultaneously $1|23$-, $2|13$- and $3|12$-separable, but not $1|2|3$-separable.
The region between the yellow and red lines are the $\{1|23, 2|13\}$-separable
CMs that are not $1|23$-separable. The region between the green and yellow lines
are biseparable, but neither $\{1|23, 2|13\}$-, nor $\{1|23, 3|12\}$-, nor
$\{2|13, 3|12\}$-separable CMs. 

We now apply the obtained results to some simple, but interesting states. As a
first example, which is a non-Gaussian state, we take an equally weighted
mixture of products of single-mode vacuum with two-mode squeezed vacuum states:
\begin{equation}\label{eq:rho3sq}
\begin{split}
    \hat{\varrho} &= \frac{1}{3}(|0\rangle_1\langle 0|\otimes|\zeta\rangle_{23}\langle\zeta| + 
    |0\rangle_2\langle 0|\otimes|\zeta\rangle_{13}\langle\zeta| \\
    &+|0\rangle_3\langle 0|\otimes|\zeta\rangle_{12}\langle\zeta|).
\end{split}
\end{equation}
This state has the same CM as the optimal dual solution from the first part of
the last lemma above. Because the first-order moments of the state
\eqref{eq:rho3sq} are zero, the CM of the mixture is the mixture of CMs and we
have
\begin{equation}\label{eq:CM3sq}
\begin{split}
    \gamma_{xx} &= \frac{1}{3}
    \begin{pmatrix}
        \cosh(2r)+\frac{1}{2} & \frac{1}{2}\sinh(2r) & \frac{1}{2}\sinh(2r) \\[1mm]
        \frac{1}{2}\sinh(2r) & \cosh(2r)+\frac{1}{2} & \frac{1}{2}\sinh(2r) \\[1mm]
        \frac{1}{2}\sinh(2r) & \frac{1}{2}\sinh(2r) & \cosh(2r)+\frac{1}{2}
    \end{pmatrix}, \\
    \gamma_{pp} &= \frac{1}{3}
    \begin{pmatrix}
        \cosh(2r)+\frac{1}{2} & -\frac{1}{2}\sinh(2r) & -\frac{1}{2}\sinh(2r) \\[1mm]
        -\frac{1}{2}\sinh(2r) & \cosh(2r)+\frac{1}{2} & -\frac{1}{2}\sinh(2r) \\[1mm]
        -\frac{1}{2}\sinh(2r) & -\frac{1}{2}\sinh(2r) & \cosh(2r)+\frac{1}{2}
    \end{pmatrix}.
\end{split}
\end{equation}
This CM is of the form given by Eq.~\eqref{eq:CM3S} with
\begin{equation}
    \alpha = \frac{1}{3}\left(\cosh(2r)+\frac{1}{2}\right), \quad
    \beta = \frac{1}{6}\sinh(2r).
\end{equation}
This parametric curve is also shown in Fig.~\ref{fig:3sepregions}.

\begin{figure}[ht]
    \includegraphics{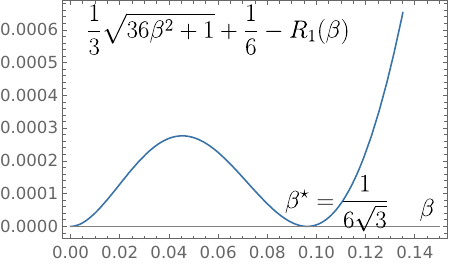}
    \caption{The difference between the curve \eqref{eq:c3s} and the curve $\alpha = R_1(\beta)$.}
    \label{fig:diff}
\end{figure}

From the proof of the first part of the last lemma above it follows that the
parametric curve described by these covariance matrices forms a part of the
boundary of the biseparable states for $|r| \leqslant r^*$. Note that this
parametric curve has the following explicit equation:
\begin{equation}\label{eq:c3s}
    \alpha = \frac{1}{3}\sqrt{36\beta^2+1}+\frac{1}{6}.
\end{equation}
This curve and the curve $\alpha = R_1(\beta)$ are very close near the origin
$\beta=0$, Fig.~\ref{fig:diff}. One would need to zoom
Fig.~\ref{fig:3sepregions} at least 1000x to see the difference, so for most, if
not all, practical purposes this difference is irrelevant, but nevertheless
these are different curves. For larger $\beta$ the curve Eq.~\eqref{eq:c3s}
diverges form the biseparability boundary and intersects the boundaries of all
the regions. The intersection with the $1|2|3$-separability region happens at
$\zeta = r_0$, where
\begin{equation}
    r_0 = \frac{1}{2}\ln\left(\frac{13+\sqrt{193}}{2}\right) = 1.2993....
\end{equation}
This result guarantees that the state \eqref{eq:rho3sq} with $\zeta = r < r_0$
is not $1|2|3$-separable, but for a larger value of squeezing we cannot say that
it is separable. The Gaussian state with the CM given by Eq.~\eqref{eq:CM3sq} is
$1|2|3$-separable if $r>r_0$, but for the non-Gaussian state \eqref{eq:rho3sq}
with the same CM this test is inconclusive. 

The intersection with $1|23$-separability regions happens at $\zeta = r'_0$,
where
\begin{displaymath}
    r'_0 = \frac{1}{2} \ln\left(\frac{7+2\sqrt{31}+2\sqrt{41+7\sqrt{31}}}{3}\right) = 1.2427....
\end{displaymath}
For $\zeta=r< r'_0$ the state \eqref{eq:rho3sq} is guaranteed to be
$1|23$-entangled (and due to symmetry also $2|13$- and $3|12$-entangled), but
for $r\geqslant r'_0$ we cannot say anything. This number was also obtained in
Refs.~\cite{Baksova2025, arXiv-2401.04376}.

The CMs \eqref{eq:CM3sq} intersect the $\{1|23, 2|13\}$-separability regions at
the squeezing parameter $r^{\prime\prime}_0$, where 
\begin{displaymath}
    r^{\prime\prime}_0 = \frac{1}{2} \ln(\tilde{r}) = 1.0128..., \quad
    3\tilde{r}^4 - 19 \tilde{r}^3 - 25 \tilde{r}^2 -25 \tilde{r} - 6 =0.
\end{displaymath}
The number $\tilde{r}$ is the 4-th root of its minimal polynomial, and one can
give an explicit expression for it in terms of radicals. As before, for the
squeezing parameter $\zeta = r < r^{\prime\prime}_0$ the state \eqref{eq:rho3sq}
is definitely not $\{1|23, 2|13\}$-separable, but for $r \geqslant
r^{\prime\prime}_0$ we cannot say anything. 

The plot of different entanglement measures for CMs of the form \eqref{eq:CM3sq}
is shown in Fig.~\ref{fig:GEM}. The genuine entanglement measure is always
non-positive, which means that the state is not genuine entangled. Since our
state is biseparable by construction, this is expectable. For $\zeta = r = 0$
the CM becomes the CM of the product of vacuum states, and this measure must be
zero. This measure is exactly zero up to where the squeezing parameter becomes
$r^*$, for larger values the measure becomes strictly negative, which means that
the corresponding CMs are no longer on the boundary of the biseparability
region. The points $(\alpha^*, \pm\beta^*)$ corresponding to $\pm r^*$ read as
\begin{equation}\label{eq:alphabetastar}
\begin{split}
    (\alpha^*, \pm\beta^*) &= \left(\frac{1}{3}\left[\cosh(2r^*)+\frac{1}{2}\right], 
    \pm\frac{1}{6}\sinh(2r^*)\right) \\
    &= \left(\frac{4+\sqrt{3}}{6\sqrt{3}}, \pm \frac{1}{6\sqrt{3}}\right).
\end{split}
\end{equation}
These points are simple non-trivial solutions of Eq.~\eqref{eq:alphabetaeq}.
Also note the inequality
\begin{equation}
	\mathcal{E}_{\mathfrak{I}_2} \leqslant \mathcal{E}_{\{1|23, 2|13\}}
	\leqslant \mathcal{E}_{1|23} \leqslant \mathcal{E}_{1|2|3},
\end{equation}
which fully agrees with the inequalities given by Eq.~\eqref{eq:EE}.

\begin{figure}[ht]
    \includegraphics{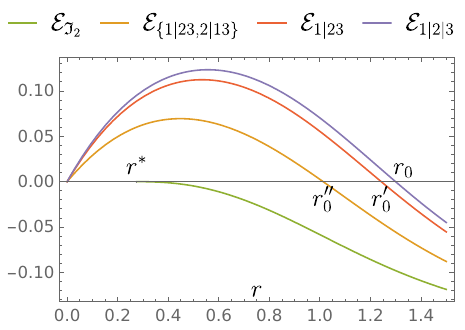}
    \caption{Genuine entanglement measure of CM \eqref{eq:CM3sq}.}\label{fig:GEM}
\end{figure}

We said before that an analytical construction of the optimal solution for
$\{1|23, 2|13\}$-separability is quite large and have not presented it
explicitly. Now we show that the CM \eqref{eq:CM3sq} with the parameters given
by Eq.~\eqref{eq:alphabetastar} is $\{1|23, 2|13\}$-entangled. If we do not aim
at the optimal violation of the separability condition, any pair of positive
semidefinite matrices $X$ and $P$, not necessarily trace-normalized, that
violate the condition will suffice. We compute the exact analytical solution,
which is expressed by algebraic numbers of 3rd order, rationalize it and get a
much simpler witness,
\begin{equation}
    X = 
    \begin{pmatrix}
        \frac{87}{368} & -\frac{87}{368} & 0 \\[1mm]
        -\frac{87}{368} & \frac{87}{368} & 0 \\[1mm]
        0 & 0 & 0
    \end{pmatrix}, \quad
    P = 
    \begin{pmatrix}
        \frac{54}{215} & \frac{54}{215} & \frac{23}{291} \\[1mm]
        \frac{54}{215} & \frac{54}{215} & \frac{23}{291} \\[1mm]
        \frac{23}{291} & \frac{23}{291} & \frac{5}{201}
    \end{pmatrix}.
\end{equation}
These matrices are positive semidefinite and for the primal objective value of
the $\{1|23, 2|31\}$-entanglement optimization problem we have
\begin{displaymath}
\begin{split}
    &\tr(X\gamma_{xx} + P\gamma_{pp}) - \str(X[1], P[1]) - \str(X[2, 3], P[2, 3]) \\
    &= \frac{771322369 + 696435289\sqrt{3}}{4627807920} - 9\sqrt{\frac{29}{9890}}
    = -0.0600...\ .
\end{split}
\end{displaymath}
For the $\{2|13\}$-partition we get the same negative value, which proves that
this state is $\{1|23, 2|31\}$-entangled. Not surprisingly, for the partition
$3|12$ we have
\begin{displaymath}
\begin{split}
    &\tr(X\gamma_{xx} + P\gamma_{pp}) - \str(X[3], P[3]) - \str(X[1, 2], P[1, 2]) \\
    &= \frac{771322369 + 696435289\sqrt{3}}{4627807920} 
    = 0.4273... > 0,
\end{split}
\end{displaymath}
so the matrices $X$ and $P$ certify only the $\{1|23, 2|13\}$-entanglement and
not the genuine entanglement (which would contradict to the fact the our state
is biseparable by construction). This gives an explicit example of a 3-mode
biseparable state which is not separable with respect to any other form of
separability (except biseparability). 

For most parameters $\alpha$ and $\beta$ the exact solution cannot be expressed
in terms of radicals and one has to work with algebraic numbers in terms of
their minimal polynomials. Consider a very simple case with easily memorable
parameters $\alpha = 3/2$ and $\beta=2/3$. The CM reads as
\begin{equation}\label{eq:bisepexample}
    \gamma_{xx} = \frac{1}{6}
    \begin{pmatrix}
        9 & 4 & 4 \\
        4 & 9 & 4 \\
        4 & 4 & 9
    \end{pmatrix}, \quad
    \gamma_{pp} = \frac{1}{6}
    \begin{pmatrix}
        9 & -4 & -4 \\
        -4 & 9 & -4 \\
        -4 & -4 & 9
    \end{pmatrix}.
\end{equation}
The corresponding point is shown in Fig.~\ref{fig:3sepregions} and lies between
the green (biseparability) and yellow ($\{1|23, 2|13\}$-separability) lines. We
show that this state is biseparable, but cannot be represented as a combination
of one or two separable states with respect to a fixed partition. 

The analytical solution of the KKT equations cannot be obtained even if we set
the above concrete values of $\alpha$ and $\beta$, but the numerical solution
produces $f^\star = -2.0$ in Eq.~\eqref{eq:bisepdual}. Adding this equation to
the KKT system of equation allows one to obtain an analytical expression for
both primal, Eq,~\eqref{eq:bisepprimal}, and dual, Eq.~\eqref{eq:bisepdual},
solutions. Here we need just the dual solution to demonstrate that the state
with the CM given by Eq.~\eqref{eq:bisepexample} is biseparable. The other
matrix elements of the dual solution are algebraic numbers with the minimal
polynomials given by
\begin{widetext}
\begin{equation}
\begin{split}
    7808 a^7 -54064 a^6 -100448 a^5 +159960 a^4 +25112 a^3 -3475 a^2 - 122a - 12 &= 0 \\
    562176 b^7 +2838528 b^6 -118416 b^5 -4411968 b^4 -587704 b^3 +6640 b^2 -3913 b -24 &= 0 \\
    281088 c^7 -1883520 c^6 +2786016 c^5 -509001 c^4 -638456 c^3 -93400 c^2 -4576 c - 144 & = 0 \\
    3072 d^7 +7808 d^6 +55600 d^5 -100448 d^4 -159960 d^3 +25112 d^2 +3379 d -122 &= 0 \\
    48 e^7 -160 e^6 -24 e^5 +1472 e^4 -2109 e^3 -3270 e^2 +5644 e -512 &= 0.
\end{split}
\end{equation}
\end{widetext}
In \textsl{Mathematica} nomenclature, the concrete roots have the following numbers:
\begin{equation}
    a^\star : 4, \quad b^\star : 5, \quad c^\star : 3, \quad d^\star : 4, \quad e^\star : 5.
\end{equation}
Numerically, the dual solution reads as
\begin{equation}
    \begin{split}
        &\gamma^{1\star}_{xx} = \gamma^{2\star}_{xx} = \gamma^{3\star}_{xx} = 
        \begin{pmatrix}
            1.1689...
        \end{pmatrix}, \\
        &\gamma^{23\star}_{xx} = \gamma^{13\star}_{xx} = \gamma^{12\star}_{xx} = 
        \begin{pmatrix}
            1.1981... & 1.1373... \\
            1.1373... & 1.1981...
        \end{pmatrix}, \\
        &\gamma^{1\star}_{pp} = \gamma^{2\star}_{pp} = \gamma^{3\star}_{pp} = 
        \begin{pmatrix}
            0.2138...
        \end{pmatrix}, \\
        &\gamma^{23\star}_{pp} = \gamma^{13\star}_{pp} = \gamma^{12\star}_{pp} = 
        \begin{pmatrix}
            2.1070... & -2 \\
            -2 & 2.1070...
        \end{pmatrix}.
    \end{split}
\end{equation}
It is possible to check that these are physical CMs and $a^\star d^\star = 1/4$,
so the single-mode states are squeezed vacuums states. In addition, the two-mode
states satisfy the equality 
\begin{equation}
    \gamma^{23\star}_{xx} \gamma^{23\star}_{pp} = \frac{1}{4}\Omega_2,
\end{equation}
and are thus pure. We have that
\begin{equation}
\begin{split}
    \gamma_{xx} &\succcurlyeq \frac{1}{3}\left(\gamma^{1\star}_{xx} \oplus \gamma^{23\star}_{xx} 
    + \gamma^{2\star}_{xx} \oplus \gamma^{13\star}_{xx} 
    + \gamma^{3\star}_{xx} \oplus \gamma^{12\star}_{xx}\right), \\
    \gamma_{pp} &\succcurlyeq \frac{1}{3}\left(\gamma^{1\star}_{pp} \oplus \gamma^{23\star}_{pp} 
    + \gamma^{2\star}_{xx} \oplus \gamma^{13\star}_{pp} 
    + \gamma^{3\star}_{xx} \oplus \gamma^{12\star}_{pp}\right),
\end{split}
\end{equation}
since the eigenvalues of the differences of the left-hand side and the
right-hand side are positive and given by
\begin{equation}
\begin{split}
    &(0.8867..., 0.0240..., 0.0240...), \\
    &(0.0240..., 0.0240..., 0.0240...),
\end{split}
\end{equation}
respectively. Note that $\lambda^\star_{\mathrm{min}} = 0.0240...$. We emphasize
that these statements can be verified exactly by computing the minimal
polynomial of the quantities in question.

\section{Open problems}\label{sec:IX}

In all the problems studied analytically above in the KKT condition
\eqref{eq:MMKKT}, which we reproduce here as
\begin{equation}\label{eq:XPrank}
	\begin{pmatrix}
		\gamma^{(I)\star}_{xx} & \frac{1}{2}E \\
		\frac{1}{2}E & \gamma^{(I)\star}_{pp}
	\end{pmatrix}
	\begin{pmatrix}
		X^\star[I] & -Z^{(I)\star} \\
		-Z^{(I)\star\mathrm{T}} & P^\star[I]
	\end{pmatrix}
	= 0,
\end{equation}
both terms have the same rank, equal to $k = |I|$, the size of the matrices
$\gamma^{(I)}_{xx}$, $\gamma^{(I)}_{pp}$, $X[I]$ and $P[I]$. Since rank of the
first term is always $\geqslant k$, the statement that both terms have rank $k$
has a simple interpretation. The optimal $\gamma^{(I)\star}_{xx}$ and
$\gamma^{(I)\star}_{pp}$ are such that the last $k$ column-vectors are linearly
dependent on the first $k$,
\begin{equation}
	\begin{pmatrix}
		* & \frac{1}{2}E \\
		\frac{1}{2}E & *
	\end{pmatrix}.
\end{equation}
Under the assumption that $\rank(X^{(I)\star}) = \rank(P^{(I)\star}) = k$ the
optimal $Z^{(I)\star}$ is such that it makes the last $k$ column-vectors
linearly dependent on the first $k$,
\begin{equation}
	\begin{pmatrix}
		X^{(I)\star} & * \\
		* & P^{(I)\star}
	\end{pmatrix}.
\end{equation}
This assumption turns out to be valid in all the analytical cases considered
above and for all $I$ participating in the optimal solution. The full optimal
matrices $X^\star$ and $P^\star$ are not necessarily of full rank. In fact, they
might even be of different rank. Numerical experiments show that the two terms
in the product \eqref{eq:XPrank} do not always have the same rank, but the
following less strong relation holds:
\begin{equation}
	\rank
	\begin{pmatrix}
		\gamma^{(I)\star}_{xx} & \frac{1}{2}E \\
		\frac{1}{2}E & \gamma^{(I)\star}_{pp}
	\end{pmatrix}
	+
	\rank
	\begin{pmatrix}
		X^\star[I] & -Z^{(I)\star} \\
		-Z^{(I)\star\mathrm{T}} & P^\star[I]
	\end{pmatrix}
	= 2k.
\end{equation}
If $A$ and $B$ are matrices of the same size $2k$ and $AB=0$, then a weaker
inequality is always satisfied,
\begin{equation}
	\rank(A) + \rank(B) \leqslant 2k.
\end{equation}
Our observation is the stronger statement that this inequality is always
equality, and in some cases an even stronger relation,
\begin{equation}
	\rank(A) = \rank(B) = k,
\end{equation}
is valid. In the full case, the similar numerical observation states that
\begin{equation}
	\rank \mathcal{P}(\gamma^{(I)\star}) + \rank(M^\star[I] + i K^{(I)\star}) = 2k.
\end{equation}
So, the first open problem is to prove this conjecture.

Another problem is to find the maximal values of the measures
$\mathcal{E}^+_{\mathfrak{I}_2}$
and $\mathcal{T}^+_{\mathfrak{I}_2}$. For the latter we have
$\mathcal{T}^+_{\mathfrak{I}_2}(\gamma) < 1$ for all
$\gamma$ and the symmetric states defined by Eq.~\eqref{eq:gammalm} numerically
show that
\begin{equation}
	\sup_{\lambda, \mu}\mathcal{T}^+_{\mathfrak{I}_2}(\gamma^S(\lambda, \mu)) = 1.
\end{equation}
An analytical proof of this equality would demonstrate that 
\begin{equation}\label{eq:maxT}
	\sup_\gamma \mathcal{T}^+_{\mathfrak{I}_2}(\gamma) = 1.
\end{equation}
For the same $n$-mode symmetric states we have, numerically
\begin{equation}
	\sup_{\lambda, \mu} \mathcal{E}^+_{\mathfrak{I}_2}(\gamma^S(\lambda, \mu)) = \frac{1}{n}.
\end{equation}
Our conjecture is then that the right-hand of this relation is also the maximal
value of the $\mathcal{E}^+$ measure,
\begin{equation}\label{eq:maxE}
	\sup_\gamma \mathcal{E}^+_{\mathfrak{I}_2}(\gamma) = \frac{1}{n}.
\end{equation}
Verifying the validity of the equations \eqref{eq:maxT} and \eqref{eq:maxE} is
our second open problem. A difficulty in maximizing these quantities is that
they are convex. Standard convex optimization refers to minimizing convex
functions or maximizing concave ones, but here wee need to maximize convex
functions.

The last open problem we formulate in this work is finding the genuine entangled
states with as many positive partial transpositions (PTs) as possible and with
$\mathcal{E}^+_{\mathfrak{I}_2}$ (or $\mathcal{T}^+_{\mathfrak{I}_2}$) as large
as possible. The proof given in Ref.~\cite{PhysRevLett.86.3658} states that
$1+n$ bipartitions of an $(n+1)$-mode state are equivalent to
$1+n$-separability, so we should exclude such bipartitions from the
consideration. The state with CM given by Eq.~\eqref{eq:BBS} has all three $2+2$
PTs positive (and all its $1+3$ PTs non-positive), but this state is
$12|34$-entangled. This is a very simple form of entanglement, but we are
interested in a much stronger notion of genuine entanglement. One way to look
for such states is to randomly generate a positive-semidefinite $8\times8$
matrix $M$ with trace 1 and then solve the following optimization problem:
\begin{equation}\label{eq:additiveM}
	\min_{\gamma} \left[\tr(M\gamma) - \min_{\mathcal{I} \in \mathfrak{I}_2}
		\sum_{I \in \mathcal{I}} \str(M[I])\right],
\end{equation}
where the minimization is over all physical CMs $\gamma$ with all three $2+2$
PTs positive and $\mathfrak{I}_2$ is the set of all seven bipartitions of 4
modes. So, given $M$, we optimize over $\gamma$. The optimal solution
$\gamma^\star$ might or might not be genuine entangled. Testing a few billions
of random matrices we found a matrix $M$ for which the optimal solution of the
problem \eqref{eq:additiveM} satisfies the inequality
\begin{equation}
	\mathcal{E}^+_{\mathfrak{I}_2}(\gamma^\star) > 0.01.	
\end{equation}
The measure $\mathcal{E}^+_{\mathfrak{I}_2}$ has been obtained by solving the
problem \eqref{eq:ent-test-full} where, given $\gamma$, we optimize over $M$.
The optimal matrix $M^\star$ corresponding to $\gamma^\star$ (and which is
typically different from the initial $M$) can be used as a new input to the
problem \eqref{eq:additiveM} and produce a new state $\gamma^{\star\star}$ with
slightly larger $\mathcal{E}^+_{\mathfrak{I}_2}$. This new state can be used as
an input to the problem \eqref{eq:ent-test-full} and produce a new matrix
$M^{\star\star}$, which can be used as an input to \eqref{eq:additiveM} and so
on. This process quickly leads to unwieldy numerical results, but
$\mathcal{E}^+_{\mathfrak{I}_2}(\gamma^{\star\star\ldots})$ seems to converge.
Our last open problem is to find 
\begin{equation}
	\sup_{\gamma} \mathcal{E}^{+}_{\mathfrak{I}_2}(\gamma),
\end{equation}
where the supremum is taken over the physical CMs with all PTs, except the ones
of the form $1+n$, positive. In the case of 4-mode states we just found that
this supremum is larger than $0.01$.

\section{Conclusion}

In this work we developed several approaches to determine the true physical
covariance matrix (CM) from uncertain experimentally measured non-physical data
and test the corresponding quantum state for entanglement. All these approaches
are based on solving convex optimization problems. Compared with existing
theoretical treatments, our current work is applicable to continuous-variable
multi-mode states of a larger size, consistent with the growing scale of optical
quantum computing experiments based on Gaussian cluster states
\cite{science2645,science4354,AghaeeRad2025,Jia2025}, though with the aim to go
beyond the scope of these experiments and their limited focus on full
inseparability. For other relevant recent works on Gaussian states we refer to
Refs.~\cite{PhysRevResearch.6.043113, roy2024}. We presented two ways to
determine the most likely physical CM given a measured non-physical one. The
first way is to compute the closest physical CM, and the other is to compute the
most probable CM under assumption of the Gaussian distribution of the measured
CM elements. These two procedures can be routinely applied to 200-mode states on
a high-end desktop PC.

Having recovered the physical CM, the next step is to test it for entanglement.
These tests are expressed in additive and multiplicative forms, for exactly
known CMs and for CMs with uncertainty. Using these tests, we analytically
establish the entanglement of one simple class of symmetric CMs of pure states
for arbitrary number of parts. These states are used for benchmarking of our
numerical procedures.

Though the representation of the symplectic trace as a minimum of some quantity
has already been known, neither its representation as a maximum of another
quantity, nor an explicit expression for it has been known before (to our
knowledge). These two results play an important role in the entanglement tests
we developed. 

Our tests are completely general in the sense that they can test for any
combinations of arbitrary partitions of the modes of the state. For example, in
the case of $n=5$ we can test not only for entanglement of concrete partitions
like $12|345$ or $1|24|35$, but also for arbitrary combinations of partitions
like $12|345, 1|24|35, 15|2|3|4$. We mostly work with genuine multipartite
entanglement, which involves all possible bipartitions, including their convex
combinations, and thus requires a more refined and less easily scalable approach
in terms of entanglement witnesses than just full inseparability. All
optimization problems are specialized for the practically relevant reduced case,
where the $xp$-part of CM is zero. We show that the specialized problems perform
much faster and take much less RAM then the problems in the full case, so we can
quickly solve the typical case without loosing the ability to solve the general
case when we really need to. 

Testing the multipartite states with $16$ parts (depending on whether it is the
full or reduced case) can be easily done on a high-end desktop PC. We use three
different solvers and show that for all of them the time to find a solution may
vary greatly depending on the state (with the same number of parts). It is a
feature of all iterative algorithms where the number of iterations (and thus the
time to converge to a solution) depends on some kind of condition on the input
data. We observe that different solvers work best for different kinds of
problem, so there is no one that is superior to others in all cases. In
addition, our figures show that more than one metric is needed to estimate the
quality of the produced solution. We are also able to test larger states, up to
$n=18$, but the accuracy of the results is not always satisfactory. We show that
the time needed to test for genuine multipartite entanglement grows
exponentially with $n$, while the RAM grows super-exponentially. As the software
for solving convex optimization problems will further advance, it will be
possible to solve larger problems on the given hardware (though this advancement
will not eliminate the exponential complexity of the entanglement tests). 

Besides purely numerical tests, in some cases the optimization problems we
presented here admit an analytical solution. We show how to analytically
construct the optimal solution of two kinds of the entanglement test
optimization problems for a few well-known example, also including a family of
non-Gaussian states, and for one representative, interesting parametrical family
of states. In addition, for this parametric family of states we show the regions
of parameters corresponding to various kinds of separability property and
demonstrate how these regions expand from the strongest property, the full
separability to the weakest property, the biseparability. All our examples are
either new or they are known from the literature, but were only numerically
treated before. We are confident that our ``modern'' treatment in this work will
find useful applications in Gaussian quantum optical technology and information
theory, and beyond. 

We acknowledge funding from the BMBF in Germany (QR.X, QuKuK, QuaPhySI,
PhotonQ), the EU’s Horizon Research and Innovation Actions (CLUSTEC), from the
EU/BMBF via QuantERA (ShoQC), and from the Deutsche Forschungsgemeinschaft (DFG,
German Research Foundation)–Project-ID 429529648 – TRR 306 QuCoLiMa (“Quantum
Cooperativity of Light and Matter”).

\appendix

\section{Proof of Theorem \ref{thrm:1}}

Introducing a new variable $s$ defined by Eq.~\eqref{eq:snorminf}, the objective
function to minimize is just $f_0(\vec{x}) = s$, where the variables of the
problem $\vec{x} = (s, \gamma)$ are $s$ and the elements of $\gamma$. There is
one semidefinite condition $\mathcal{P}(\gamma) \succcurlyeq 0$, which in real
form reads as
\begin{equation}
	-\begin{pmatrix}
		\gamma & \frac{1}{2}\Omega \\
		-\frac{1}{2}\Omega & \gamma
	\end{pmatrix}
	\preccurlyeq 0,
\end{equation}
and $2n(2n+1)$ linear conditions
\begin{equation}
	\pm\gamma \mp \gamma^\circ - s\sigma \leqslant 0.
\end{equation}
Thus, the dual problem has three symmetric matrix variables $(\bm{\lambda}_1,
\bm{\lambda}_2, \bm{\lambda}_3) = (Z, U^+, U^-)$, a $4n \times 4n$ matrix $Z
\succcurlyeq 0$ and two $2n \times 2n$ matrices $U^\pm \geqslant 0$. The
Lagrangian of this problem reads as follows:
\begin{displaymath}
\begin{split}
	L(s, &\gamma, Z, U^+, U^-) = s - \tr\left[Z\begin{pmatrix}
		\gamma & \frac{1}{2}\Omega \\
		-\frac{1}{2}\Omega & \gamma
	\end{pmatrix}\right] \\
	&+ \tr[U^+(\gamma - \gamma^\circ - s\sigma)] - \tr[U^-(\gamma - \gamma^\circ + s\sigma)].
\end{split}
\end{displaymath}
The objective of the dual problem is obtained from the Lagrangian by minimizing
it over the primal variables
\begin{equation}
	g(Z, U^+, U^-) = \inf_{s, \gamma} L(s, \gamma, Z, U^+, U^-),
\end{equation}
where the infimum is taken over all real $s$ and all symmetric matrices
$\gamma$. Let us decompose $Z$ as
\begin{equation}\label{eq:Zdec}
	Z = 
	\begin{pmatrix}
		Z_1 & Z_2 \\
		Z^{\mathrm{T}}_2 & Z_3
	\end{pmatrix},
\end{equation}
then the Lagrangian can be written as follows:
\begin{equation}
\begin{split}
	L &= \frac{1}{2}\tr[(Z_2 - Z^{\mathrm{T}}_2)\Omega] -\tr[(U^+ - U^-)\gamma^\circ] \\
	  &+(1-\tr[(U^+ + U^-)\sigma])s \\
	  &-\tr[(Z_1+Z_3 - U^+ + U^-)\gamma].
\end{split}
\end{equation}
A linear function is bounded from below if and only if it is identically equal
to zero, so the dual objective is
\begin{equation}
	g(Z, U^+, U^-) = \frac{1}{2}\tr[(Z_2 - Z^{\mathrm{T}}_2)\Omega] - \tr[(U^+ - U^-)\gamma^\circ]
\end{equation}
provided that the dual variables satisfy the conditions
\begin{equation}
\begin{split}
	\tr[(U^+ + U^-)\sigma] &= 1 \\ 
	U^+ - U^- &= Z_1 + Z_3,
\end{split}
\end{equation}
and $g(Z, U^+, U^-) = -\infty$ otherwise. Using the second condition we can write
the first, finite case as
\begin{equation}
	g(Z, U^+, U^-) = \frac{1}{2}\tr[(Z_2 - Z^{\mathrm{T}}_2)\Omega] - \tr[(Z_1 + Z_3)\gamma^\circ].
\end{equation}
We now introduce the complex Hermitian matrix
\begin{equation}
	\Lambda = \Lambda(Z) = Z_1 + Z_3 + i(Z_2 - Z^{\mathrm{T}}_2).
\end{equation}
It is easy to verify that
\begin{equation}
	\tr[\Lambda \mathcal{P}(\gamma^\circ)] = \tr[(Z_1 + Z_3)\gamma^\circ] - \frac{1}{2}\tr[(Z_2 - Z^{\mathrm{T}}_2)\Omega],
\end{equation}
so we derive that in the finite case the dual objective reads as
\begin{equation}
	g(Z, U^+, U^-) = \tr[\Lambda(Z) \mathcal{P}(\gamma^\circ)].
\end{equation}
The matrix $\Lambda(Z)$ is positive semidefinite since
\begin{equation}
\begin{split}
	\widetilde{\Lambda(Z)} &= 
	\begin{pmatrix}
		Z_1 + Z_3 & Z_2 - Z^{\mathrm{T}}_2 \\
		-Z_2 + Z^{\mathrm{T}}_2 & Z_1 + Z_3
	\end{pmatrix} \\
	&= 
	Z + 
	\begin{pmatrix}
		0 & E \\
		-E & 0
	\end{pmatrix}^{\mathrm{T}}
	Z
	\begin{pmatrix}
		0 & E \\
		-E & 0
	\end{pmatrix}
	\succcurlyeq 0.
\end{split}
\end{equation}
On the other hand, given a complex Hermitian positive semidefinite matrix 
\begin{equation}\label{eq:LambdaMK}
	\Lambda = M + i K,
\end{equation}
we can construct a real positive semidefinite matrix
\begin{equation}
	Z = \frac{1}{2}\tilde{\Lambda} = 
	\frac{1}{2}
	\begin{pmatrix}
		M & K \\
		-K & M
	\end{pmatrix}
\end{equation}
such that
\begin{equation}
	\Lambda(Z) = \Lambda.
\end{equation}
It means that when the variable $Z$ runs over all symmetric positive
semidefinite $4n\times4n$-matrices, the corresponding $\Lambda(Z)$ runs over all
Hermitian positive semidefinite $2n\times2n$-matrices and instead of $Z$ we can
use $\Lambda$. The dual problem now read as: Maximize 
\begin{equation}
	g(\Lambda, U^+, U^-) = \tr[\Lambda \mathcal{P}(\gamma^\circ)]
\end{equation} 
over all $\Lambda \succcurlyeq 0$ and $U^\pm \geqslant 0$ provided that
\begin{equation}\label{eq:UV}
\begin{split}
	\tr[(U^+ + U^-)\sigma] &= 1 \\
	U^+ - U^- &= \re(\Lambda).
\end{split}
\end{equation}
It is exactly the problem given in the main text.


In the reduced case the variables are $\vec{x} = (s, \gamma_{xx}, \gamma_{pp})$.
The semidefinite constrain of the primal problem is
\begin{equation}\label{eq:xpUV1}
	-\begin{pmatrix}
		\gamma_{xx} & \frac{1}{2}E \\
		\frac{1}{2}E & \gamma_{pp}
	\end{pmatrix}
	\preccurlyeq 0
\end{equation}
and the $2n(n+1)$ linear constraints are
\begin{equation}\label{eq:xpUV2}
\begin{split}
	\pm\gamma_{xx} \mp \gamma^\circ_{xx} - s\sigma_{xx} &\leqslant 0 \\
	\pm\gamma_{pp} \mp \gamma^\circ_{pp} - s\sigma_{pp} &\leqslant 0,
\end{split}
\end{equation}
so in this case there are five dual variables, a real positive semidefinite $2n
\times 2n$ matrix $Z \succcurlyeq 0$ and four $n \times n$ matrices $U^\pm, V^\pm \geqslant 0$
with non-negative elements. The Lagrangian reads as
\begin{displaymath}
\begin{split}
	&L(s, \gamma_{xx}, \gamma_{pp}, \Lambda, U^\pm, V^\pm)
	= s - \tr\left[Z\begin{pmatrix}
		\gamma_{xx} & \frac{1}{2}E \\
		\frac{1}{2}E & \gamma_{pp}
	\end{pmatrix}\right] \\
	&+ \tr[U^+(\gamma_{xx} - \gamma^\circ_{xx} - s\sigma_{xx})-U^-(\gamma_{xx} - \gamma^\circ_{xx} + s\sigma_{xx})] \\
	&+ \tr[V^-(\gamma_{pp} - \gamma^\circ_{pp} + s\sigma_{pp})-V^+(\gamma_{pp} - \gamma^\circ_{pp} - s\sigma_{pp})].
\end{split}
\end{displaymath}
The dual objective function is 
\begin{displaymath}
	g(\Lambda, U^\pm, V^\pm) = \min_{s, \gamma_{xx}, \gamma_{pp}} L(s, \gamma_{xx}, \gamma_{pp}, \Lambda, U^\pm, V^\pm),
\end{displaymath}
where the minimization is over all real $s$ and all symmetric $\gamma_{xx}$ and
$\gamma_{pp}$. For $g$ to be bounded the coefficient of $s$ must be zero, so we have
\begin{equation}
	\tr[(U^+ + U^-)\sigma_{xx} + (V^+ + V^-)\sigma_{pp}] = 1.
\end{equation}
Using the same decomposition for $\Lambda$ as \eqref{eq:Zdec} for $Z$ above, we
derive
\begin{displaymath}
	\tr\left[
	\begin{pmatrix}
		\Lambda_1 & \Lambda_2 \\
		\Lambda^{\mathrm{T}}_2 & \Lambda_3
	\end{pmatrix}
	\begin{pmatrix}
		\gamma_{xx} & \frac{1}{2}E \\
		\frac{1}{2}E & \gamma_{pp}
	\end{pmatrix}\right]
	= \tr[\Lambda_1 \gamma_{xx} + \Lambda_3 \gamma_{pp} + \Lambda_2].
\end{displaymath}
Equating the coefficients of $\gamma_{xx}$ and $\gamma_{pp}$ to zero, we obtain
\begin{equation}
\begin{split}
	\Lambda_{xx} &\equiv \Lambda_1 = U^+ - U^- \\
	\Lambda_{pp} &\equiv \Lambda_3 = V^+ - V^-.
\end{split}
\end{equation}
The dual objective is now readily seen to be given by
\begin{displaymath}
\begin{split}
	g(\Lambda, U^\pm, V^\pm) &= 
	-\tr[(U^+ - U^-)\gamma^\circ_{xx} + (V^+ - V^-)\gamma^\circ_{pp}] \\
	&- \tr(\Lambda_2) = 
	-\tr[\Lambda \mathcal{P}(\gamma^\circ_{xx}, \gamma^\circ_{pp})].
\end{split}
\end{displaymath}
The KKT conditions follow from the constrains \eqref{eq:xpUV1} and
\eqref{eq:xpUV2} of the primal problem.

The other theorems of the main text are derived in a similar way, by writing
the Lagrangian and simplifying it with matrix manipulations.

\end{document}